\documentclass[a4paper,UKenglish,hyperref,cleveref, autoref]{llncs}

\usepackage{graphicx,amssymb,amsmath,subcaption,algorithm2e,amsfonts,balance,enumerate,pgfplots,tikz,tikz-3dplot,url}

\usepackage[sort,nocompress]{cite}

\usetikzlibrary{shapes.geometric}
\usetikzlibrary{positioning}
\usetikzlibrary{calc}

\def\UIG{\mathrm{UIG}}
\def\APUD{\mathrm{APUD}}

\def\UU{\mathcal{U}}
\def\XX{\mathcal{X}}
\def\YY{\mathcal{Y}}
\def\WW{\mathcal{W}}

\newtheorem{open}[theorem]{Open problem}
\newtheorem{corl}[theorem]{Corollary}
\newtheorem{rem}[theorem]{Remark}
\newtheorem{lem}[theorem]{Lemma}
\newtheorem{clm}[theorem]{Claim}
\newtheorem{thm}[theorem]{Theorem}
\newtheorem{conj}[theorem]{Conjecture}

\makeatletter
\newcommand\footnoteref[1]{\protected@xdef\@thefnmark{\ref{#1}}\@footnotemark}
\makeatother

\def\centerarc[#1](#2)(#3:#4:#5)
{ \draw[#1] ($(#2)+({#5*cos(#3)},{#5*sin(#3)})$) arc (#3:#4:#5); }

\makeatletter
\def\grd@save@target#1{%
	\def\grd@target{#1}}
\def\grd@save@start#1{%
	\def\grd@start{#1}}
\tikzset{
	grid with coordinates/.style={
		to path={%
			\pgfextra{%
				\edef\grd@@target{(\tikztotarget)}%
				\tikz@scan@one@point\grd@save@target\grd@@target\relax
				\edef\grd@@start{(\tikztostart)}%
				\tikz@scan@one@point\grd@save@start\grd@@start\relax
				\draw[minor help lines] (\tikztostart) grid (\tikztotarget);
				\draw[major help lines] (\tikztostart) grid (\tikztotarget);
				\grd@start
				\pgfmathsetmacro{\grd@xa}{\the\pgf@x/1cm}
				\pgfmathsetmacro{\grd@ya}{\the\pgf@y/1cm}
				\grd@target
				\pgfmathsetmacro{\grd@xb}{\the\pgf@x/1cm}
				\pgfmathsetmacro{\grd@yb}{\the\pgf@y/1cm}
				\pgfmathsetmacro{\grd@xc}{\grd@xa + \pgfkeysvalueof{/tikz/grid with coordinates/major step}}
				\pgfmathsetmacro{\grd@yc}{\grd@ya + \pgfkeysvalueof{/tikz/grid with coordinates/major step}}
				\foreach \x in {\grd@xa,...,\grd@xb}
				\node[anchor=north] at (\x,\grd@ya) {\tiny \pgfmathprintnumber{\x}};
				\foreach \y in {\grd@ya,...,\grd@yb}
				\node[anchor=east] at (\grd@xa,\y) {\tiny \pgfmathprintnumber{\y}};
			}
		}
	},
	minor help lines/.style={
		help lines,
		step=\pgfkeysvalueof{/tikz/grid with coordinates/minor step}
	},
	major help lines/.style={
		help lines,
		line width=\pgfkeysvalueof{/tikz/grid with coordinates/major line width},
		step=\pgfkeysvalueof{/tikz/grid with coordinates/major step}
	},
	grid with coordinates/.cd,
	minor step/.initial=1,
	major step/.initial=1,
	major line width/.initial=0.7pt,
}

\newif\ifistoreview

\setboolean{istoreview}{true}

\newcommand{\setreviewsoff}{\istoreviewfalse}

\newcommand{\alertColor}{\textcolor{red}}

\newcommand{\addColor}{\textcolor{blue}}

\newcommand{\remove}[1]{\ifistoreview\alertColor{\sout{#1}}\else {}\fi}
\newcommand{\add}[1]{\ifistoreview\addColor{#1}\else #1\fi}
\newcommand{\replace}[2]{\ifistoreview\remove{#1}~\add{#2}\else #2\fi}

\setreviewsoff

\title{APUD(1,1) Recognition in Polynomial Time}

\author{Deniz A\u{g}ao\u{g}lu \c{C}a\u{g}{\i}r{\i}c{\i}\inst{1}\orcidID{0000-0002-1691-0434}%
	\thanks{This author was supported by the Czech Science Foundation, project no.~20-04567S.}
	\and \\
	Onur \c{C}a\u{g}{\i}r{\i}c{\i}\inst{2}\orcidID{0000-0002-4785-7496}}
\authorrunning{D.~A\u{g}ao\u{g}lu~\c{C}a\u{g}{\i}r{\i}c{\i} and O.~\c{C}a\u{g}{\i}r{\i}c{\i}}

\institute{
	Masaryk University, Brno, Czech Republic
	\email{agaoglu@mail.muni.cz}\\
	\and
	Toronto Metropolitan University, Toronto, Canada
	\email{cagirici@ryerson.ca}}

\begin{document}
	\maketitle

		\begin{abstract}
			
			A \emph{unit disk graph} is the intersection graph of a set of disk of unit radius in the Euclidean plane. In 1998, Breu and Kirkpatrick showed that the recognition problem for unit disk graphs is NP-hard. Given $k$ horizontal and $m$ vertical lines, an \emph{APUD(k,m)} is a unit disk graph such that each unit disk is centered either on a given horizontal or vertical line. \c{C}a\u{g}{\i}r{\i}c{\i} showed in 2020 that APUD($k,m$) recognition is NP-hard. In this paper, we show that APUD($1,1$) recognition is polynomial time solvable.
			
			
		\end{abstract}



	\section{Introduction}\label{sec:intro}
	
	\emph{Unit disk graphs} are the intersection graphs of a set of equal radius disks. Specifically, given a set $\mathcal{D} = \{ d_1, \dots, d_n\}$ of $n$ unit disks in the Euclidean plane, the corresponding unit disk graph $G=(V,E)$ has a vertex $v \in V(G)$ for each disk $d_v$, and there exists an edge $uv \in E(G)$ if and only if $d_u$ and $d_v$ intersect on the plane.
	In this paper, we study the recognition problem on the unit disk graph class \cite{udgRecognition}.
	The recognition problem for unit disk graph is a well-studied problem with various parameterizations and approximation algorithms \cite{udgOptimization,udgParameterized, udgApprox}.
	In general, the recognition problem is $\exists\mathbb{R}$-complete \cite{sphereAndDotProduct}.
	We study the unit disk graph recognition problem with restricted domain for the disk centers.
	Specifically, we limit the positions of disk centers onto pre-given straight lines in the Euclidean plane.
	The graphs those can be realized onto pre-given axes-parallel straight lines is called axes parallel unit disk graphs.
	This graph class is denoted by $\APUD(k,m)$ \cite{CSRapud}, where $k$ is the number of lines that are parallel to $x$-axis, and $m$ is the number of lines that are parallel to $y$-axis.
	Recently, \c{C}a\u{g}{\i}r{\i}c{\i} has shown that the problem becomes NP-complete when the solution domain for the disk centers are restricted to be on a set of pre-given parallel lines \cite{CSRapud}.
	They also left the following interesting problem open, which we consider in this paper.
	
	\begin{open}\label{op1}
		Can we decide whether an input graph $G$ is an $\APUD(1,1)$ in polynomial time?
	\end{open}
	
	This open problem essentially asks whether we can recognize a graph that can be realized as the intersection graph of unit disks such that the center of each disk is on one of two perpendicular lines in polynomial time.
	We answer this question positively and give a polynomial time algorithm to recognize an $\APUD(1,1)$.
	
	\section{Preliminaries}
	In this section, we give the necessary definitions and the notations that we use throughout the manuscript.
	An \emph{intersection graph} is a graph $G = (V,E)$ where each $u \in V(G)$ represents a geometric entity, and there exists an edge $uv \in E(G)$ iff the pair of geometric entities which correspond to $u, v \in V(G)$ intersect.
	
	The recognition problem on a geometric graph $G = (V,E)$ is to determine whether there exists a mapping $\Sigma: V(G) \to \mathfrak{U}$ such that all the intersection relations given in $E(G)$ are satisfied, where $\mathfrak{U}$ is the universe (number line, Euclidean plane, unit cube, etc.) in which the geometric entities lie.
		The mapping $\Sigma$ is called an \emph{embedding}, and an embedding of $G$ is denoted by $\Sigma(G)$ which is also referred as a \emph{representation}.
	
	One of the basic intersection graphs is \emph{interval graphs}, which represents the intersection of a set of intervals on the number line ($\mathfrak{U} = \mathbb{R}$).
	When all the intervals are of unit length, then the graph is called a \emph{unit interval graph}, and denoted by $\UIG$. It is known that unit interval graphs can be recognized in linear time~\cite{DBLP:journals/ipl/Keil85}.
	
	Unit interval graphs are a subclass of \emph{chordal graphs}. A \emph{chord} is an edge joining two nonconsecutive vertices of a cycle.
	A graph is called \emph{chordal} if it contains no chordless cycle of length more tan three. A chordal graph has linearly many maximal cliques which can be listed in linear time, and thus a chordal graph can be recognized in linear time \cite{recogChordaLinear}.
	
	A disk graph is the intersection graph of a set of disks in the Euclidean plane ($\mathfrak{U} = \mathbb{R} \times \mathbb{R}$). When a disk $A$ intersects another disk $B$, it also means that $B$ intersects $A$. We write ``$A$ and $B$ intersect'' since the intersection is a symmetric relation.
	In a unit disk graph, all disks have the same radius.
	The recognition problem is NP-hard on disk graphs \cite{dgRecognition} and also on unit disk graphs \cite{udgRecognition}. Unlike chordal graphs, unit disk graphs may have exponentially many maximal cliques \cite{expoMaxCliUDGpublished} which renders the method of listing all the maximal cliques and applying a greedy algorithm inefficient.
	In this paper, we focus on the recognition problem when the disk centers are restricted to be on pre-given axes-parallel straight lines only, then the corresponding graph is called an \emph{axes-parallel unit disk graph} ($\APUD$).
	An instance of $\APUD$ with $k$ horizontal and $m$ vertical lines is denoted by $\APUD(k,m)$.
	Note that, $\APUD(1,0) = \APUD(0,1) = \UIG$ simply because $\mathfrak{U} = \mathbb{R}$, and the fact that disks being two-dimensional does not have any effect on the intersection relations.
	
	A \emph{clique} in a graph $G$ is a subset $\mathcal{Q} \subseteq V(G)$ of vertices (analogously, the corresponding geometric entities) where each pair $u \neq v \in \mathcal{Q}$ of vertices are adjacent. A clique of size $n$ is denoted by $K_n$. A clique $\mathcal{Q}$ is called \emph{maximal} if it can not be extended to a larger clique $\mathcal{Q}' \supsetneq \mathcal{Q}$ by adding new vertices to $\mathcal{Q}$.
	
	Two sets $\UU$ and $\WW$ are called \emph{disjoint} if $\UU \cap \WW = \emptyset$. A \emph{partitioning} of a graph $G=(V,E)$ is to divide $V(G)$ into at least two disjoint sets. A \emph{complete bipartite graph} is a graph $G$ which admits a partitioning into two disjoint subsets $\UU,\WW \subseteq V(G)$ such that there exists an edge $uw \in E(G)$ iff $u \in \UU$ and $w \in \WW$.
	A complete bipartite graph is denoted by $K_{m,n}$ where $n$ is the cardinality of $\UU$ and $m$ is the cardinality of $\WW$. A clique on four vertices with one missing edge $e$ is called a \emph{diamond}, and denoted by $K_4 - e$. 
	
	For a set $\{v_i, \dots, v_j\}=\UU \subseteq V(G)$, the \emph{induced subgraph} of $G$ on $\UU$, denoted by $G[v_i \cup \dots \cup v_j]  = G[\UU]$, is the subgraph of $G$ which consists of all vertices in $\UU$ and all edges appearing in $G$ among the vertices in $\UU$. A \emph{connected component} of a graph $G$ is an induced subgraph of $G$ which is connected and can not be extended to a larger induced subgraph of $G$ by adding new vertices. Therefore, the connected components of $G$ are disjoint. Let $\UU$ and $\WW$ be two disjoint subsets of $G$. Then, the \emph{attachment} of $\UU$ on $\WW$ is the family of neighbors of every vertex $u \in \UU$ on $\WW$. 
	
	A cycle of length $k$ is denoted by $C_k$, and an induced $C_k$ is the chordless cycle of the same length. A \emph{wheel graph} on $k$ vertices, denoted by $W_k$, is a graph which consists of an induced $C_k$ and one universal vertex adjacent to all vertices of that cycle.

	Considering a graph $G \in \APUD(1,1)$, a vertex $v_i \in V(G)$ and an $\APUD(1,1)$ embedding $\Sigma(G)$, we denote the disk corresponding to $v_i$ in $\Sigma(G)$ by $I$ and the center of $I$ in $\Sigma(G)$ by $(x_i,y_i)$. Throughout the text, when we write a set $\mathcal{S}$ of disks is a clique (or any other graph-theoretical structure such as cycle, diamond, etc.), it means that the induced subgraph $G[\bigcup_{I \in \mathcal{S}} v_i]$ is a clique.
	
	We note here that the embedding of an $\APUD(1,1)$ may not be unique since a unit interval graph may have more than one representation. However, in the upcoming sections, it will be enough for us to consider any embedding to obtain the characterizations of $\APUD(1,1)$ to recognize them. This is also since we do not require a representation given on the input.
	
	Without loss of generality, we assume that the two perpendicular lines in an $\APUD(1,1)$ embedding are $x$- and $y$- axes of the Cartesian coordinate system ($\mathfrak{U} = (0, \mathbb{R}) \times (\mathbb{R}, 0)$), and we illustrate them with black dashed lines throughout the paper.
	We say that a disk $A$ has its center in $x^+$ if $A$ is centered on the ray which starts at the origin $(0,0)$ and passes through the point $(1,0)$, i.e., the positive side of the $x$-axis. Analogously, we say that $A$ has its center in $y^+$, $x^-$, and $y^-$ if $A$ is centered on the positive side of the $y$-axis, negative side of the $x$-axis, and negative side of the $y$-axis, respectively. We denote by $\XX^+$, $\YY^+$, $\XX^-$ and $\YY^-$ the sets of disks which have their centers on $x^+$, $y^+$, $x^-$ and $y^-$, respectively. The sets $\XX^+ \cup \XX^-$ and $\YY^+ \cup \YY^-$, i.e., the disks with their centers on $x$- and $y$-axes are denoted by $\XX$ and $\YY$, respectively. We note here that, for $G \in \APUD(1,1)$, $\XX^+$, $\XX^-$, $\YY^+$ and $\YY^-$ (thus, $\XX$ and $\YY$) is a partitioning of $V(G)$ with respect to any $\Sigma(G)$, and if a disk has its center on $(0,0)$, we assume that it belongs to exactly one of the sets $\XX^+$, $\YY^+$, $\XX^-$ and $\YY^-$.
	
	While proving our main claim, which says that whether a graph $G$ can be recognized as an $\APUD(1,1)$ in polynomial time, we use the geometric property called Helly property which is defined over cliques.
		A clique is said to have \emph{Helly property} if a set of entities form a clique, then they all have at least one common point. A clique which satisfies the Helly property in an intersection representation is called a \emph{Helly clique}, and otherwise, a \emph{non-Helly clique}.
	Every unit interval graph has a representation that satisfies the Helly property.

	\section{Some properties of $\boldsymbol{\APUD(1,1)}$} \label{sec:property}
	
	In this section, we give some simpler properties of an $\APUD(1,1)$. We first introduce the characterization given in \cite{CSRapud}.
	
	\begin{corl}[by combining Lemmas~5, ~7 and~8 in \cite{CSRapud}]\label{corl:combine}
		If $G=(V,E)$ is a connected $\APUD(1,1)$, then the following hold:
		\begin{enumerate}
			\item [\textbf{A1.}] The length of the largest induced cycle of $G$ is at most $4$.
			
			\item [\textbf{A2.}] $V(G)$ can be partitioned into four sets such that the union of any two of them induces a unit interval subgraph of $G$.
			
			\item [\textbf{A3.}] Given two $4$-cycles $(A,B,C,D)$ and $(U,V,W,X)$ both are counter-clockwise ordered sets in $\Sigma(G)$, each one of the sets $\{A,B,U,V\}$, $\{B,C,V,W\}$, $\{C,D,W,X\}$, and $\{D,A,X,U\}$ forms either a $K_4$ or an induced diamond.
		\end{enumerate}
	\end{corl}
	
	Considering this characterization, we first give the following two remarks, and then a sequence of statements which we use to recognize $\APUD(1,1)$ in polynomial time.

	\begin{rem}\label{rem:allInducedinPoly}
		By the characterization item  \textbf{A1}, every induced $C_4$ of an $\APUD(1,1)$ can be identified in polynomial time since there are polynomially many 4-tuples of vertices in the size of the input graph.
	\end{rem}

	\begin{lem}\label{lem:4cycle}
		If the set $\{A,B,C,D\}$ \add{of disks} in an $\APUD(1,1)$ forms an induced \replace{$4$-cycle}{$C_4$}, then \add{up to symmetry} $A \in \XX^+$, $B \in \YY^+$, $C \in \XX^-$, and $D \in \YY^-$ hold\remove{up to symmetry}. Moreover, $A$ and $C$ do not intersect, and, $B$ and $D$ do not intersect. 
	\end{lem}
	
	\begin{proof}
		By the characterization item \textbf{A2}, if $\{A,B,C,D\} \subseteq \XX^+$ up to symmetry, then $\{A,B,C,D\}$ does not form an induced \replace{$4$-cycle}{$C_4$} since $\XX^+$ induces a unit interval graph that can not contain an induced \replace{$4$-cycle}{$C_4$} due to chordality. Similarly, if $\{A,B,C,D\} \subseteq \XX^+ \cup \YY^+$ or $\{A,B,C,D\} \subseteq \XX^+ \cup \XX^-$ up to symmetry, then $\{A,B,C,D\}$ does not form an induced \replace{$4$-cycle}{$C_4$}. Therefore, every induced \replace{$4$-cycle}{$C_4$} of an $\APUD(1,1)$ contains disks belonging to at least three distinct sets from \add{the set family} $\{\XX^+$, $\YY^+$, $\XX^-$, $\YY^-\}$. 
		
		Up to symmetry, suppose that $\{A,C\} \subseteq \XX^+$, $B \in \YY^+$, and $D \in \YY^-$, the centers of $A$, $B$, $C$ and $D$ are at the coordinates $(x_a,0)$, $(0,y_b)$, $(x_c,0)$ and $(0,-y_d)$ for $x_a,y_b,x_c,y_d \in \mathbb{R}^+$, respectively, and $x_a \leq x_c$. Since $\{A,B,C,D\}$ forms an induced \replace{$4$-cycle}{$C_4$}, every disk $U \in \{A,B,C,D\}$ intersects two disks in $\{A,B,C,D\} \setminus U$, and if the pair $\{U,V\} \subsetneq \{A,B,C,D\}$ intersect, then the pair $\{A,B,C,D\} \setminus \{U,V\}$ intersect. 
		
		If $A$ and $C$ intersect, then $B$ and $D$ must intersect. Then, both $B$ and $D$ must intersect one more disk from $\{A,B,C,D\} \setminus \{B,D\}= \{A,C\}$, and since $x_a \leq x_c$, both $B$ and $D$ intersect $A$. Thus, $\{A,B,D\}$ forms an induced \replace{$3$-cycle}{$C_3=K_3$} which contradicts that $\{A,B,C,D\}$ forms an induced \replace{$4$-cycle}{$C_4$}. Otherwise, $A$ and $C$ do not intersect which means that $x_c - x_a > 2$ since $\{A,C\} \subseteq \XX^+$ and $x_a \leq x_c$. Then, ${x_c}^2 > 4$ which means that $C$ intersects neither $B$ nor $D$ as ${x_c}^2 + {y_b}^2> 4$ and ${x_c}^2 + {y_d}^2> 4$, and thus, $\{A,B,C,D\}$ does not form an induced \replace{$4$-cycle}{$C_4$} which is a contradiction. Hence, if $B \in \YY^+$ and $D \in \YY^-$, then $A \in \XX^+$ if and only if $C \in \XX^-$.\qed
	\end{proof}

	Unless stated otherwise, we assume that the centers of four disks $\{A,B,C,D\}$ of an $\APUD(1,1)$ forming a \replace{not necessarily induced}{(not necessarily induced)} \replace{$4$-cycle}{$C_4$,} \replace{(an induced $C_4$, an induced diamond or a $K_4$)}{i.e., an induced $C_4$, an induced diamond or a $K_4$,} are at the coordinates $(x_a,0)$, $(0,y_b)$, $(-x_c,0)$, and $(0,-y_d)$, respectively for $x_a,y_b,x_c,y_d \in \mathbb{R}^+$.
	
	\begin{lem}\label{lem:helCl}
		A non-Helly clique of an $\APUD(1,1)$ contains at least one disk centered on $x$-axis and at least one disk centered on $y$-axis.
	\end{lem}
	
	\begin{proof}
		Every clique of size $2$ is already a Helly clique. Then, it directly follows from that if all disks of a clique are centered on the $x$-axis (or analogously on the $y$-axis), then they form a unit interval graph and must intersect at a common point which is contained by the intersection of the two disks furthest from each other in that clique.\qed
	\end{proof}
	
	\begin{lem}\label{lem:noHellyOn00}
		If a set $\{A,B,C\}$ in an $\APUD(1,1)$ forms a non-Helly clique, then none of $A$, $B$ and $C$ has its center on $(0,0)$.
	\end{lem}
	\begin{proof}
		By Lemma~\ref{lem:helCl}, let $A \in \XX^+$ and $B \in \YY^+$ hold up to symmetry. Assume that $C$ has its center on  $(0,0)$. The following cases may occur:
		\begin{itemize}
			\item If $A$ and $C$ intersect at only one point, then the center of $A$ is at $(2,0)$. Then, $A$ intersects $B$ if and only if the center of $B$ is at $(0,0)$. However, now $\{A,B,C\}$ forms a Helly clique as they mutually intersect at $(0,0)$.			
			\item Else if $A$ and $B$ intersect at only one point, let the centers of $A$ and $B$ be on $(x_a,0)$ and $(0,y_b)$, respectively, where $x_a,y_b \in \mathbb{R}^+$. Then, the line segment $\ell$ between $(x_a,0)$ and $(0,y_b)$ has length $2$. This means that the line segment ${\ell}'$ between $(0,0)$ and the median point $(x_{ab},y_{ab})$ of $\ell$ has length $1$ since it is the median line for the right-angled triangle on the corners $(0,0)$, $(x_a,0)$ and $(0,y_b)$. However, now $\{A,B,C\}$ forms a Helly clique as they mutually intersect at $(x_{ab},y_{ab})$ since $C$ with radius $1$ contains $(x_{ab},y_{ab})$.
			
			\item Otherwise, both $A$ and $B$ are closer to $C$ than in the above two cases. Thus, $\{A,B,C\}$ forms a Helly clique as they mutually intersect at a common point.
		\end{itemize} 
		Therefore, none of $A$, $B$ and $C$ has its center on the point $(0,0)$ if $\{A,B,C\}$ forms a non-Helly clique.\qed
	\end{proof}

	\begin{corl}\label{corl:pointon00forNonHelly}
		If a set $\mathcal{S}$ on $d$ disks in an $\APUD(1,1)$ forms a non-Helly clique and a disk $A \in \mathcal{S}$ has its center on $(0,0)$, then $\mathcal{S} \setminus A$ forms a non-Helly clique.
	\end{corl}
	\begin{proof}
		It follows directly from Lemma~\ref{lem:noHellyOn00}.\qed
	\end{proof}
	
	\begin{lem}\label{lem:3clique}
		If a set \add{of disks} $\{A,B,C\}$ in an $\APUD(1,1)$ forms a non-Helly clique, then $A \in \mathcal{I}$, $B \in \mathcal{J}$, and $C \in \mathcal{K}$ where $\mathcal{I}, \mathcal{J}, \mathcal{K} \in  \{\XX^+, \YY^+, \XX^-, \YY^-\}$ and $\mathcal{I} \neq \mathcal{J} \neq \mathcal{K}$.
	\end{lem}
	
	\begin{proof}
		By Lemma~\ref{lem:helCl}, we know that all three disks cannot be on the same line, thus $\{A,B,C\} \not\subseteq \XX^+ \cup \XX^-$ and $\{A,B,C\} \not\subseteq \YY^+ \cup \YY^-$. By Lemma~\ref{lem:noHellyOn00}, none of $A$, $B$ and $C$ has its center on the point $(0,0)$.
		Let us assume for a contradiction that $A$ and $C$ belong to the same set from $\{\XX^+, \YY^+, \XX^-, \YY^-\}$. Up to symmetry, let $\{A,C\} \subseteq \XX^+$ hold such that $A$ is closer to the point $(0,0)$. Since $\{A,B,C\} \not\subseteq \XX^+ \cup \XX^-$, either $B \in \YY^+$ or $B \in \YY^-$ holds. Again up to symmetry, let $B \in \YY^+$ hold. Then, the intersection of $A$ and $C$ contains the intersection of $B$ and $C$ forming a Helly clique. Therefore, if $\{A,B,C\}$ does not form a Helly clique, then either $A \in \XX^-$ or $C \in \XX^-$ must hold.\qed
	\end{proof}
	
	Given an $\APUD(1,1)$ $G$, we give the following lemma which applies to all $\APUD(1,1)$ representations of $G$.
	We use this lemma to show that $G$ contains polynomially many maximal cliques when $G \in \APUD(1,1)$.
	
	\begin{lem}\label{lem:ONLYnonHelly}
		Every non-Helly clique in an $\APUD(1,1)$ contains a non-Helly clique on three disks.
	\end{lem}
	
	\begin{proof}
		By Helly theorem, if every three sets in a family of \replace{convex sets}{convex geometric object sets} in the Euclidean plane have a common intersection, then the whole family has a common intersection~\cite{HellyTheorem}. Since disks are convex, for every non-Helly clique on more than three disks in an $\APUD(1,1)$, there exists a non-Helly clique on three disks.\qed
	\end{proof}

		Here, we informally explain why we need Lemma~\ref{lem:ONLYnonHelly} before using it to prove Theorem~\ref{theo:linearlyMany}. Since every non-Helly clique in some $\APUD(1,1)$ embedding $\Sigma(G)$ of an $\APUD(1,1)$ $G$ contains a non-Helly clique on three disks, $G$ contains no clique $\mathcal{Q}$ of size $d > 3$ such that all $d-1$ tuples of vertices in $\mathcal{Q}$ intersect at a common point but not all $d$ of them in $\Sigma(G)$. This can equivalently be stated as that if an $\APUD(1,1)$ contains a non-Helly clique of size at least $4$, then it is non-Helly because of the non-Helly cliques of size $3$ it contains and thus, it has polynomially many non-Helly cliques since there are at most $\mathcal{O}(n^3)$ such triples for $G$ on $n$ vertices.

	For the upcoming claims, we use the following terminology.
	$G$ is a simple graph, and if $G$ is an $\APUD(1,1)$, $\Sigma(G)$ denotes some $\APUD(1,1)$ embedding of $G$. Let $\{A,B,C,D\}$ be four disks forming a \add{(not necessarily induced)} \replace{$4$-cycle}{$C_4$} \remove{such that the centers of $\{A,B,C,D\}$ are at the coordinates $(x_a,0)$, $(0,y_b)$, $(-x_c,0)$, and $(0,-y_d)$, respectively for $x_a,y_b,x_c,y_d \in \mathbb{R}^+$ }in $\Sigma(G)$.
	Let $\Gamma_{ABCD}$ denote the set of points $\mathcal{P}$ that is enclosed by the boundaries of $A$, $B$, $C$ and $D$, i.e. for every $p \in \mathcal{P}$, $\vert x_p \vert \leq x_a$, $\vert y_p \vert \leq y_b$, $\vert x_p \vert \leq x_c$, $\vert y_p \vert \leq y_d$, and no point in $\mathcal{P}$ is contained by the disks $A$, $B$, $C$ and $D$. If $\{A,B,C,D\}$ is an induced $C_4$, then $\Gamma_{ABCD}$ consists of one continuous region.
	If $\{A,B,C,D\}$ is an induced diamond, then $\Gamma_{ABCD}$ consists of at least one and at most two continuous regions.
	If $\{A,B,C,D\}$ is a $K_4$, then $\Gamma_{ABCD}$ is an empty set. Figure~\ref{fig:gammaExamples} shows $\Gamma_{ABCD}$ for those kinds of 4-cycles.
	
	\begin{figure}[h!]
		\centering
		\begin{subfigure}[t]{0.23\linewidth}
			\centering
			\begin{tikzpicture}[scale=0.4]
				
				\coordinate (E) at (0,0);
				
				\filldraw[color=white, fill=green, opacity=1] (E) circle (1);

				\begin{scope}
					\clip (-1.1,0) circle (1);
					\clip (0,0) circle (1);
					\filldraw[color=white, fill=white, opacity=1](0,0) circle (1);
				\end{scope}

				\begin{scope}
					\clip (1.1,0) circle (1);
					\clip (0,0) circle (1);
					\filldraw[color=white, fill=white, opacity=1](0,0) circle (1);
				\end{scope}
				
				\begin{scope}
					\clip (0,1.67) circle (1);
					\clip (0,0) circle (1);
					\filldraw[color=white, fill=white, opacity=1](0,0) circle (1);
				\end{scope}
				
				\begin{scope}
					\clip (0,-1.67) circle (1);
					\clip (0,0) circle (1);
					\filldraw[color=white, fill=white, opacity=1](0,0) circle (1);
				\end{scope}
				
				\draw (-3,-3) to[grid with coordinates] (3,3);
				
				
				\coordinate (A) at (-1.1,0);
				\coordinate (B) at (1.1,0);
				\coordinate (C) at (0,1.67);
				\coordinate (D) at (0,-1.67);
				
				\filldraw[color=red, fill=red, opacity=0.3] (A) circle (1);
				
				\filldraw[color=red, fill=red, opacity=0.3] (B) circle (1);
				
				\filldraw[color=red, fill=red, opacity=0.3] (C) circle (1);
				
				\filldraw[color=red, fill=red, opacity=0.3] (D) circle (1);
				
				
				\node (xa) at (0,-4.5) {(a)};
			\end{tikzpicture}
			\label{fig:nonHelly17}
		\end{subfigure}
		~
		\begin{subfigure}[t]{0.23\linewidth}
			\centering
			\begin{tikzpicture}[scale=0.4]
				
				\coordinate (E) at (0,0);
				
				\filldraw[color=white, fill=green, opacity=1] (E) circle (1);
				
				\begin{scope}
					\clip (-1,0) circle (1);
					\clip (0,0) circle (1);
					\filldraw[color=white, fill=white, opacity=1](0,0) circle (1);
				\end{scope}

				\begin{scope}
					\clip (1,0) circle (1);
					\clip (0,0) circle (1);
					\filldraw[color=white, fill=white, opacity=1](0,0) circle (1);
				\end{scope}
				
				\begin{scope}
					\clip (0,1.67) circle (1);
					\clip (0,0) circle (1);
					\filldraw[color=white, fill=white, opacity=1](0,0) circle (1);
				\end{scope}
				
				\begin{scope}
					\clip (0,-1) circle (1);
					\clip (0,0) circle (1);
					\filldraw[color=white, fill=white, opacity=1](0,0) circle (1);
				\end{scope}
				
				\draw (-3,-3) to[grid with coordinates] (3,3);
				
				
				\coordinate (A) at (-1,0);
				\coordinate (B) at (1,0);
				\coordinate (C) at (0,1.67);
				\coordinate (D) at (0,-1);
				
				\filldraw[color=red, fill=red, opacity=0.3] (A) circle (1);
				
				\filldraw[color=red, fill=red, opacity=0.3] (B) circle (1);
				
				\filldraw[color=red, fill=red, opacity=0.3] (C) circle (1);
				
				\filldraw[color=red, fill=red, opacity=0.3] (D) circle (1);
				
				\node (xb) at (0,-4.5) {(b)};
				
			\end{tikzpicture}
			\label{fig:nonHelly16}
		\end{subfigure}
		~
		\begin{subfigure}[t]{0.23\linewidth}
			\centering
			\begin{tikzpicture}[scale=0.4]
				
				\coordinate (E) at (0,0);
				
				\filldraw[color=white, fill=green, opacity=1] (E) circle (1);
				
				\begin{scope}
					\clip (-1,0) circle (1);
					\clip (0,0) circle (1);
					\filldraw[color=white, fill=white, opacity=1](0,0) circle (1);
				\end{scope}

				\begin{scope}
					\clip (1,0) circle (1);
					\clip (0,0) circle (1);
					\filldraw[color=white, fill=white, opacity=1](0,0) circle (1);
				\end{scope}
				
				\begin{scope}
					\clip (0,1.67) circle (1);
					\clip (0,0) circle (1);
					\filldraw[color=white, fill=white, opacity=1](0,0) circle (1);
				\end{scope}
				
				\begin{scope}
					\clip (0,-1.67) circle (1);
					\clip (0,0) circle (1);
					\filldraw[color=white, fill=white, opacity=1](0,0) circle (1);
				\end{scope}
				
				\draw (-3,-3) to[grid with coordinates] (3,3);
				
				
				\coordinate (A) at (-1,0);
				\coordinate (B) at (1,0);
				\coordinate (C) at (0,1.67);
				\coordinate (D) at (0,-1.67);
				
				\filldraw[color=red, fill=red, opacity=0.3] (A) circle (1);
				
				\filldraw[color=red, fill=red, opacity=0.3] (B) circle (1);
				
				\filldraw[color=red, fill=red, opacity=0.3] (C) circle (1);
				
				\filldraw[color=red, fill=red, opacity=0.3] (D) circle (1);
				
				\node (xc) at (0,-4.5) {(c)};
				
			\end{tikzpicture}
			\label{fig:nonHelly14}
		\end{subfigure}
		~
		\begin{subfigure}[t]{0.23\linewidth}
			\centering
			\begin{tikzpicture}[scale=0.4]
				
				\draw (-3,-3) to[grid with coordinates] (3,3);
				
				
				\coordinate (A) at (-1,0);
				\coordinate (B) at (1,0);
				\coordinate (C) at (0,1);
				\coordinate (D) at (0,-1);
				
				\filldraw[color=red, fill=red, opacity=0.3] (A) circle (1);
				
				\filldraw[color=red, fill=red, opacity=0.3] (B) circle (1);
				
				\filldraw[color=red, fill=red, opacity=0.3] (C) circle (1);
				
				\filldraw[color=red, fill=red, opacity=0.3] (D) circle (1);
				
				\node (xd) at (0,-4.5) {(d)};

			\end{tikzpicture}
			\label{fig:nonHelly15}
		\end{subfigure}
		\caption{$\Gamma_{ABCD}$ of (a) an induced $C_4$, (b) and (c) an induced diamond, and (d) a $K_4$ shown by green shading.}
		\label{fig:gammaExamples} 
	\end{figure}
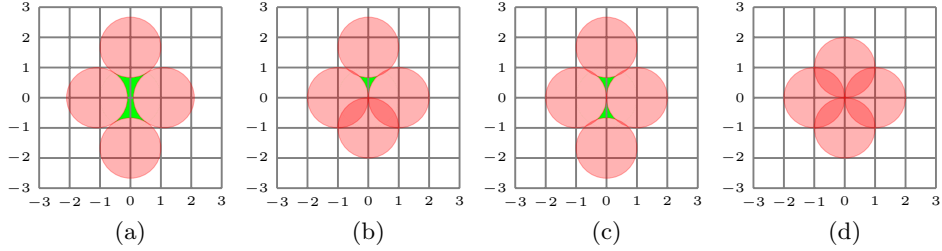

	\begin{lem}\label{lem:atLeastThree}
		If the set $\{A,B,C,D\}$ in an $\APUD(1,1)$ forms an induced $C_4$, 
		then $x_a,y_b,x_c,y_d \leq 2$ holds, and another disk $E$ whose 
		center is 
		contained in $\Gamma_{ABCD}$ intersects at least three disks from 
		$\{A,B,C,D\}$.
	\end{lem}

	\begin{proof}
		By Lemma~\ref{lem:4cycle}, $A$ does not intersect $C$, and $B$ does not intersect $D$. However, since $A$ intersects both $B$ and $D$, ${x_a}^2+{y_b}^2 \leq 4$, ${x_a}^2+{y_d}^2 \leq 4$, and $C$ also intersects both $B$ and $D$, ${x_c}^2+{y_b}^2 \leq 4$, ${x_c}^2+{y_d}^2 \leq 4$, and thus, $0 < x_a,y_b,x_c,y_d \leq 2$ since $x_a,y_b,y_d \in \mathbb{R}^+$. Suppose that $E$ is centered at the point $(x_e, 0)$, i.e. $E \in 
		\XX^+$, where $e \in \mathbb{R}^+$, up to symmetry.
		Then, $\sqrt{{x_e}^2 + {y_b}^2} \leq 2$ since $x_e \leq x_a$.
		Therefore, if the center of $E$ is in $\Gamma_{ABCD}$, then $E$ intersects $B$, and
		analogously $D$. Thus, $E$ intersects all disks from $\{ A,B,D \}$.\qed
	\end{proof}

	\begin{lem}\label{lem:middleClique}
		Let the set $\{A,B,C,D\}$ in an $\APUD(1,1)$ form an induced $C_4$ and $\mathcal{F}$ be the set of disks that are centered in $\Gamma_{ABCD}$. 
		Then, $\mathcal{F}$ is a Helly clique in $\Sigma(G)$.
	\end{lem}
	\begin{proof}
		One of the intersection points of $A,B,C$ and $D$ with the axes that their center points lie on are $(x_a-1,0)$, $(0,y_b-1)$, $(-(x_c-1),0)$, and $(0,-(y_d-1)$, respectively. Note that these four points are also on the boundary of $\Gamma_{ABCD}$.
		
		Consider the quadrilateral $A'B'C'D'$ that is formed by the points $A'(x_a-1,0)$, $B'(0,y_b-1)$, $C'(-(x_c-1),0)$, and $D'(0,-(y_d-1))$.
		Since we consider $\APUD(1,1)$, a pair $F_i, F_j \in \mathcal{F}$ of disks that are centered inside $A'B'C'D'$ are also centered on the diagonals of $A'B'C'D'$ (either on the line segment $[A'C']$ or on the line segment $[B'D']$). Consider the following statements.
		\begin{align*}
			{x_a}^2 + {y_b}^2 & \leq 4\\
			{y_b}^2 + {x_c}^2 & \leq 4\\
			{x_c}^2 + {y_d}^2 & \leq 4\\
			{x_a}^2 + {y_d}^2 & \leq 4	
		\end{align*} 
		
		Since $x_a,y_b,x_c,y_d \in \mathbb{R}^+$, the following also hold.
		\begin{align*}
			x_a,y_b,x_c,y_d &\leq 2\\
			\therefore & \\
			x_a-1,y_b-1,x_c-1,y_d-1 &\leq 1\\
			x_a-1+x_c-1 &= \vert [A'C'] \vert \\
			y_b-1+y_d-1 &= \vert [B'D']  \vert\\
			\therefore & \\
			\vert [A'C'] \vert &\leq 2\\
			\vert [B'D'] \vert &\leq 2\\
		\end{align*}
		
		Since the lengths of both diagonals $[A'C']$ and $[B'D']$ are at most $2$ units,
		every pair $F_i, F_j \in \mathcal{F}$ intersect since each disk is centered on $x$- or $y$-axis in $\Sigma(G)$, and $\mathcal{F}$ forms a clique.
		
		Now, assume that the set $\{F_i,F_j,F_k\} \subseteq \mathcal{F}$ forms a non-Helly clique in $\Sigma(G)$. By Lemma~\ref{lem:3clique}, $F_i$, $F_j$ and $F_k$ belong to distinct sets in $\{\XX^+,\YY^+,\XX^-,\YY^-\}$. Up to symmetry, let $F_i \in \XX^+$, $F_j \in \XX^-$, $F_k \in \YY^+$. Since $x_c-1 \leq 1$), $F_j$ intersects the common intersection of $F_i$ and $F_k$. Since no three disks in $\mathcal{F}$ form a non-Helly clique, the disks in $\mathcal{F}$ satisfy the Helly property by Lemma~\ref{lem:ONLYnonHelly}, and thus, the claim holds.\qed
	\end{proof}

	\begin{lem}\label{lem:someMandatoryEdges}
		Let $\{A_1,\dots,A_i\}$, $\{B_1,\dots,B_j\}$ and $\{C_1,\dots,C_k\}$ be three sets of disks of a connected $\APUD(1,1)$ which belong to $\XX^+$, $\YY^+$ and $\XX^-$, respectively, \replace{where :}{such that the following hold.}

		\begin{alignat*}{5}
			x_{a_1} &\leq x_{a_l} &< x_{a_m} &\leq x_{a_i}\ \text{ for }\ 1 &\leq l &< m &\leq i\\
			y_{b_1} &\leq y_{b_l} &< y_{b_m} &\leq y_{b_j}\ \text{ for }\ 1 &\leq l &< m &\leq j\\
			x_{c_1} &\leq x_{c_l} &< x_{c_m} &\leq x_{c_k}\ \text{ for }\ 1 &\leq l &< m &\leq k
		\end{alignat*}
		
		Then, the following \add{also} hold.
		
		\begin{enumerate}[(i)]
			\item If $A_i$ intersects $B_j$ or $C_k$, then $A_i$ intersects every $A_m$ for $m<i$.
			
			\item If $A_i$ intersects $B_j$, then $A_i$ intersects every $B_m$ for $m<j$ and  $B_j$ intersects every $A_m$ for $m<i$.
			
			\item If $A_i$ intersects $C_k$, then $A_i$ intersects every $C_m$ 
			with $m<k$ and $C_k$ intersects every $A_m$ with $m<i$. Moreover, 
			if $x_{c_k} \geq x_{a_i}$, $A_i$ intersects every $B_j$ that 
			intersects $C_k$, and if $x_{a_i} \geq x_{c_k}$, $C_k$ intersects 
			every $B_j$ that intersects $A_i$.

			\item Let $A_l$ be the disk with the maximum $x_{a_l}$ intersecting some disk from $\{C_1, \dots, C_k\}$, and $C_{l'}$ be such a disk with the maximum $x_{c_{l'}}$. Let $B_{l^*}$ be the disk among $\{B_1, \dots, B_j\}$ with the maximum $y_{b_{l^*}}$ that is intersected by both $A_l$ and $C_{l'}$. Then, the set $\{A_1$, $\dots$, $A_l$, $B_1$, $\dots$, $B_{l^*}$, $C_1$, $\dots$, $C_{l'}\}$ forms a clique, and the set $\{A_{l+1}$, $\dots$, $A_i$, $B_{l^*+1}$, $\dots$, $B_j$, $C_{l'+1}$, $\dots$, $C_k\}$ forms a disconnected unit interval graph on at least two and at most three connected components.

			\item Let $B_l$ be the disk with the maximum $y_{b_{l}}$ intersecting some disk from $\{A_1, \dots, A_i\}$ and some disk from $\{C_1, \dots, C_k\}$. Let $A_{l'}$ and  $C_{l^*}$ be the disks with the maximum $x_{a_{l'}}$ and $x_{c_{l^*}}$, respectively. If $A_{l'}$ and  $C_{l^*}$ intersect, then the set $\{A_1$, $\dots$, $A_{l'}$, $B_1$, $\dots$, $B_{l}$, $C_1$, $\dots$, $C_{l^*}\}$ forms a clique, and the set $\{A_{l'+1}$, $\dots$, $A_i$, $B_{l+1}$, $\dots$, $B_j$, $C_{l^*+1}$, $\dots$, $C_k\}$ forms a disconnected unit interval graph on at least two and at most four connected components. Otherwise, both sets $\{A_1$, $\dots$, $A_{l'}$, $B_1$, $\dots$, $B_{l}\}$ and $\{B_1$, $\dots$, $B_{l}$, $C_1$, $\dots$, $C_{l^*}\}$ form cliques, and each of the sets $\{A_{l'+1}$, $\dots$, $A_i$, $B_{l+1}$, $\dots$, $B_j\}$ and $\{B_{l+1}$, $\dots$, $B_j$, $C_{l^*+1}$, $\dots$, $C_k\}$ forms a disconnected unit interval graph on at least two and at most four connected components.
		\end{enumerate}	
	\end{lem}
	
	\begin{proof}
		\replace{We note here that the second item implies to such pairs $B_j$ and $C_k$.}{We prove the above items (i)-(v) one by one to show that the lemma holds.}
		
		\begin{enumerate}[(i)]
			\item If $A_i$ and $B_j$ intersect, then $\sqrt{{x_{a_i}}^2 + {y_{b_j}}^2} \leq 2$, and if $A_i$ and $C_k$ intersect, then $x_{a_i} + x_{c_k} \leq 2$, both implying that $x_{a_i} \leq 2$ and therefore, $x_{a_i} - x_{a_m} \leq 2$ for $0 \leq x_{a_m} \leq x_{a_i}$.
			
			\item Since $x_{a_i} = \max(x_{a_1}, \dots, x_{a_i})$, $y_{b_j} = \max(y_{b_1}, \dots, y_{b_j})$, and $\sqrt{{x_{a_i}}^2 + {y_{b_j}}^2} \leq 2$, $A_i$ intersects all $B_m$ with $m < j$ as $y_{b_m} < y_{b_j}$ and thus $\sqrt{{x_{a_i}}^2 + {y_{b_m}}^2} \leq 2$. Similarly, $B_j$ intersects all $A_m$ with $m < i$ since $x_{a_m} < x_{a_i}$ and thus $\sqrt{{x_{a_m}}^2 + {y_{b_j}}^2} \leq 2$. Note that this also applies to pairs such as $B_j$ and $C_k$ which are centered on $\XX^-$ and $\YY^+$.
			
			\item Since $x_{a_i} = \max(x_{a_1}, \dots, x_{a_i})$, $x_{c_k} = \max(x_{c_1}, \dots, x_{c_k})$, and $x_{a_i} + x_{c_k} \leq 2$, $A_i$ intersects all $C_m$ with $m < j$ as
			$x_{c_m} < x_{c_k}$ and thus $x_{a_i} + x_{c_m} \leq 2$. Similarly, $C_k$ intersects all $A_m$ with $m < i$ since $x_{a_m} < x_{a_i}$and thus $x_{a_m} + x_{c_k} \leq 2$.
			By triangle inequality, if $x_{c_k} 
			\geq x_{a_i}$ and $C_k$ intersects some $B_j$, then $2 \geq 
			\sqrt{{x_{c_k}}^2 + {y_{b_j}}^2} \geq \sqrt{{x_{a_i}}^2 + 
				{y_{b_j}}^2}$, thus $A_i$ intersects $B_j$ which also applies to 
			the case when $x_{a_i} \geq x_{c_k}$ and $A_i$ intersects some 
			$B_j$.
			
			\item Since $A_l$ intersects $C_{l'}$, $A_l$ intersects all $C_m$ with $1 \leq m \leq l'$ and $C_{l'}$ intersects all $A_m$ with $1 \leq m \leq l$ by item (iii). Since $B_{l^*}$ intersects both $A_l$ and $C_{l'}$, it intersects all $A_m$ with $1 \leq m \leq l$ and all $C_{m}$ with $1 \leq m \leq l'$ by item (ii). Moreover, $A_l$ intersects all $A_m$ with $m \leq l$, $B_{l^*}$ intersects all $B_m$ with $m \leq {l^*}$ and $C_{l'}$ intersects all $C_m$ with $m \leq l'$ by item (i). Therefore, $\{A_1$, $\dots$, $A_l$, $B_1$ $\dots$, $B_{l^*}$, $C_1$, $\dots$, $C_{l'}\}$ forms a clique. In addition, since $B_{l^*}$ has the maximum $b_{l^*}$ intersecting both $A_l$ and $C_{l'}$, it cannot intersect both $A_{l+1}$ and $C_{l'+1}$. Then, if $B_{l^*}$ intersects one of them, say $A_l$, $\{A_{l+1}$, $\dots$, $A_i$, $B_{l^*+1}$, $\dots$, $B_j$, $C_{l'+1}$, $\dots$, $C_k\}$ is a disconnected unit interval graph on two connected components, and otherwise, $\{A_{l+1}$, $\dots$, $A_i$, $B_{l^*+1}$, $\dots$, $B_j$, $C_{l'+1}$, $\dots$, $C_k\}$ is a disconnected unit interval graph on three connected components.
			
			\item It directly follows from the proofs of items (iii) and (iv).
		\end{enumerate}\qed
	\end{proof}
	
	\remove{INSTEAD OF REPEATING "Note that there might be more disks centered between the centers of $A$ and $A'$ (resp. $B$ and $B'$, $C$ and $C'$, $D$ and $D'$)." -> A' max, B' max, etc, thus this holds}
	
	The removal of a maximal clique of an interval graph may result in more than two connected components. On the other hand, for an $\APUD(1,0)$ which is a unit interval graph, we get the following.  
	
	\begin{corl}\label{corl:uigcomponents}
		Let $\mathcal{Q}$ be a maximal clique in a connected graph $G \in \APUD(1,0)$. Then, $G - \mathcal{Q}$ has at most two connected components each forming a unit interval graph.
	\end{corl}
	
	\begin{proof}
		The removal of a maximal clique of a unit interval graph results in at most two connected components as an interval graph contains no induced $K_{1,3}$ \cite{UIGK13free}. Moreover, the class of unit interval graphs is hereditary, i.e. any induced subgraph of a unit interval graph is also a unit interval graph \cite{UIGhereditary}. Since every $\APUD(1,0)$ is a unit interval graph, the corollary holds.\qed
	\end{proof}

	\begin{lem}\label{lem:components}
		Let $\mathcal{Q}$ be a maximal clique in a connected graph $G \in \APUD(1,1)$, and let $q$ denote the number of connected components in $G - \mathcal{Q}$. Then, the following hold.
		
		\begin{enumerate}[(i)]
			\item $1 \leq q \leq 4$.
			\item If $q = 4$ or $q = 3$, then every connected component of $G - \mathcal{Q}$ is a unit interval graph.
			\item If $q = 2$, then at least one connected component of $G - \mathcal{Q}$ is a unit interval graph.
		\end{enumerate}
	\end{lem}
	
	\begin{proof}
		
		Recall that a set $\mathcal{S} = \{I,\dots,J\}$ of disks of $\Sigma(G)$ is said to form a special graph class if the induced graph $G[\bigcup_{k=i}^{j} v_k]$ forms that special graph class. We consider the sets \remove{$\XX$, $\YY$,} $\XX^+$, $\YY^+$, $\XX^-$, and $\YY^-$\add{, such that  $\XX^+ \cup \XX^- = \XX$ and $\YY^+ \cup \YY^- = \YY$} in some embedding $\Sigma(G)$, and assume that none of $\XX^+$, $\YY^+$, $\XX^-$ and $\YY^-$\remove{, thus none of $\XX$ and $\YY$,} is empty \replace{for the sake}{without loss} of generality.
		
		\begin{enumerate}[(i)]
			\item By Corollary~\ref{corl:uigcomponents}, removing a maximal clique from a given $\APUD(1,0)$ partitions that graph into at most two connected components.
			Since there are two induced $\APUD(1,0)$ subgraphs $\XX$ and $\YY$ of $G$, removing $\mathcal{Q}$ partitions each of $\XX$ and $\YY$ into at most two connected components, thus $G-\mathcal{Q}$ into at most four connected components.
			
			\item If $q = 4$, then removing $\mathcal{Q}$ from $G$ partitions both $\XX$ and $\YY$, thus, these four connected components of $G -\mathcal{Q}$ are exactly the sets $\XX^+ \setminus \mathcal{Q}$, $\YY^+ \setminus \mathcal{Q}$, $\XX^- \setminus \mathcal{Q}$ and $\YY^- \setminus \mathcal{Q}$, which means that they all form unit interval graphs by Lemma 6 in \cite{CSRapud}. If $q=3$, then removing $\mathcal{Q}$ from $G$ partitions at least one of $\XX$ and $\YY$. Then, up to symmetry, these three connected components of $G -\mathcal{Q}$ are either $\XX \setminus \mathcal{Q}$, $\YY^+ \setminus \mathcal{Q}$ and $\YY^- \setminus \mathcal{Q}$, or $(\XX^ + \cup \YY^+) \setminus \mathcal{Q}$, $\XX^- \setminus \mathcal{Q}$ and $\YY^- \setminus \mathcal{Q}$. In this case, all three connected components form unit interval graphs again by Lemma 6 in \cite{CSRapud}.

			\item If $q = 2$, then up to symmetry, two connected components are $(\XX^+ \cup \YY^+) \setminus \mathcal{Q}$ and $(\XX^- \cup \YY^-)  \setminus \mathcal{Q}$, or $(\XX^+ \cup \YY^+ \cup \XX^-)  \setminus \mathcal{Q}$ and $\YY^-  \setminus \mathcal{Q}$, or $(\XX^+ \cup \YY^+ \cup \XX^- \cup \mathcal{Y'}^-)  \setminus \mathcal{Q}$ and $\mathcal{Y''}^-  \setminus \mathcal{Q}$ where $\mathcal{Y'}^- \cup \mathcal{Q} \cup \mathcal{Y''}^- = \YY^-$ (i.e., $\mathcal{Q} \subseteq \YY^-$ and the removal of $\mathcal{Q}$ only separates the disks in $\YY^-$). In the former case, again by Lemma 6 in \cite{CSRapud}, both connected components form unit interval graphs. In two latter cases, only the second connected component forms a unit interval graph.\qed	
		\end{enumerate}
	\end{proof}

	\begin{figure} [h!]
		\centering
		
		\begin{subfigure}[t]{0.18\linewidth}
			\centering
			
			\begin{tikzpicture}[xscale=-0.5,yscale=0.5]
				\coordinate (x+) at (-2.25,0);
				\coordinate (x-) at (2.25,0);
				\coordinate (y+) at (0,2.25);
				\coordinate (y-) at (0,-2.25);
				
				\coordinate (a) at (-1,0);
				\coordinate (b) at (0,1);
				\coordinate (c) at (1,0);
				\coordinate (d) at (0,-1);
				\coordinate (a') at (-2,0);
				\coordinate (b') at (0,2);
				\coordinate (c') at (2,0);
				\coordinate (d') at (0,-2);
				
				\draw [line width=5pt, gray, opacity=0.4] (-1.1,0) -- (-2,0)  node[midway, above]{};
				\draw [line width=5pt, gray, opacity=0.4] (1.1,0) -- (2,0)  node[midway, above]{};
				\draw [line width=5pt, gray, opacity=0.4] (0,-1.1) -- (0,-2)  node[midway, above]{};
				\draw [line width=5pt, gray, opacity=0.4] (0,1.1) -- (0,2)  node[midway, above]{};
				
				\draw [line width=5pt, pink, opacity=0.4] (-1.35,0) -- (0,1.35)  node[midway, above]{};
				\draw [line width=5pt, pink, opacity=0.4] (1.35,0) -- (0,-1.35)  node[midway, above]{};

				\draw [dashed,<->] (x-) -- (x+) node[midway, above]{};
				\node[] at (-2.4,0) {\tiny{$x$}};
				
				\draw [dashed,<->] (y-) -- (y+) node[midway, above]{};
				\node[] at (0,2.4) {\tiny{$y$}};

				\node[] at (-0.6,0.15) {\tiny{$A$}};
				\node[] at (0.2,1.15) {\tiny{$B$}};
				\node[] at (0.6,0.15) {\tiny{$C$}};
				\node[] at (0.3,-1) {\tiny{$D$}};
				
				\draw [blue,thick] (a) -- (b) node[midway, above]{};
				
				\draw [blue,thick] (a) -- (c) node[midway, above]{};
				
				\draw [blue,thick] (a) -- (d) node[midway, above]{};
				
				\draw [blue,thick] (b) -- (c) node[midway, above]{};
				
				\draw [blue,thick,dotted] (b) .. controls (-0.4,0.) .. (d) node[midway, above]{};
				
				\draw [blue,thick] (c) -- (d) node[midway, above]{};
				
				\draw [red,thick] (a) -- (a') node[midway, above]{};
				
				\draw [red,thick] (b) -- (b') node[midway, above]{};
				
				\draw [red,thick] (c) -- (c') node[midway, above]{};
				
				\draw [red,thick] (d) -- (d') node[midway, above]{};
				
				\node[fill,circle,scale=0.3] at (-1,0) {};
				\node[fill,circle,scale=0.3] at (0,1) {};
				\node[fill,circle,scale=0.3] at (1,0) {};
				\node[fill,circle,scale=0.3] at (0,-1) {};
				\node[fill,circle,scale=0.3] at (-2,0) {};
				\node[fill,circle,scale=0.3] at (0,2) {};
				\node[fill,circle,scale=0.3] at (2,0) {};
				\node[fill,circle,scale=0.3] at (0,-2) {};
				
				\node (xa) at (0,-3.5) {(a)};
				
			\end{tikzpicture}
			\label{fig:fourUINT}
		\end{subfigure}
		~
		\begin{subfigure}[t]{0.18\linewidth}
			\centering
			\begin{tikzpicture}[xscale=-0.5,yscale=0.5]
				
				\coordinate (x-) at (2.25,0);
				\coordinate (x+) at (-2.25,0);
				\coordinate (y-) at (0,-2.25);
				\coordinate (y+) at (0,2.25);
				
				\coordinate (a) at (-1,0);
				\coordinate (b) at (0,1);
				\coordinate (c) at (1,0);
				\coordinate (d) at (0,-1);
				\coordinate (a') at (-2,0);
				\coordinate (b') at (0,2);
				\coordinate (c') at (2,0);
				\coordinate (d') at (0,-2);
				
				\draw [line width=5pt, gray, opacity=0.4] (2,0) -- (-2,0)  node[midway, above]{};
				\draw [line width=5pt, gray, opacity=0.4] (0,-1.1) -- (0,-2)  node[midway, above]{};
				\draw [line width=5pt, gray, opacity=0.4] (0,1.1) -- (0,2)  node[midway, above]{};
				
				\draw [dashed,<->] (x-) -- (x+) node[midway, above]{};
				\node[] at (-2.4,0) {\tiny{$x$}};
				
				\draw [dashed,<->] (y-) -- (y+) node[midway, above]{};
				\node[] at (0,2.4) {\tiny{$y$}};
				
				\node[] at (-0.6,0.15) {\tiny{$A$}};
				\node[] at (0.2,1.15) {\tiny{$B$}};
				\node[] at (0.6,0.15) {\tiny{$C$}};
				\node[] at (0.3,-1) {\tiny{$D$}};
				
				\draw [blue,thick] (a) -- (b) node[midway, above]{};
				
				\draw [blue,thick] (a) -- (c) node[midway, above]{};
				
				\draw [blue,thick] (a) -- (d) node[midway, above]{};
				
				\draw [blue,thick] (b) -- (c) node[midway, above]{};
				
				\draw [blue,thick,dotted] (b) .. controls (-0.4,0.) .. (d) node[midway, above]{};
				
				\draw [blue,thick] (c) -- (d) node[midway, above]{};
				
				\draw [red,thick] (a) -- (a') node[midway, above]{};
				
				\draw [red,thick] (b) -- (b') node[midway, above]{};
				
				\draw [red,thick] (c) -- (c') node[midway, above]{};
				
				\draw [red,thick] (d) -- (d') node[midway, above]{};
				
				\node[fill,circle,scale=0.3] at (-1,0) {};
				\node[fill,circle,scale=0.3] at (0,1) {};
				\node[fill,circle,scale=0.3] at (1,0) {};
				\node[fill,circle,scale=0.3] at (0,-1) {};
				\node[fill,circle,scale=0.3] at (-2,0) {};
				\node[fill,circle,scale=0.3] at (0,2) {};
				\node[fill,circle,scale=0.3] at (2,0) {};
				\node[fill,circle,scale=0.3] at (0,-2) {};
				
				\node (xb) at (0,-3.5) {(b)};
				
			\end{tikzpicture}
			\label{fig:fourUINTx}
		\end{subfigure}
		~
		\begin{subfigure}[t]{0.18\linewidth}
			\centering
			\begin{tikzpicture}[xscale=-0.5,yscale=0.5]
				
				\draw [line width=5pt, gray, opacity=0.4] (-2,0) -- (-1.35,0) -- (0,1.35) -- (0,2) -- (0,1.35) --(1.35,0) -- (2,0) node[midway, above]{};
				
				\draw [line width=5pt, gray, opacity=0.4] (0,-1.1) -- (0,-2)  node[midway, above]{};
				
				\coordinate (x-) at (2.25,0);
				\coordinate (x+) at (-2.25,0);
				\coordinate (y-) at (0,-2.25);
				\coordinate (y+) at (0,2.25);
				
				\coordinate (a) at (-1,0);
				\coordinate (b) at (0,1);
				\coordinate (c) at (1,0);
				\coordinate (d) at (0,-1);
				\coordinate (a') at (-2,0);
				\coordinate (b') at (0,2);
				\coordinate (c') at (2,0);
				\coordinate (d') at (0,-2);

				\draw [dashed,<->] (x-) -- (x+) node[midway, above]{};
				\node[] at (-2.4,0) {\tiny{$x$}};
				
				\draw [dashed,<->] (y-) -- (y+) node[midway, above]{};
				\node[] at (0,2.4) {\tiny{$y$}};
				
				\node[] at (-0.6,0.15) {\tiny{$A$}};
				\node[] at (0.2,1.15) {\tiny{$B$}};
				\node[] at (0.6,0.15) {\tiny{$C$}};
				\node[] at (0.3,-1) {\tiny{$D$}};
				
				\draw [blue,thick] (a) -- (b) node[midway, above]{};
				
				\draw [blue,thick] (a) -- (c) node[midway, above]{};
				
				\draw [blue,thick] (a) -- (d) node[midway, above]{};
				
				\draw [blue,thick] (b) -- (c) node[midway, above]{};
				
				\draw [blue,thick,dotted] (b) .. controls (-0.4,0.) .. (d) node[midway, above]{};
				
				\draw [blue,thick] (c) -- (d) node[midway, above]{};
				
				\draw [red,thick] (a) -- (a') node[midway, above]{};
				
				\draw [red,thick] (b) -- (b') node[midway, above]{};
				
				\draw [red,thick] (c) -- (c') node[midway, above]{};
				
				\draw [red,thick] (d) -- (d') node[midway, above]{};
				
				\node[fill,circle,scale=0.3] at (-1,0) {};
				\node[fill,circle,scale=0.3] at (0,1) {};
				\node[fill,circle,scale=0.3] at (1,0) {};
				\node[fill,circle,scale=0.3] at (0,-1) {};
				\node[fill,circle,scale=0.3] at (-2,0) {};
				\node[fill,circle,scale=0.3] at (0,2) {};
				\node[fill,circle,scale=0.3] at (2,0) {};
				\node[fill,circle,scale=0.3] at (0,-2) {};
				
				\node (xc) at (0,-3.5) {(c)};
				
			\end{tikzpicture}
			\label{fig:oneUINTonenonUINTdiamond1}
		\end{subfigure}
		~
		\begin{subfigure}[t]{0.18\linewidth}
			\centering
			\begin{tikzpicture}[xscale=-0.5,yscale=0.5]
				
				\draw [line width=5pt, gray, opacity=0.4] (-1.6,0) -- (-1.35,0) -- (0,1.35) -- (0,2) -- (0,1.35) --(1.35,0) -- (2,0) -- (1.35,0) -- (0,-1.35) -- (0,-2) -- (0,-1.35) -- (-1.35,0) -- (-1.55,0) node[midway, above]{};
				
				\draw [line width=5pt, gray, opacity=0.4] (-1.7,0) -- (-2,0)  node[midway, above]{};
				
				\coordinate (x-) at (2.25,0);
				\coordinate (x+) at (-2.25,0);
				\coordinate (y-) at (0,-2.25);
				\coordinate (y+) at (0,2.25);
				
				\coordinate (a) at (-1,0);
				\coordinate (b) at (0,1);
				\coordinate (c) at (1,0);
				\coordinate (d) at (0,-1);
				\coordinate (a') at (-2,0);
				\coordinate (b') at (0,2);
				\coordinate (c') at (2,0);
				\coordinate (d') at (0,-2);

				\draw [dashed,<->] (x-) -- (x+) node[midway, above]{};
				\node[] at (-2.4,0) {\tiny{$x$}};
				
				\draw [dashed,<->] (y-) -- (y+) node[midway, above]{};
				\node[] at (0,2.4) {\tiny{$y$}};
				
				\node[] at (-0.6,0.15) {\tiny{$A$}};
				\node[] at (0.2,1.15) {\tiny{$B$}};
				\node[] at (0.6,0.15) {\tiny{$C$}};
				\node[] at (0.3,-1) {\tiny{$D$}};
				
				\draw [blue,thick] (a) -- (b) node[midway, above]{};
				
				\draw [blue,thick] (a) -- (c) node[midway, above]{};
				
				\draw [blue,thick] (a) -- (d) node[midway, above]{};
				
				\draw [blue,thick] (b) -- (c) node[midway, above]{};
				
				\draw [blue,thick,dotted] (b) .. controls (-0.4,0.) .. (d) node[midway, above]{};
				
				\draw [blue,thick] (c) -- (d) node[midway, above]{};
				
				\draw [red,thick] (a) -- (a') node[midway, above]{};
				
				\draw [red,thick] (b) -- (b') node[midway, above]{};
				
				\draw [red,thick] (c) -- (c') node[midway, above]{};
				
				\draw [red,thick] (d) -- (d') node[midway, above]{};
				
				\node[fill,circle,scale=0.3] at (-1,0) {};
				\node[fill,circle,scale=0.3] at (0,1) {};
				\node[fill,circle,scale=0.3] at (1,0) {};
				\node[fill,circle,scale=0.3] at (0,-1) {};
				\node[fill,circle,scale=0.3] at (-2,0) {};
				\node[fill,circle,scale=0.3] at (0,2) {};
				\node[fill,circle,scale=0.3] at (2,0) {};
				\node[fill,circle,scale=0.3] at (0,-2) {};
				
				\node (xd) at (0,-3.5) {(d)};
				
			\end{tikzpicture}
			\label{fig:oneUINTonenonUINTdiamond2}
		\end{subfigure}
		~
		\begin{subfigure}[t]{0.18\linewidth}
			\centering
			\begin{tikzpicture}[xscale=-0.5,yscale=0.5]
				
				\draw [line width=5pt, gray, opacity=0.4] (-2,0) -- (-1.35,0) -- (0,1.35) -- (0,2) -- (0,1.35) --(1.35,0) -- (2,0) -- (1.35,0) -- (0,-1.35) -- (0,-2) node[midway, above]{};
				
				\draw [line width=5pt, pink, opacity=0.4] (-1.35,0) -- (0,-1.35) node[midway, above]{};
				
				\coordinate (x-) at (2.25,0);
				\coordinate (x+) at (-2.25,0);
				\coordinate (y-) at (0,-2.25);
				\coordinate (y+) at (0,2.25);
				
				\coordinate (a) at (-1,0);
				\coordinate (b) at (0,1);
				\coordinate (c) at (1,0);
				\coordinate (d) at (0,-1);
				\coordinate (a') at (-2,0);
				\coordinate (b') at (0,2);
				\coordinate (c') at (2,0);
				\coordinate (d') at (0,-2);

				\draw [dashed,<->] (x-) -- (x+) node[midway, above]{};
				\node[] at (-2.4,0) {\tiny{$x$}};
				
				\draw [dashed,<->] (y-) -- (y+) node[midway, above]{};
				\node[] at (0,2.4) {\tiny{$y$}};
				
				\node[] at (-0.6,0.15) {\tiny{$A$}};
				\node[] at (0.2,1.15) {\tiny{$B$}};
				\node[] at (0.6,0.15) {\tiny{$C$}};
				\node[] at (0.3,-1) {\tiny{$D$}};
				
				\draw [blue,thick] (a) -- (b) node[midway, above]{};
				
				\draw [blue,thick] (a) -- (c) node[midway, above]{};
				
				\draw [blue,thick] (a) -- (d) node[midway, above]{};
				
				\draw [blue,thick] (b) -- (c) node[midway, above]{};
				
				\draw [blue,thick,dotted] (b) .. controls (-0.4,0.) .. (d) node[midway, above]{};
				
				\draw [blue,thick] (c) -- (d) node[midway, above]{};
				
				\draw [red,thick] (a) -- (a') node[midway, above]{};
				
				\draw [red,thick] (b) -- (b') node[midway, above]{};
				
				\draw [red,thick] (c) -- (c') node[midway, above]{};
				
				\draw [red,thick] (d) -- (d') node[midway, above]{};
				
				\node[fill,circle,scale=0.3] at (-1,0) {};
				\node[fill,circle,scale=0.3] at (0,1) {};
				\node[fill,circle,scale=0.3] at (1,0) {};
				\node[fill,circle,scale=0.3] at (0,-1) {};
				\node[fill,circle,scale=0.3] at (-2,0) {};
				\node[fill,circle,scale=0.3] at (0,2) {};
				\node[fill,circle,scale=0.3] at (2,0) {};
				\node[fill,circle,scale=0.3] at (0,-2) {};
				
				\node (xe) at (0,-3.5) {(e)};
				
			\end{tikzpicture}
			\label{fig:onenonUINTpath}
		\end{subfigure}

		\caption{The possible cases in the setting of Lemma~\ref{lem:components}. }
		\label{fig:chordalS4COPONENTS}
	\end{figure}
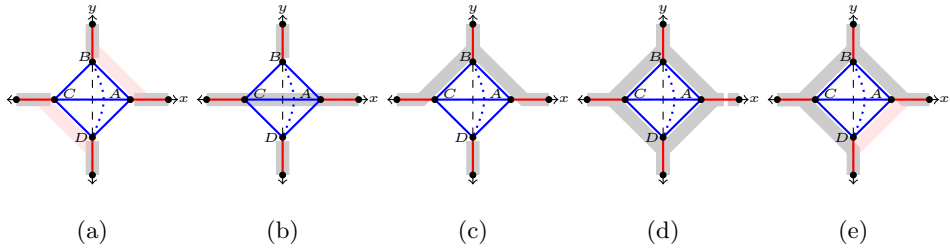

	In Figure~\ref{fig:chordalS4COPONENTS} we demonstrate the possible cases investigated in the proof of Lemma~\ref{lem:components} where the \add{thick} gray \replace{edges}{lines} denote existing \replace{connections}{intersections} and the \add{thick} pink \replace{edges}{lines} denote the possible additional \replace{connections}{intersections}. \add{The diagram in} (a) without the pink edges shows when $G \setminus \mathcal{Q}$ has four components which are all unit interval graphs. \add{If we exclude one pink line from the diagram shown in} (a) \replace{with one pink edge, then (a)} and (b) \add{together} show the cases when $G \setminus \mathcal{Q}$ has three components which are all unit interval graphs. \add{If we include both pink lines from the diagram shown in} (a), \replace{with two pink edges,}{then (a),} (c) and (d) show the cases when $G \setminus \mathcal{Q}$ has two components which are all unit interval graphs. (e) show the case when $G \setminus \mathcal{Q}$ has one component that is not a unit interval graph.

	Suppose that the set $\mathcal{L}=\{A,B,C,D\}$ is a 4-cycle in $\Sigma(G)$ which forms an induced \replace{cycle}{$C_4$}, an induced diamond or a $K_4$ such that the centers of $A,B,C$ and $D$ are at $(x_a,0)$, $(0,y_b)$, $(-x_c,0)$ and $(0,-y_d)$, respectively where $x_a,y_b,x_c,y_d \in \mathbb{R}^+$. We call $\mathcal{L}$ \emph{the minimum 4-cycle} if $x_a$, $y_b$, $x_c$ and $y_d$ are the minimum coordinates among the centers of the disks belonging to $\XX^+$, $\YY^+$, $\XX^-$ and $\YY^-$, respectively in $\Sigma(G)$.
	
	We also call $\mathcal{L}$ a \emph{minimum induced $C_4$} (analogously, a \emph{minimum induced diamond} and a \emph{minimum $K_4$}) if $(x_a+x_c) \cdot (y_b+y_d)$ is the minimum area among the areas of all induced $C_4$s (analogously, induced diamonds and $K_4$s) such that both $x_a$ and $x_c$ or both $y_b$ and $y_d$ are the minimum possible coordinates among their respective sets from $\{\XX^+,\YY^+, \XX^-,\YY^-\}$. Unless stated otherwise, we always assume that the centers of four such disks are at the coordinates $(x_a,0)$, $(0,y_b)$, $(-x_c,0)$, and $(0,-y_d)$, respectively for $x_a,y_b,x_c,y_d \in \mathbb{R}^+$. Note that the minimum 4-cycle of an $\APUD(1,1)$ may not exists and this means that there exists no minimum $K_4$, no minimum induced diamond and no minimum induced $C_4$.

	\begin{lem}\label{lem:numberOfMinCycles}
		Let $G$ be an $\APUD(1,1)$ and $\Sigma(G)$ be one of its embeddings. Then, the following hold.
		\begin{enumerate}[(i)]
			\item There is at most one minimum 4-cycle in $\Sigma(G)$.	
			\item There is at most one minimum $K_4$ in $\Sigma(G)$.
			\item There are at most four minimum induced diamonds in $\Sigma(G)$.
			\item There are at most eight minimum induced $C_4$s in $\Sigma(G)$.
		\end{enumerate}
	\end{lem}
	
	\begin{proof}
		We prove each of these items separately.
		\begin{enumerate}[(i)] 
			\item Suppose that the minimum cycle is on the unique disks $A$, $B$, $C$ and $D$. Then, the coordinates of the centers, $x_a$, $y_b$, $x_c$ and $y_d$ are the smallest values on $x^+$, $y^+$, $x^-$, and $y^-$, respectively. Then by definition, there exists at most one minimum 4-cycle.
			\item The existence of more than one minimum $K_4$ is a contradiction since all vertices of all $K_4$s pairwise intersect as among a set of $K_4$s, one obtains the unique minimum $K_4$ (if it exists) on the four disks with minimum coordinates $x_a,y_b,x_c$ and $y_d$.
			\item If $A$ and $C$ intersect, then with respect to the minimum $x_a$ and $x_c$, there are at most two possible minimum coordinates for $y_b$ and $y_d$ one having the smallest possible $y_b$ and the other having the smallest possible $y_d$ such that $B$ and $D$ do not intersect. Similarly, if $B$ and $D$ intersect, there are at most two possible minimum coordinates for $x_a$ and $x_c$ one having the smallest possible $x_a$ and the other having the smallest possible $x_c$ such that $A$ and $C$ do not intersect. Therefore, there are at most four minimum induced diamonds. Moreover, if the minimum 4-cycle of $G$ is already a minimum induced diamond, then it is unique.
			\item Neither $A$ and $C$ nor $B$ and $D$ intersect. For $A$ with the minimum $x_a$ (analogously, for $C$ with the minimum $x_c$), there are two possible minimum coordinates for $y_b$ and $y_d$ one having the smallest possible $y_b$ and the other having the smallest possible $y_d$, and they fix the smallest possible coordinate for $x_c$ (analogously, for $x_a$) by the definition of a minimum cycle. Similarly, for $B$ with the minimum $y_b$ (analogously, for $D$ with the minimum $y_d$), there are two possible minimum coordinates for $x_a$ and $x_c$ one having the smallest possible $x_a$ and the other having the smallest possible $x_c$, and they fix the smallest possible coordinate for $y_d$ (analogously, for $y_b$) by the definition of a minimum cycle. Therefore, there are at most eight minimum induced cycles. Moreover, if the minimum 4-cycle of $G$ is already a minimum induced $C_4$, then it is unique.
		\end{enumerate}
		Therefore, the lemma holds.\qed
	\end{proof}

	\replace{In the rest of the paper}{Henceforth}, whenever we mention a minimum induced diamond or a minimum induced $C_4$, we consider any minimum 4-cycle of such a kind since they provide the same arguments in the proofs. Specifically, we spoil the obtained characterizations for some representation to recognize whether an input graph is $\APUD(1,1)$.

	\remove{check commented claim above!!!}
	
	\remove{mention that we assume no twins since they can be placed on the same center. If we dont say this, min cycle creates problems!!!}

	\remove{IN ALL OPTION-BASED LEMMAS, WHENEVER WE SAY THERE EXISTS THREE DISKS A' IN MATHCALX+, CHANGED ALL TO  A' NEQ A IN MATHCALX+}
	
	\section{Recognizing a chordal $\boldsymbol{\APUD(1,1)}$ in polynomial time}
	\label{sec:chordalpoly}
	
	In this section, we show how to recognize whether a chordal input graph $G$ is an $\APUD(1,1)$.

		\begin{figure} [t]
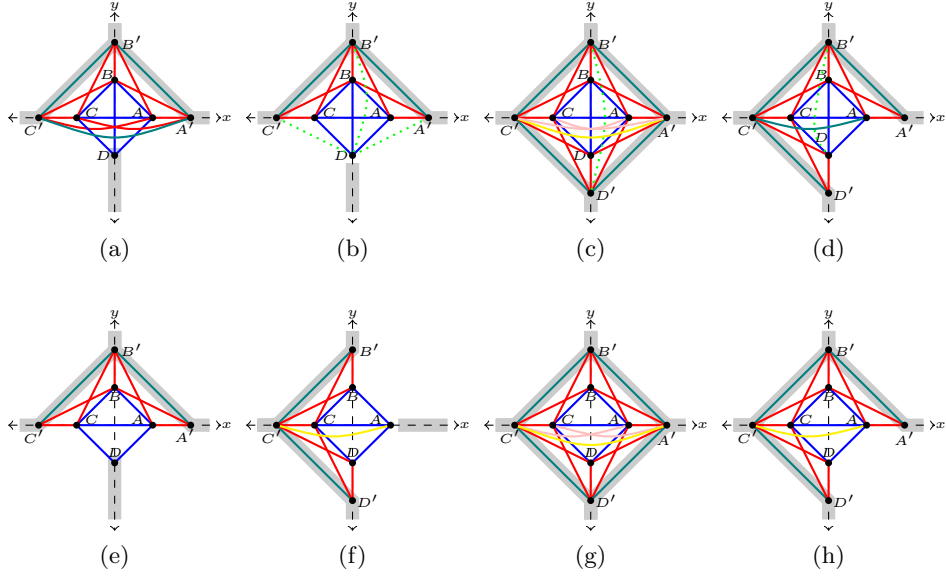

			\centering
			\begin{subfigure}[t]{0.23\linewidth}
				\centering

				\label{fig:case4Diamond}
			\end{subfigure}

			\caption{The possible cases in the setting of Lemma~\ref{lem:CHOR4into4unitINTS}.}
			\label{fig:chordalS4}
		\end{figure}

		\begin{lem}\label{lem:CHOR4into4unitINTS}
			If a given connected graph $G \in \APUD(1,1)$ is chordal and the minimum 4-cycle exists in $\Sigma(G)$, then there exists a maximal clique $\mathcal{Q}$ of $G$ such that $G - \mathcal{Q}$ has $q$ connected components each forming a unit interval graph.
		\end{lem}
		
		\begin{proof}
		Consider the minimum 4-cycle $\mathcal{L}=\{A,B,C,D\}$ in $\Sigma(G)$. Since $G$ is chordal, it contains no induced \replace{cycles of length $>3$}{$C_k$ for $k > 3$}. We, therefore, consider the following two cases:
		\begin{enumerate}[(i)]
			\item $\mathcal{L}$ forms a $K_4$.
			\item $\mathcal{L}$ forms an induced diamond.
		\end{enumerate}
		
		We know that $q \leq 4$, by Lemma~\ref{lem:components}.
		If all $q$ components of $G-\mathcal{Q}$ are unit interval graphs, then the lemma holds.
		We also know that if $q = 3$ or $q = 4$, then the lemma holds, again by Lemma~\ref{lem:components}.
		Therefore, we only study the following subcases:
		\begin{enumerate}[a)]
			\item $q=2$ and there exists exactly one component that is unit interval graph.
			\item $q=1$ and that component is not a unit interval graph.
		\end{enumerate}
		
		We show that if one of these cases occur, then there exists a maximal clique whose removal separates $G$ into disjoint connected components, such that all of those components are unit interval graphs. In the upcoming part of the proof, whenever we write ``$\mathcal{S}=\{A,B,C,D\}$ forms a clique'', we mean the clique $\mathcal{S'}$ consisting of  $\{A,B,C,D\}$ and the other disks that intersect all in $\{A,B,C,D\}$ as well as each other, i.e.  $\mathcal{S'} \supseteq \{\mathcal{A}',\mathcal{B}',\mathcal{C}',\mathcal{D}'\}$ where for every $A' \in \mathcal{A}' \subseteq \XX^+$, $B' \in \mathcal{B}' \subseteq \YY^+$, $C' \in \mathcal{C}' \subseteq \XX^-$, and $D' \in \mathcal{D}' \subseteq \YY^-$, $x_{a'} \leq x_a$, $y_{b'} \leq y_b$, $-x_{c'} \geq -x_c$, and $-y_{d'} \geq x_d$, respectively.

		\begin{enumerate}[(i)]
			\item $\mathcal{L}$ forms a $K_4$.
			If $\mathcal{L}$ is not a maximal clique, then it can be extended to some maximal clique $\mathcal{Q} \supsetneq \mathcal{L}$ by checking linearly many maximal cliques in $G$. Now, let us study the cases a) and b) that are mentioned above.

			\begin{enumerate}[a)]
				\item $G - \mathcal{Q}$ has two components where exactly one of them is not a unit interval graph.
				Let these components be $\Delta_1 = (\XX^+ \cup \YY^+ \cup \XX^-) \setminus \mathcal{Q}$, and $\Delta_2 = \YY^- \setminus \mathcal{Q}$, up to symmetry.
				Then, there exist at least three disks, say $A' \in \XX^+$, $B' \in \YY^+$, and $C' \in \XX^-$ such that $\{A',B',C'\} \subseteq \Delta_1$, $A'$ intersects $B'$, and $B'$ intersects $C'$.
				Moreover, these disks $A'$,$B'$ and $C'$ are such that the Euclidean distance between the centers of the pairs $(A', B')$ and $(B',C')$ are maximum.
				Note that there might be more disks centered between the centers of $A$ and $A'$ (resp. $B$ and $B'$, $C$ and $C'$).
				By triangle inequality, $A'$ intersects all in $\{A,B,B'\}$, $B'$ intersects all in $\{A,B,C,A',C'\}$, and $C'$ intersects \replace{all in $\{B,C,C'\}$}{both $B$ and $C$}.
				Note that none of the disks in $\Delta_1$ intersects a disk in $\Delta_2$ since $\Delta_1$ and $\Delta_2$ are disjoint. Now, we study the following two subcases.
				\begin{itemize}
					\item \textbf{If $\boldsymbol{A'}$ and $\boldsymbol{C'}$ intersect,} then \replace{the pairs $(A,C)$ and $(A', C')$ intersect}{$A$ intersects $C'$ and $A'$ intersects $C$} by the triangle inequality.
					In this case, $\mathcal{S} = \{A,B,C,A',B',C'\}$ is a clique.
					Since the Euclidean distance between center points of the pairs $(A', B')$ and $(B',C')$ are maximum, there exists no disks that intersect $B'$ that are farther to origin compared to $A'$ in $\XX^+$ and $C'$ in $\XX^-$.
					Thus, removing the prescribed maximal clique $\mathcal{S'} \supseteq \mathcal{S}$, the connected component $\Delta_1$ is separated into three disjoint connected components, all of which are unit interval graphs.
					Thus, there are four components all of which are unit interval graphs, including $\Delta_2$.
					
					
					\item \textbf{Otherwise, $\boldsymbol{A'}$ and $\boldsymbol{C'}$ do not intersect.} In this case, if $B'$ intersects $D$, then $\mathcal{S}=\{A,B,C,D,B'\}$ is a clique.
					Considering that $\Delta_1$ and $\Delta_2$ are disjoint, none of $A'$, $B'$ and $C'$ intersects some disk in $\Delta_2$.
					Since the Euclidean distance between center points of the pairs $(A', B')$ and $(B',C')$ are maximum, neither $A'$ nor $C'$ intersects some disk in $\YY^+ \cap \Delta_1$.
					Thus, removing $\{A,B,C,D,B'\}$ separates $G$ into four disjoint connected components, all of which are unit interval graphs. Otherwise, $B'$ does not intersect $D$ and removing the prescribed maximal clique $\mathcal{S'} \supseteq \mathcal{S}=\{A,B,C,B'\}$  separates $G$ into four disjoint connected components, all of which are unit interval graphs. 
				\end{itemize}
				
				\remove{check first intersected sets for A' and D'}
				\item $G - \mathcal{Q}$ has exactly one component $\Delta = (\XX^+ \cup \YY^+ \cup \XX^- \cup \YY^-) \setminus \mathcal{Q}$ which is not a unit interval graph.
				There must exist at least fours disks, say $A' \in \XX^+$, $B' \in \YY^+$, $C' \in \XX^-$, and $D' \in \YY^-$ such that $\{A',B',C',D'\} \subseteq \Delta$, and $A'$ intersects all in $\{A,B,B'\}$, $B'$ intersects all in $\{A,B,C,A',C'\}$, $C'$ intersects all in $\{B,C,D,B',D'\}$, and $D'$ intersects all in $\{C,D,C'\}$. 
				This is because there is only one component when the set $\{A,B,C,D\}$ is removed, and therefore the some pairs of the remaining disks should mutually intersect to preserve the connectivity.
				
				\begin{itemize}
					\item \textbf{If $\boldsymbol{A'}$ and $\boldsymbol{D'}$ intersect,} then the pairs $A'$ intersects $D$ and $A$ intersects $D'$ due to the triangle inequality. 
					Since $G$ does not contain an induced $C_4$, $A'$ intersects $C'$ or $B'$ intersects $D'$.
					If both of these pairs intersect, then the prescribed $\mathcal{S'} \supseteq \mathcal{S}=\{A,B,C,D,A',B',C',D'\}$ forms a clique and removing $\mathcal{S'}$ separates $G$ into four components, all of which are unit interval graphs.
					Otherwise, assume that $A'$ intersects $C'$ up to symmetry to the case when $B'$ intersects $D'$.
					Then, both  $\{A,B,C,D,A',B',C'\}$ and $\{A,B,C,D,A',C',D'\}$ form cliques and the removal of any of them separates $G$ into at four components, all of which are unit interval graphs.
					
					\item \textbf{Otherwise, $\boldsymbol{A'}$ and $\boldsymbol{D'}$ do not intersect.} In this case, $A$ intersects $C'$ or $D$ intersects $B'$ since otherwise, $\{A,B',C',D\}$ forms an induced $C_4$. Assume that $A$ intersects $C'$ up to symmetry to the case $D$ intersects $B'$. If $D$ and $B'$ also intersect, then the prescribed $\mathcal{S'} \supseteq \mathcal{S}=\{A,B,C,D,B',C'\}$ forms a maximal clique.
					Thus, removing $\mathcal{S'}$ separates $G$ into four components, all of which are unit interval graphs.
					Otherwise, $D$ and $B'$ do not intersect, and the prescribed $\mathcal{S'} \supseteq \mathcal{S}=\{A,B,C,B',C'\}$ forms a maximal clique and removing $\mathcal{S'}$ separates $G$ into at most four unit interval graphs.	
				\end{itemize}
			\end{enumerate}
			
			\item $\mathcal{L}$ is a diamond. Therefore, either $A$ intersects $C$, or $B$ intersects $D$. Let $A$ and $C$ intersect up to symmetry. Then both $\mathcal{S}_1=\{A,B,C\}$ and $\mathcal{S}_2=\{A,C,D\}$ form cliques. Let $\mathcal{S}_1' \supseteq \mathcal{S}_1$ and $\mathcal{S}_2' \supseteq \mathcal{S}_2$ be the prescribed maximal cliques. Assume that neither $G - \mathcal{S}_1'$ nor $G - \mathcal{S}_2'$ results in at most four (not necessarily connected) components which are all unit interval graphs. Similar to the previous case (i), $G - \mathcal{S}_1'\supseteq \mathcal{S}_1=\{A,B,C\}$ has a) two components where exactly one of them is not a unit interval graph, or b) exactly one component which is not a unit interval graph since even removing whole $\mathcal{L}$ does not separate the disks which belong to same set from $\{\XX^+,\YY^+,\XX^-,\YY^-\}$. Let us study the cases a) and b) that are mentioned above.
			\begin{enumerate}[a)]
				\item If $G - \{A,B,C\}$ has two components $\Delta_1$ and $\Delta_2$ where exactly one of them is not a unit interval graph, let $\Delta_1$ be that component. Then, up to symmetry, either $\Delta_1 = (\XX^+ \cup \YY^+ \cup \XX^-) \setminus \{A,B,C\}$ or $\Delta_1 = ( \YY^+ \cup \XX^- \cup \YY^-) \setminus \{A,B,C\}$ hold.
				
				
				\begin{itemize}
					\item \textbf{If $\boldsymbol{\Delta_1 = (\XX^+ \cup \YY^+ \cup \XX^-) \setminus \{A,B,C\}}$}, then there exist three disks $A' \in \XX^+$, $B' \in \YY^+$, and $C' \in \XX^-$ such that $\{A',B',C'\} \in \Delta_1$, $A'$ intersects $\{A,B,B'\}$, $B'$ intersects $\{A,B,C,A',C'\}$, $C'$ intersects $\{B,C,C'\}$, $D$ does not intersect $A'$ or $C'$, and $B'$ cannot intersect $D$ nor $D'$ since $B$ does not intersect $D$. If $A'$ and $C'$ intersect, then $\{A,B,C,A',B',C'\}$, else if $A'$ and $C$ intersect, then $\{A,B,C,A',B'\}$, else if $A$ and $C'$ intersect, then $\{A,B,C,B',C'\}$, and otherwise $\{A,B,C,B'\}$ forms a maximal clique and its removal results in four unit interval components.
					
					\item \textbf{Otherwise, those three disks are (up to symmetry) $\boldsymbol{B'}$, $\boldsymbol{C'}$, and $\boldsymbol{D'}$,} and intersect $\{B,C,C'\}$, $\{B,C,D,B',D'\}$ and $\{C,D,C'\}$, respectively\replace{, and}{. Note that} $B'$ cannot intersect $D$ since $B$ does not intersect $D$. \replace{Then}{Thus}, $C'$ intersects $A$ since otherwise, $\{A,B,C',D\}$ forms an induced $C_4$.
					If $A$ and $B'$ intersect, then $\{A,B,C,B',C'\}$, and otherwise, $\{A,B,C,C'\}$ forms a maximal clique and its removal results in four unit interval components.
				\end{itemize} 
				
				\item If $G - \{A,B,C\}$ has exactly one component $\Delta = (\XX^+ \cup \YY^+ \cup \XX^- \cup \YY^-) \setminus \mathcal{Q}$ which is not a unit interval, then there must exist at least fours disks, say $A'$, $B'$, $C'$, and $D'$ intersecting $\{A,B,B'\}$, $\{A,B,C,A',C'\}$, $\{B,C,D,B',D'\}$ and $\{C,D,C'\}$, respectively. $B'$ intersect neither $D$ nor $D'$, and $D'$ intersect neither $B$ nor $B'$ since $B$ and $D$ does not intersect.
				\begin{itemize}
					\item \textbf{If $\boldsymbol{A'}$ and $\boldsymbol{D'}$ intersect,} then $A'$ and $C'$ intersect since otherwise, $\{A',B',C',D'\}$ form an induced $C_4$. Therefore, $A'$ intersects $C$, and $A$ intersects $C'$. Then, each of $\{A,B,C,A',B',C'\}$ and $\{A,C,D,A',C',D'\}$ forms a maximal clique and its removal results in four unit interval components.
					
					\item \textbf{Otherwise, $\boldsymbol{A'}$ and $\boldsymbol{D'}$ do not intersect.} Then, $A$ intersects $C'$ since \remove{otherwise,} $\{A,B',C',D\}$ \replace{form}{is not} an induced $C_4$.
					\replace{Now}{Moreover}, if $A'$ intersects $D$, then, $A'$ and $C$ intersect since otherwise $\{A',B,C,D\}$ form an induced $C_4$\replace{, and}{. Also,} $A'$ and $C'$ intersect since otherwise $\{A',B,C',D\}$ form an induced $C_4$. 
					\replace{Then}{Considering these, the set} $\{A,B,C,A',B',C'\}$ forms a maximal clique and its removal results in four unit interval components. 
					\add{Because} otherwise, $A'$ and $D$ do not intersect.
					If $A$ and $D'$ intersect, then $A$ and $C'$ intersect since otherwise $\{A,C,C',D\}$ forms an induced $C_4$.
					Then, $\{A,B,C,B',C'\}$ or $\{A,B,C,A',B',C'\}$ forms a maximal clique and its removal results in four unit interval components. Else, $A$ and $D'$ do not intersect. Then, again, $\{A,B,C,B',C'\}$ or $\{A,B,C,A',B',C'\}$ forms a maximal clique and its removal results in four unit interval components.
				\end{itemize}
			\end{enumerate}
		\end{enumerate}
		
		Thus, the lemma holds.\qed
	\end{proof}

	Figure~\ref{fig:chordalS4} shows the possible cases investigated in the proof of Lemma~\ref{lem:CHOR4into4unitINTS} where the nodes correspond to disk centers, the blue edges exist in the induced diamond or $K_4$, the teal edges must exist in connected components, the red edges must exist due to triangle inequalities by the teal edges, the yellow edges must exist due to chordality, the pink edges must exist due to triangle inequalities by the yellow edges, and the green dotted edges exist in the considered subcases.
	
	\begin{lem}\label{lem:noMinThenChordal}
		If a given connected graph $G \in \APUD(1,1)$ does not contain a minimum 4-cycle, then $G$ is chordal.
	\end{lem}
	\begin{proof}
		Let the set $\mathcal{L}=\{A,B,C,D\}$ be on the disks such that $x_a$, $y_b$, $x_c$ and $y_d$ are the minimum coordinates among the centers of the disks belonging to $\XX^+$, $\YY^+$, $\XX^-$ and $\YY^-$, respectively in $\Sigma(G)$. Since $\mathcal{L}$ does not form a 4-cycle, $G$ can not contain an induced $C_4$ by Lemma~\ref{lem:4cycle} and Lemma~\ref{lem:someMandatoryEdges}. Also, $G$ contains no induced $C_k$ for $k >4$ by the characterization item \textbf{A1} of Corollary~\ref{corl:combine}. Therefore, $G$ is chordal.\qed
	\end{proof}

	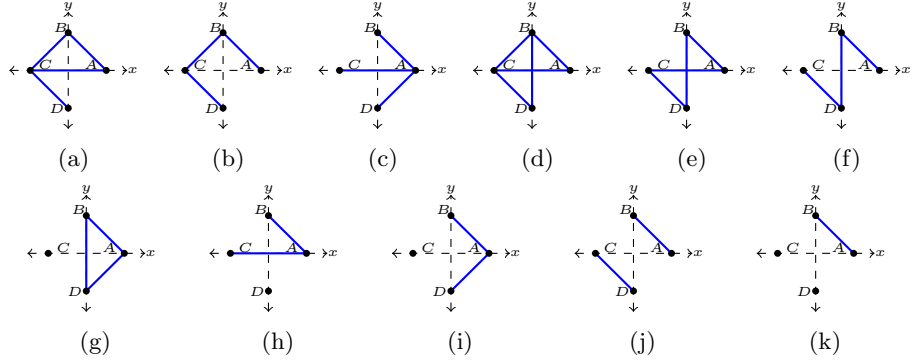
\begin{figure} [h!]
		\captionsetup[subfigure]{position=b}	
		\centering
		
		\begin{subfigure}[t]{0.14\linewidth}
			\centering
			\begin{tikzpicture}[xscale=-0.5,yscale=0.5]
				\coordinate (x-) at (1.55,0);
				\coordinate (x+) at (-1.55,0);
				\coordinate (y-) at (0,-1.55);
				\coordinate (y+) at (0,1.55);
				
				\draw [dashed,<->] (x-) -- (x+) node[midway, above]{};
				\node[] at (-1.7,0) {\tiny{$x$}};
				
				\draw [dashed,<->] (y-) -- (y+) node[midway, above]{};
				\node[] at (0,1.7) {\tiny{$y$}};

				\coordinate (a) at (-1,0);
				\coordinate (b) at (0,1);
				\coordinate (c) at (1,0);
				\coordinate (d) at (0,-1);
				
				\node[] at (-0.6,0.15) {\tiny{$A$}};
				\node[] at (0.2,1.15) {\tiny{$B$}};
				\node[] at (0.6,0.15) {\tiny{$C$}};
				\node[] at (0.3,-1) {\tiny{$D$}};
				
				\draw [blue,thick] (a) -- (b) node[midway, above]{};
				
				\draw [blue,thick] (a) -- (c) node[midway, above]{};
				
				\draw [blue,thick] (b) -- (c) node[midway, above]{};
				
				\draw [blue,thick] (c) -- (d) node[midway, above]{};
				
				\node[fill,circle,scale=0.3] at (-1,0) {};
				\node[fill,circle,scale=0.3] at (0,1) {};
				\node[fill,circle,scale=0.3] at (1,0) {};
				\node[fill,circle,scale=0.3] at (0,-1) {};
				
			\end{tikzpicture}
			\caption{}
			\label{fig:case1}
		\end{subfigure}
		~
		\begin{subfigure}[t]{0.14\linewidth}
			\centering
			\begin{tikzpicture}[xscale=-0.5,yscale=0.5]
				\coordinate (x-) at (1.55,0);
				\coordinate (x+) at (-1.55,0);
				\coordinate (y-) at (0,-1.55);
				\coordinate (y+) at (0,1.55);
				
				\draw [dashed,<->] (x-) -- (x+) node[midway, above]{};
				\node[] at (-1.7,0) {\tiny{$x$}};
				
				\draw [dashed,<->] (y-) -- (y+) node[midway, above]{};
				\node[] at (0,1.7) {\tiny{$y$}};
				
				\coordinate (a) at (-1,0);
				\coordinate (b) at (0,1);
				\coordinate (c) at (1,0);
				\coordinate (d) at (0,-1);
				
				\node[] at (-0.6,0.15) {\tiny{$A$}};
				\node[] at (0.2,1.15) {\tiny{$B$}};
				\node[] at (0.6,0.15) {\tiny{$C$}};
				\node[] at (0.3,-1) {\tiny{$D$}};
				
				\draw [blue,thick] (a) -- (b) node[midway, above]{};
				
				\draw [blue,thick] (b) -- (c) node[midway, above]{};
				
				\draw [blue,thick] (c) -- (d) node[midway, above]{};
				
				\node[fill,circle,scale=0.3] at (-1,0) {};
				\node[fill,circle,scale=0.3] at (0,1) {};
				\node[fill,circle,scale=0.3] at (1,0) {};
				\node[fill,circle,scale=0.3] at (0,-1) {};
				
			\end{tikzpicture}
			\caption{}
			\label{fig:case2}
		\end{subfigure}
		~
		\begin{subfigure}[t]{0.14\linewidth}
			\centering
			\begin{tikzpicture}[xscale=-0.5,yscale=0.5]
				\coordinate (x-) at (1.55,0);
				\coordinate (x+) at (-1.55,0);
				\coordinate (y-) at (0,-1.55);
				\coordinate (y+) at (0,1.55);
				
				\draw [dashed,<->] (x-) -- (x+) node[midway, above]{};
				\node[] at (-1.7,0) {\tiny{$x$}};
				
				\draw [dashed,<->] (y-) -- (y+) node[midway, above]{};
				\node[] at (0,1.7) {\tiny{$y$}};
				
				\coordinate (a) at (-1,0);
				\coordinate (b) at (0,1);
				\coordinate (c) at (1,0);
				\coordinate (d) at (0,-1);
				
				\node[] at (-0.6,0.15) {\tiny{$A$}};
				\node[] at (0.2,1.15) {\tiny{$B$}};
				\node[] at (0.6,0.15) {\tiny{$C$}};
				\node[] at (0.3,-1) {\tiny{$D$}};
				
				\draw [blue,thick] (a) -- (b) node[midway, above]{};
				
				\draw [blue,thick] (a) -- (c) node[midway, above]{};
				
				\draw [blue,thick] (a) -- (d) node[midway, above]{};
				
				\node[fill,circle,scale=0.3] at (-1,0) {};
				\node[fill,circle,scale=0.3] at (0,1) {};
				\node[fill,circle,scale=0.3] at (1,0) {};
				\node[fill,circle,scale=0.3] at (0,-1) {};
				
			\end{tikzpicture}
			\caption{}
			\label{fig:case3}
		\end{subfigure}
		~	
		\begin{subfigure}[t]{0.14\linewidth}
			\centering
			\begin{tikzpicture}[xscale=-0.5,yscale=0.5]
				\coordinate (x-) at (1.55,0);
				\coordinate (x+) at (-1.55,0);
				\coordinate (y-) at (0,-1.55);
				\coordinate (y+) at (0,1.55);
				
				\draw [dashed,<->] (x-) -- (x+) node[midway, above]{};
				\node[] at (-1.7,0) {\tiny{$x$}};
				
				\draw [dashed,<->] (y-) -- (y+) node[midway, above]{};
				\node[] at (0,1.7) {\tiny{$y$}};
				
				\coordinate (a) at (-1,0);
				\coordinate (b) at (0,1);
				\coordinate (c) at (1,0);
				\coordinate (d) at (0,-1);
				
				\node[] at (-0.6,0.15) {\tiny{$A$}};
				\node[] at (0.2,1.15) {\tiny{$B$}};
				\node[] at (0.6,0.15) {\tiny{$C$}};
				\node[] at (0.3,-1) {\tiny{$D$}};
				
				\draw [blue,thick] (a) -- (b) node[midway, above]{};
				
				\draw [blue,thick] (a) -- (c) node[midway, above]{};
				
				\draw [blue,thick] (b) -- (c) node[midway, above]{};
				
				\draw [blue,thick] (b) -- (d) node[midway, above]{};
				
				\draw [blue,thick] (c) -- (d) node[midway, above]{};
				
				\node[fill,circle,scale=0.3] at (-1,0) {};
				\node[fill,circle,scale=0.3] at (0,1) {};
				\node[fill,circle,scale=0.3] at (1,0) {};
				\node[fill,circle,scale=0.3] at (0,-1) {};
				
			\end{tikzpicture}
			\caption{}
			\label{fig:case4}
		\end{subfigure}
		~
		\begin{subfigure}[t]{0.14\linewidth}
			\centering
			\begin{tikzpicture}[xscale=-0.5,yscale=0.5]
				\coordinate (x-) at (1.55,0);
				\coordinate (x+) at (-1.55,0);
				\coordinate (y-) at (0,-1.55);
				\coordinate (y+) at (0,1.55);
				
				\draw [dashed,<->] (x-) -- (x+) node[midway, above]{};
				\node[] at (-1.7,0) {\tiny{$x$}};
				
				\draw [dashed,<->] (y-) -- (y+) node[midway, above]{};
				\node[] at (0,1.7) {\tiny{$y$}};
				
				\coordinate (a) at (-1,0);
				\coordinate (b) at (0,1);
				\coordinate (c) at (1,0);
				\coordinate (d) at (0,-1);
				
				\node[] at (-0.6,0.15) {\tiny{$A$}};
				\node[] at (0.2,1.15) {\tiny{$B$}};
				\node[] at (0.6,0.15) {\tiny{$C$}};
				\node[] at (0.3,-1) {\tiny{$D$}};
				
				\draw [blue,thick] (a) -- (b) node[midway, above]{};
				
				\draw [blue,thick] (a) -- (c) node[midway, above]{};
				
				\draw [blue,thick] (b) -- (d) node[midway, above]{};
				
				\draw [blue,thick] (c) -- (d) node[midway, above]{};
				
				\node[fill,circle,scale=0.3] at (-1,0) {};
				\node[fill,circle,scale=0.3] at (0,1) {};
				\node[fill,circle,scale=0.3] at (1,0) {};
				\node[fill,circle,scale=0.3] at (0,-1) {};
				
			\end{tikzpicture}
			\caption{}
			\label{fig:case5}
		\end{subfigure}
		~
		\begin{subfigure}[t]{0.14\linewidth}
			\centering
			\begin{tikzpicture}[xscale=-0.5,yscale=0.5]
				\coordinate (x-) at (1.55,0);
				\coordinate (x+) at (-1.55,0);
				\coordinate (y-) at (0,-1.55);
				\coordinate (y+) at (0,1.55);
				
				\draw [dashed,<->] (x-) -- (x+) node[midway, above]{};
				\node[] at (-1.7,0) {\tiny{$x$}};
				
				\draw [dashed,<->] (y-) -- (y+) node[midway, above]{};
				\node[] at (0,1.7) {\tiny{$y$}};
				
				\coordinate (a) at (-1,0);
				\coordinate (b) at (0,1);
				\coordinate (c) at (1,0);
				\coordinate (d) at (0,-1);
				
				\node[] at (-0.6,0.15) {\tiny{$A$}};
				\node[] at (0.2,1.15) {\tiny{$B$}};
				\node[] at (0.6,0.15) {\tiny{$C$}};
				\node[] at (0.3,-1) {\tiny{$D$}};
				
				\draw [blue,thick] (a) -- (b) node[midway, above]{};
				
				\draw [blue,thick] (b) -- (d) node[midway, above]{};
				
				\draw [blue,thick] (c) -- (d) node[midway, above]{};
				
				\node[fill,circle,scale=0.3] at (-1,0) {};
				\node[fill,circle,scale=0.3] at (0,1) {};
				\node[fill,circle,scale=0.3] at (1,0) {};
				\node[fill,circle,scale=0.3] at (0,-1) {};
				
			\end{tikzpicture}
			\caption{}
			\label{fig:case6}
		\end{subfigure}
		
		\begin{subfigure}[t]{0.17\linewidth}
			\centering
			\begin{tikzpicture}[xscale=-0.5,yscale=0.5]
				\coordinate (x-) at (1.55,0);
				\coordinate (x+) at (-1.55,0);
				\coordinate (y-) at (0,-1.55);
				\coordinate (y+) at (0,1.55);
				
				\draw [dashed,<->] (x-) -- (x+) node[midway, above]{};
				\node[] at (-1.7,0) {\tiny{$x$}};
				
				\draw [dashed,<->] (y-) -- (y+) node[midway, above]{};
				\node[] at (0,1.7) {\tiny{$y$}};
				
				\coordinate (a) at (-1,0);
				\coordinate (b) at (0,1);
				\coordinate (c) at (1,0);
				\coordinate (d) at (0,-1);
				
				\node[] at (-0.6,0.15) {\tiny{$A$}};
				\node[] at (0.2,1.15) {\tiny{$B$}};
				\node[] at (0.6,0.15) {\tiny{$C$}};
				\node[] at (0.3,-1) {\tiny{$D$}};
				
				\draw [blue,thick] (a) -- (b) node[midway, above]{};
				
				\draw [blue,thick] (a) -- (d) node[midway, above]{};
				
				\draw [blue,thick] (b) -- (d) node[midway, above]{};
				
				\node[fill,circle,scale=0.3] at (-1,0) {};
				\node[fill,circle,scale=0.3] at (0,1) {};
				\node[fill,circle,scale=0.3] at (1,0) {};
				\node[fill,circle,scale=0.3] at (0,-1) {};
				
			\end{tikzpicture}
			\caption{}
			\label{fig:case7}
		\end{subfigure}
		~
		\begin{subfigure}[t]{0.17\linewidth}
			\centering
			\begin{tikzpicture}[xscale=-0.5,yscale=0.5]
				\coordinate (x-) at (1.55,0);
				\coordinate (x+) at (-1.55,0);
				\coordinate (y-) at (0,-1.55);
				\coordinate (y+) at (0,1.55);
				
				\draw [dashed,<->] (x-) -- (x+) node[midway, above]{};
				\node[] at (-1.7,0) {\tiny{$x$}};
				
				\draw [dashed,<->] (y-) -- (y+) node[midway, above]{};
				\node[] at (0,1.7) {\tiny{$y$}};
				
				\coordinate (a) at (-1,0);
				\coordinate (b) at (0,1);
				\coordinate (c) at (1,0);
				\coordinate (d) at (0,-1);
				
				\node[] at (-0.6,0.15) {\tiny{$A$}};
				\node[] at (0.2,1.15) {\tiny{$B$}};
				\node[] at (0.6,0.15) {\tiny{$C$}};
				\node[] at (0.3,-1) {\tiny{$D$}};
				
				\draw [blue,thick] (a) -- (b) node[midway, above]{};
				
				\draw [blue,thick] (a) -- (c) node[midway, above]{};
				
				\node[fill,circle,scale=0.3] at (-1,0) {};
				\node[fill,circle,scale=0.3] at (0,1) {};
				\node[fill,circle,scale=0.3] at (1,0) {};
				\node[fill,circle,scale=0.3] at (0,-1) {};
				
			\end{tikzpicture}
			\caption{}
			\label{fig:case8}
		\end{subfigure}
		~
		\begin{subfigure}[t]{0.17\linewidth}
			\centering
			\begin{tikzpicture}[xscale=-0.5,yscale=0.5]
				\coordinate (x-) at (1.55,0);
				\coordinate (x+) at (-1.55,0);
				\coordinate (y-) at (0,-1.55);
				\coordinate (y+) at (0,1.55);
				
				\draw [dashed,<->] (x-) -- (x+) node[midway, above]{};
				\node[] at (-1.7,0) {\tiny{$x$}};
				
				\draw [dashed,<->] (y-) -- (y+) node[midway, above]{};
				\node[] at (0,1.7) {\tiny{$y$}};
				
				\coordinate (a) at (-1,0);
				\coordinate (b) at (0,1);
				\coordinate (c) at (1,0);
				\coordinate (d) at (0,-1);
				
				\node[] at (-0.6,0.15) {\tiny{$A$}};
				\node[] at (0.2,1.15) {\tiny{$B$}};
				\node[] at (0.6,0.15) {\tiny{$C$}};
				\node[] at (0.3,-1) {\tiny{$D$}};
				
				\draw [blue,thick] (a) -- (b) node[midway, above]{};
				
				\draw [blue,thick] (a) -- (d) node[midway, above]{};
				
				\node[fill,circle,scale=0.3] at (-1,0) {};
				\node[fill,circle,scale=0.3] at (0,1) {};
				\node[fill,circle,scale=0.3] at (1,0) {};
				\node[fill,circle,scale=0.3] at (0,-1) {};
				
			\end{tikzpicture}
			\caption{}
			\label{fig:case9}
		\end{subfigure}
		~
		\begin{subfigure}[t]{0.17\linewidth}
			\centering
			\begin{tikzpicture}[xscale=-0.5,yscale=0.5]
				\coordinate (x-) at (1.55,0);
				\coordinate (x+) at (-1.55,0);
				\coordinate (y-) at (0,-1.55);
				\coordinate (y+) at (0,1.55);
				
				\draw [dashed,<->] (x-) -- (x+) node[midway, above]{};
				\node[] at (-1.7,0) {\tiny{$x$}};
				
				\draw [dashed,<->] (y-) -- (y+) node[midway, above]{};
				\node[] at (0,1.7) {\tiny{$y$}};
				
				\coordinate (a) at (-1,0);
				\coordinate (b) at (0,1);
				\coordinate (c) at (1,0);
				\coordinate (d) at (0,-1);
				
				\node[] at (-0.6,0.15) {\tiny{$A$}};
				\node[] at (0.2,1.15) {\tiny{$B$}};
				\node[] at (0.6,0.15) {\tiny{$C$}};
				\node[] at (0.3,-1) {\tiny{$D$}};
				
				\draw [blue,thick] (a) -- (b) node[midway, above]{};
				
				\draw [blue,thick] (c) -- (d) node[midway, above]{};
				
				\node[fill,circle,scale=0.3] at (-1,0) {};
				\node[fill,circle,scale=0.3] at (0,1) {};
				\node[fill,circle,scale=0.3] at (1,0) {};
				\node[fill,circle,scale=0.3] at (0,-1) {};
				
			\end{tikzpicture}
			\caption{}
			\label{fig:case10}
		\end{subfigure}
		~
		\begin{subfigure}[t]{0.17\linewidth}
			\centering
			\begin{tikzpicture}[xscale=-0.5,yscale=0.5]
				\coordinate (x-) at (1.55,0);
				\coordinate (x+) at (-1.55,0);
				\coordinate (y-) at (0,-1.55);
				\coordinate (y+) at (0,1.55);
				
				\draw [dashed,<->] (x-) -- (x+) node[midway, above]{};
				\node[] at (-1.7,0) {\tiny{$x$}};
				
				\draw [dashed,<->] (y-) -- (y+) node[midway, above]{};
				\node[] at (0,1.7) {\tiny{$y$}};
				
				\coordinate (a) at (-1,0);
				\coordinate (b) at (0,1);
				\coordinate (c) at (1,0);
				\coordinate (d) at (0,-1);
				
				\node[] at (-0.6,0.15) {\tiny{$A$}};
				\node[] at (0.2,1.15) {\tiny{$B$}};
				\node[] at (0.6,0.15) {\tiny{$C$}};
				\node[] at (0.3,-1) {\tiny{$D$}};
				
				\draw [blue,thick] (a) -- (b) node[midway, above]{};
				
				\node[fill,circle,scale=0.3] at (-1,0) {};
				\node[fill,circle,scale=0.3] at (0,1) {};
				\node[fill,circle,scale=0.3] at (1,0) {};
				\node[fill,circle,scale=0.3] at (0,-1) {};
				
			\end{tikzpicture}
			\caption{}
			\label{fig:case11}
		\end{subfigure}
		
		\caption{The possible simple graphs of order four except a $K_4$, an induced diamond and an induced $C_4$, up to symmetry.}
		\label{fig:CHOR4noMinYy}
	\end{figure}
	
	Figure~\ref{fig:CHOR4noMinYy} shows all simple graphs of order four except a $K_4$, an induced diamond and an induced $C_4$. Considering the set $\mathcal{L}$ given in the proof of Lemma~\ref{lem:noMinThenChordal}, we give the following. 
	
	\begin{rem}\label{rem:allcases}
		\remove{$\mathcal{L}$ can not induce t}{T}he graphs given in Figure~\ref{fig:case4}, Figure~\ref{fig:case5} \replace{or}{and} Figure~\ref{fig:case6} \add{cannot appear in $\mathcal{L}$ as induced graphs} since the illustrated edges implies that the missing edges must exist by triangle similarity and inequality, and thus $\{A,B,C,D\}$ would form an induced diamond or a $K_4$. Also, $\mathcal{L}$ can not induce the graphs given in Figure~\ref{fig:case7} to Figure~\ref{fig:case11} as they would imply that $G$ is disconnected which contradicts our assumption. Therefore, 
		$\mathcal{L}$ induces the graphs given in Figure~\ref{fig:case1}, Figure~\ref{fig:case2} or Figure~\ref{fig:case3}.
	\end{rem}

	\begin{lem}\label{lem:CHOR4noMin}
		In a given connected graph $G \in \APUD(1,1)$, if the\add{re is no} minimum 4-cycle \remove{does not exist} in $\Sigma(G)$, then there exists a maximal clique $\mathcal{Q}$ of $G$ such that $G - \mathcal{Q}$ has $q$ connected components each forming a unit interval graph.
	\end{lem}

	\begin{proof}
		We know that $q \leq 4$, by Lemma~\ref{lem:components}.
		If all $q$ components of $G-\mathcal{Q}$ are unit interval graphs, then the lemma holds.
		We also know that if $q = 3$ or $q = 4$, then the lemma holds, again by Lemma~\ref{lem:components}.
		Therefore, we only study the following cases:
		\begin{enumerate}[(i)]
			\item $q=2$ and there exists exactly one component that is unit interval graph.
			\item $q=1$ and that component is not a unit interval graph.
		\end{enumerate}
		
		Since the minimum 4-cycle does not exist in $\Sigma(G)$, let us consider $\{A,B,C,D\}$ which is the induced subgraph of $G$ on $\{A,B,C,D\}$ not forming a 4-cycle where $x_a$, $y_b$, $x_c$ and $y_d$ are the minimum coordinates in $\mathcal{X}^+$, $\mathcal{Y}^+$, $\mathcal{X}^-$ and $\mathcal{Y}^-$, respectively. We, therefore, consider the following three subcases:
		\begin{enumerate}[a)]
			\item $\{A,B,C\}$ is a $K_3$ and $D$ intersect only $C$.
			\item $\{A,B,C,D\}$ is an induced path of length four in this order.
			\item $\{A,B,C,D\}$ forms an induced $K_{1,3}$ where all $\{A,C,D\}$ intersect $B$. 
		\end{enumerate}
		
		Using analogous arguments as in the proof of Lemma~\ref{lem:CHOR4into4unitINTS}, we get the following.
		
		\begin{enumerate}[(i)]
			\item $G-\{A,B,C,D\}$ has two connected components where exactly one of them is not a unit interval graph.
			These components are $\Delta_1 = (\mathcal{X}^+ \cup \mathcal{Y}^+ \cup \mathcal{X}^-) \setminus \{A,B,C,D\}$, and $\Delta_2 = \mathcal{Y}^- \setminus \{A,B,C,D\}$.
			Then, there exist at least three disks $A' \in \mathcal{X}^+$, $B' \in \mathcal{Y}^+$, and $C' \in \mathcal{X}^-$ such that $\{A',B',C'\} \subseteq \Delta_1$, $B'$ intersects both $A'$ and $C'$.
			Moreover, these disks $A'$,$B'$ and $C'$ are such that the Euclidean distance between the centers of the pairs $(A', B')$ and $(B',C')$ are maximum.
			Note that there might be more disks centered between the centers of $A$ and $A'$ (resp. $B$ and $B'$, $C$ and $C'$).
			By triangle inequality, $A'$ intersects all in $\{A,B,B'\}$, $B'$ intersects all in $\{A,B,C,A',C'\}$, and $C'$ intersects all in $\{B,C,B'\}$.
			Note that none of the disks in $\Delta_1$ intersects a disk in $\Delta_2$ since $\Delta_1$ and $\Delta_2$ are disjoint. Now, we study the following three subcases.
			\begin{enumerate}[a)]
				\item $\{A,B,C\}$ is a $K_3$ and $D$ intersect only $C$. If $A'$ and $C'$ intersect, then, $\{A,B,C,A',B',C'\}$ forms a maximal clique and its removal results in four unit interval components.
				Otherwise, if $A$ and $C'$ (analogously, $A'$ and $C$) intersect, $\{A,B,C,B',C'\}$ (analogously, $\{A,B,C,A',B'\}$), and else, i.e. neither $A$ and $C'$ nor $A'$ and $C$ intersect, $\{A,B,C,B'\}$ forms a maximal clique and its removal results in four unit interval components.
				
				\item $\{A,B,C,D\}$ is an induced path of length four in this order. Then, $\{B,C,B',C'\}$ may not form a maximal clique or its removal may not result in four unit interval components. If the removal of  $\{B,C,B',C'\}$ does not result in four unit interval components, then there exists another disk $B'' \in \mathcal{Y}^+$ intersecting $A$ such that the Euclidean distance between the centers of $(A,B'')$ is the maximum. Then, the following hold.	
				
				\begin{align*}
					{x_a}^2 + {y_{b''}}^2 \leq 4, \text{  } {x_c}^2 + {x_{c'}}^2  \leq 4, \text{  } {x_a}^2 + {x_{c}}^2 & > 4 \\
					\therefore {x_a}^2 + {x_c}^2 + {y_{b''}}^2 + {x_{c'}}^2 \leq 8, \text{  } {y_{b''}}^2 + {x_{c'}}^2 < 4, \text{  } \sqrt{{y_{b''}}^2 + {x_{c'}}^2}& < 2\\			
				\end{align*}
				Then, $\{B,C,B',C',B''\}$ forms a maximal clique and its removal results in four unit interval components.
				
				\item $\{A,B,C,D\}$ forms an induced $K_{1,3}$ where all $\{A,C,D\}$ intersect $B$. Then, both $\{A,B,A',B'\}$ and $\{B,C,B',C'\}$ form maximal cliques and the removal of any of them results in four unit interval components.
			\end{enumerate}
			
			\item $G-\{A,B,C,D\}$ has one connected component $\Delta = (\mathcal{X}^+ \cup \mathcal{Y}^+ \cup \mathcal{X}^- \cup \mathcal{Y}^+) \setminus \{A,B,C,D\}$ which is not a unit interval graph. Then, there exist at least four disks $A' \in \mathcal{X}^+$, $B' \in \mathcal{Y}^+$, $C' \in \mathcal{X}^-$ and $D' \in \mathcal{Y}^-$ such that $\{A',B',C',D'\} \subseteq \Delta$.
			Note that there might be more disks centered between the centers of $A$ and $A'$ (resp. $B$ and $B'$, $C$ and $C'$, $D$ and $D'$).
			Now, we study the following three subcases.
			
			\begin{enumerate}[a)]
				\item $\{A,B,C\}$ is a $K_3$ and $D$ intersect only $C$. $A'$ intersects $B'$, $B'$ intersects $C'$ and $C'$ intersects $D'$. Moreover, these disks $A'$,$B'$, $C'$ and $D'$ are such that the Euclidean distance between the centers of the pairs $(A', B')$, $(B',C')$ and $(C',D')$ are maximum.
				By triangle inequality, $A'$ intersects all in $\{A,B,B'\}$, $B'$ intersects all in $\{A,B,C,A',C'\}$, $C'$ intersects all in $\{B,C,D,B',D'\}$, and $D'$ does not intersect $A'$ since $D$ does not intersect $A$. If $A'$ and $C'$ intersect, then $\{A,B,C,A',B',C'\}$ forms a maximal clique and its removal results in four unit interval components.
				Otherwise, if $A$ and $C'$ (analogously, $A'$ and $C$) intersect, $\{A,B,C,B',C'\}$ (analogously, $\{A,B,C,A',B'\}$), and else, i.e. neither $A$ and $C'$ nor $A'$ and $C$ intersect, $\{B,C,B',C'\}$ forms a maximal clique and its removal results in four unit interval components. 
				
				\item $\{A,B,C,D\}$ is an induced path of length four in this order. $A'$ intersects $B'$, $B'$ intersects $C'$, and $C'$ intersects $D'$.
				Moreover, these disks $A'$, $B'$, $C'$ and $D'$ are such that the Euclidean distance between the centers of the pairs $(A', B')$, $(B',C')$ and $(C',D')$ are maximum. By triangle inequality, $A'$ intersects all in $\{A,B,B'\}$, $B'$ intersects all in $\{A,B,C,A',C'\}$, $C'$ intersects all in $\{B,C,D,B',D'\}$, and $D'$ intersects all in $\{C,D,C'\}$. Then, $\{B,C,B',C'\}$ may not form a maximal clique or its removal may not result in four unit interval components. If the removal of  $\{B,C,B',C'\}$ does not result in four unit interval components, then there exists another disk $B'' \in \mathcal{Y}^+$ intersecting $A$ such that the Euclidean distance between the centers of $(A,B'')$ is the maximum or another disk $C'' \in \mathcal{X}^+-$ intersecting $D$ such that the Euclidean distance between the centers of $(D,C'')$ is the maximum. If both $B''$ and $C''$ exist, then the following hold.	
				\begin{align*}
					{x_a}^2 + {y_{b''}}^2 \leq 4, \text{  }
					{y_d}^2 + {x_{c''}}^2 \leq 4, \text{  }
					{x_a}^2 + {y_d}^2 & > 4\\
					\therefore 
					{x_a}^2 + {y_d}^2 + {y_{b''}}^2 + {x_{c''}}^2 \leq 8, \text{  }
					{y_{b''}}^2 + {x_{c''}}^2 < 4, \text{  }
					\sqrt{{y_{b''}}^2 + {x_{c''}}^2} &< 2
					\\
					\therefore \sqrt{{y_{b''}}^2 + {e}^2} < 2 \text{ for } e \in \{x_{c'},x_{c''}\}
					, \text{  }
					\sqrt{{x_{c''}}^2 + {f}^2} < 2 \text{ for } f \in \{y_{b'},&y_{b''}\}	
				\end{align*}
				Then, $\{B,C,B',C',B'',C''\}$ forms a maximal clique and its removal results in four unit interval components. If only one of $B''$ and $C''$ exists, say $B''$ up to symmetry, then, it follows from item (ii) a) that $\{B,C,B',C',B''\}$ forms a maximal clique and its removal results in four unit interval components.
				
				\item $\{A,B,C,D\}$ forms an induced $K_{1,3}$ where all $\{A,C,D\}$ intersect $B$. $B'$ intersects $A'$, $C'$ and $D'$.
				Moreover, these disks $A'$, $B'$, $C'$ and $D'$ are such that the Euclidean distance between the centers of the pairs $(A', B')$, $(B',C')$ and $(B',D')$ are maximum. By triangle inequality, $A'$ intersects all in $\{A,B,B'\}$, $B'$ intersects all in $\{A,B,C,D,A',C',D'\}$, $C'$ intersects all in $\{B,C,B'\}$, and $D'$ intersects all in $\{B,D,B'\}$. Then, all $\{A,B,A',B'\}$, $\{B,C,B',C'\}$, $\{B,D,B',D'\}$ form maximal cliques and the removal of any of them results in four unit interval components.
			\end{enumerate}
		\end{enumerate}
		Thus, the lemma holds.\qed
	\end{proof}

	Figure~\ref{fig:CHOR4noMin} shows the possible cases investigated in the proof of Lemma~\ref{lem:CHOR4noMin} where the nodes correspond to disk centers, the blue edges exist in the considered induced graph on $\{A,B,C,D\}$, the teal edges must exist in connected components, the red edges must exist due to triangle inequalities by the teal edges, and the green dotted edges exist in the considered subcases.

	\begin{figure} [t]
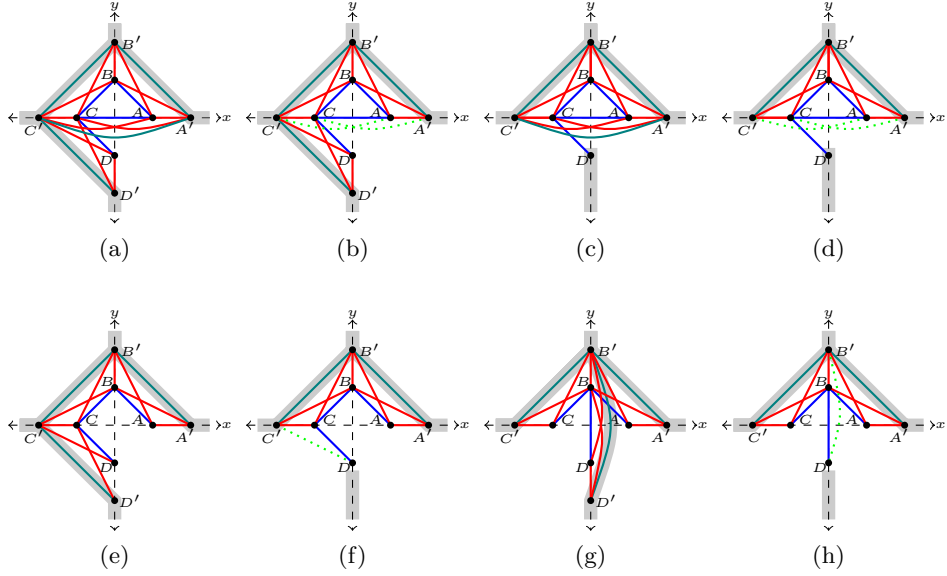

		\centering
		
		\begin{subfigure}[t]{0.23\linewidth}
			\centering

			\label{fig:case1K4xxxy}
		\end{subfigure}
		
		\caption{The possible cases in the setting of Lemma~\ref{lem:CHOR4noMin}.}
		\label{fig:CHOR4noMin}
	\end{figure}

	For a tree $T$, a \emph{$T$-graph} is the intersection graph of connected subtree of a subdivision of $T$. The complete bipartite graph $K_{1,d}$ is called a \emph{star} with $d$ rays, and also denoted by $S_d$.
	An \emph{$S_d$-graph} is the intersection graph of connected substars of a subdivision of the star $S_d$ with $d$ rays and they form a subset of $T$-graphs. Every $S_d$-graph $G$ contains a maximal clique $\mathcal{Q}$ such that the partial order on the connected components of $G-\mathcal{Q}$ can be covered by four chains each forming an interval graph and $S_d$-graphs can be recognized in polynomial time~\cite{zemanWG}. It is known that every chordal graph is a $T$-graph for some tree $T$ \cite{chordalityInters} and next, we prove a stronger result considering a chordal $\APUD(1,1)$.
	
	\begin{lem}\label{lem:Sd}
		If a given graph $G \in \APUD(1,1)$ is chordal, then $G \in S_4$-graph.
	\end{lem}  
	\begin{proof}
		It follows from that an $S_4$-graph $G$ contains a maximal clique $\mathcal{Q}$ such that the partial order on the connected components of $G-\mathcal{Q}$ can be covered by four chains each forming an interval graph \cite{zemanWG}.
		
		Formally, since $G$ is chordal, it has linearly many maximal cliques which can be listed in linear time \cite{recogChordaLinear}. Then, one can identify all maximal cliques which adapts the setting of Lemma~\ref{lem:CHOR4into4unitINTS} or Lemma~\ref{lem:CHOR4noMin} depending on whether the minimum 4-cycle exists in some $\Sigma(G)$ or not. Among them, any maximal clique $\mathcal{Q}$ which results in at most four unit interval graphs can be placed on the center of $S_4$ and those at most four unit interval graphs can be placed on four rays of $S_4$ if the attachments of each of them forms a chain by inclusion on $\mathcal{Q}$. Since $G$ is an $\APUD(1,1)$, there exists at least one such $\mathcal{Q}$ by Lemma~\ref{lem:CHOR4into4unitINTS} and Lemma~\ref{lem:CHOR4noMin}.\qed
	\end{proof}

	Since any induced $C_4$ in an $\APUD(1,1)$ is on four disks belonging to distinct sets in $\{\XX^+,\YY^+,\XX^-,\YY^-\}$ by Lemma~\ref{lem:4cycle}, Lemma~\ref{lem:Sd} immediately gives the following.
	
	\begin{corl}\label{corl:S3}
		If a given graph $G \in \APUD(1,1)$ is chordal and has an embedding $\Sigma(G)$ such that every disk belongs to  $\XX^+$, $\YY^+$ or $\XX^-$ (up to symmetry) in $\Sigma(G)$, then $G \in S_3$-graph.
	\end{corl}

	A \emph{$4$-sun} is a graph on eight vertices such that four of them form a $K_4$, and each of the other vertices is adjacent to a distinct pair of vertices from that $K_4$.

	\begin{rem}\label{lem:nosubclass}
		The graph class $\APUD(1,1)$ is neither a subclass nor a superclass of $S_d$-graphs. This directly follows from the characterization item \textbf{A1}. A $4$-sun is a forbidden induced subgraph for an $\APUD(1,1)$ while $S_d$-graphs for $d \geq 4$ can contain an induced $4$-sun. On the other side, all $S_d$-graphs are chordal meaning that they do not contain any induced cycles while an $\APUD(1,1)$ can contain induced cycles of length $4$.
	\end{rem}

	By Lemma~\ref{lem:Sd} and Corollary~\ref{corl:S3}, a chordal $\APUD(1,1)$ is an $S_4$-graph since $S_3$-graphs are a subclass of $S_4$-graphs \cite{zemanWG}. However, by Lemma~\ref{lem:nosubclass}, there exist $S_4$-graphs which are not $\APUD(1,1)$. Therefore, we may not simply use the $S_d$-graph recognition algorithm which works in polynomial time independently from the value of $d \leq n$ for a graph of order $n$. We show how we recognize a chordal $\APUD(1,1)$ next.
	
	\begin{algorithm}
		\KwIn{A chordal graph $G$}
		\KwOut{Whether $G \in \APUD(1,1)$ holds}
		\lIf{$G$ contains a 4-sun}{\Return{\textsc{false}}}
		\ForEach{maximal clique $Q \in G$}{
			Remove $Q$ from $G$\;
			\lIf{$G$ contains at most four unit interval graphs}{\Return{\textsc{true}}}
		}
		\caption{Testing whether a given chordal graph is an $\APUD(1,1)$}
		\label{alg:chordalapud}
	\end{algorithm}
	
	
	\begin{corl}
		Since the recognition problem can be solved in linear time on unit interval graphs \cite{DBLP:journals/ipl/Keil85} and chordal graphs have linearly many cliques which can be listed in linear time \cite{recogChordaLinear}, Algorithm~\ref{alg:chordalapud} runs in polynomial time with respect to the number of vertices of $G$ by\cite{zemanWG}, and correctly determines whether a given chordal graph can be embedded onto two perpendicular lines as unit disks by Lemma~\ref{lem:Sd}.
	\end{corl}
	
	\section{Recognizing an $\boldsymbol{\APUD(1,1)}$ in polynomial time} \label{sec:poly}
	
	In this section, we consider general $\APUD(1,1)$ and give a polynomial time recognition algorithm. We first start with the following.
	
	\begin{rem}\label{rem:chordalInducedC4}
		Every $\APUD(1,1)$ $G$ which is not chordal contains an induced cycle of length at least four, and by Corollary~\ref{corl:combine}, $G$ contains an induced cycle of length at most four. Then, $G$ must have a minimum cycle by Lemma~\ref{lem:4cycle}.
	\end{rem}

	%
	%


	\begin{theorem}\label{theo:linearlyMany}
		If $G$ is an $\APUD(1,1)$, then $G$ has $\mathcal{O}(n^3)$ maximal cliques.
	\end{theorem}

	\begin{proof}
		If $G$ is chordal or a Helly graph, then the theorem holds since they have linearly many maximal cliques~\cite{zemanWG,recogChordaLinear}. Otherwise, we claim that there is a polynomial number of non-Helly cliques on the number of disks witnessing that $G$ has polynomially many maximal cliques. By Lemma~\ref{lem:ONLYnonHelly}, every non-Helly clique in $G$ without disks with the same neighborhood contains a non-Helly clique on three disks. Let $\mathcal{C}$ be a non-Helly clique on three distinct disks $A,B,C$ such that $A \in \XX^+$, $B  \in \YY^+$ and $C \in \XX^-$ by Lemma~\ref{lem:3clique}. We give Claim~\ref{clm:onxx}, Claim~\ref{clm:onyy}, and Claim~\ref{clm:C4together} in this setting to prove Theorem~\ref{theo:linearlyMany}.
		
		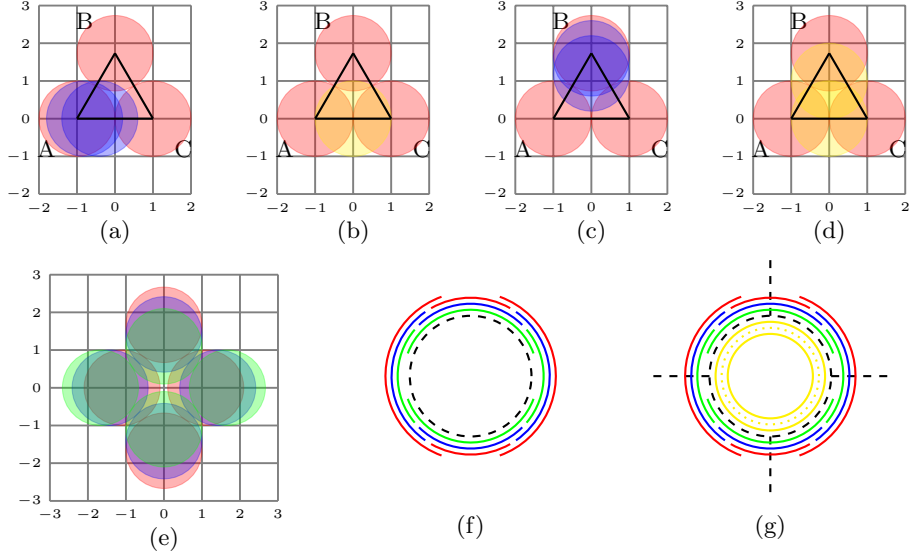
\begin{figure} [h!]
			\captionsetup[subfigure]{position=b}	
			\centering
			
			\begin{subfigure}[b]{0.23\linewidth}
				\centering
				\begin{tikzpicture}[scale=0.5]
					\draw (-2,-2) to[grid with coordinates] (2,3);
					
					\coordinate (A) at (-1,0);
					\node[] at (-1.8,-0.8) {A};
					
					\coordinate (B) at (0,1.735);
					\node[] at (-0.8,2.6) {B};
					
					\coordinate (C) at (1,0);
					\node[] at (1.8,-0.8) {C};
					
					\coordinate (D1) at (-0.8,0);
					\coordinate (D3) at (-0.4,0);
					
					\filldraw[color=red, fill=red, opacity=0.3] (A) circle (1);
					
					\filldraw[color=red, fill=red, opacity=0.3] (B) circle (1);
					
					\filldraw[color=red, fill=red, opacity=0.3] (C) circle (1);
					
					\filldraw[color=blue, fill=blue, opacity=0.3] (D1) circle (1);
					
					
					\filldraw[color=blue, fill=blue, opacity=0.3] (D3) circle (1);
					
					\draw[thick] (A) -- (B);
					\draw[thick] (A) -- (C);
					\draw[thick] (B) -- (C);
					
					\node (xa) at (0,-3) {(a)};
					
				\end{tikzpicture}
				\label{fig:nonHelly1}
			\end{subfigure}
			~
			\begin{subfigure}[b]{0.23\linewidth}
				\centering
				\begin{tikzpicture}[scale=0.5]
					\draw (-2,-2) to[grid with coordinates] (2,3);
					
					\coordinate (A) at (-1,0);
					\node[] at (-1.8,-0.8) {A};
					
					\coordinate (B) at (0,1.735);
					\node[] at (-0.8,2.6) {B};
					
					\coordinate (C) at (1,0);
					\node[] at (1.8,-0.8) {C};
					
					\coordinate (D1) at (0,0);
					
					\filldraw[color=red, fill=red, opacity=0.3] (A) circle (1);
					
					\filldraw[color=red, fill=red, opacity=0.3] (B) circle (1);
					
					\filldraw[color=red, fill=red, opacity=0.3] (C) circle (1);
					
					\filldraw[color=yellow, fill=yellow, opacity=0.3] (D1) circle (1);
					
					\draw[thick] (A) -- (B);
					\draw[thick] (A) -- (C);
					\draw[thick] (B) -- (C);
					
					\node (xb) at (0,-3) {(b)};
				\end{tikzpicture}
				\label{fig:Helly1}
			\end{subfigure}
			~
			\begin{subfigure}[b]{0.23\linewidth}
				\centering
				\begin{tikzpicture}[scale=0.5]
					\draw (-2,-2) to[grid with coordinates] (2,3);
					
					\coordinate (A) at (-1,0);
					\node[] at (-1.8,-0.8) {A};
					
					\coordinate (B) at (0,1.735);
					\node[] at (-0.8,2.6) {B};
					
					\coordinate (C) at (1,0);
					\node[] at (1.8,-0.8) {C};
					
					\coordinate (D1) at (0,1.6);
					\coordinate (D3) at (0,1.2);
					
					\filldraw[color=red, fill=red, opacity=0.3] (A) circle (1);
					
					\filldraw[color=red, fill=red, opacity=0.3] (B) circle (1);
					
					\filldraw[color=red, fill=red, opacity=0.3] (C) circle (1);
					
					\filldraw[color=blue, fill=blue, opacity=0.3] (D1) circle (1);
					
					
					\filldraw[color=blue, fill=blue, opacity=0.3] (D3) circle (1);
					
					\draw[thick] (A) -- (B);
					\draw[thick] (A) -- (C);
					\draw[thick] (B) -- (C);
					
					\node (x) at (0,-3) {(c)};
					
				\end{tikzpicture}
				\label{fig:nonHelly2}
			\end{subfigure}
			~
			\begin{subfigure}[b]{0.23\linewidth}
				\centering
				\begin{tikzpicture}[scale=0.5]
					\draw (-2,-2) to[grid with coordinates] (2,3);
					
					\coordinate (A) at (-1,0);
					\node[] at (-1.8,-0.8) {A};
					
					\coordinate (B) at (0,1.735);
					\node[] at (-0.8,2.6) {B};
					
					\coordinate (C) at (1,0);
					\node[] at (1.8,-0.8) {C};
					
					\coordinate (D4) at (0,1);
					\coordinate (D6) at (0,0);
					
					\filldraw[color=red, fill=red, opacity=0.3] (A) circle (1);
					
					\filldraw[color=red, fill=red, opacity=0.3] (B) circle (1);
					
					\filldraw[color=red, fill=red, opacity=0.3] (C) circle (1);
					
					\filldraw[color=yellow, fill=yellow, opacity=0.3] (D4) circle (1);
					
					\filldraw[color=yellow, fill=yellow, opacity=0.3] (D6) circle (1);
					
					\draw[thick] (A) -- (B);
					\draw[thick] (A) -- (C);
					\draw[thick] (B) -- (C);
					
					\node (xd) at (0,-3) {(d)};
					
				\end{tikzpicture}
				\label{fig:Helly2}
			\end{subfigure}
			
			\begin{subfigure}[b]{0.31\linewidth}
				\centering
				\begin{tikzpicture}[scale=0.5]
					\draw (-3,-3) to[grid with coordinates] (3,3);
					
					
					\coordinate (A) at (-1.1,0);
					\coordinate (B) at (1.1,0);
					\coordinate (C) at (0,1.67);
					\coordinate (D) at (0,-1.67);
					\coordinate (A2) at (-1.415,0);
					\coordinate (B2) at (1.415,0);
					\coordinate (C2) at (0,1.415);
					\coordinate (D2) at (0,-1.415);	
					\coordinate (A3) at (-1.67,0);
					\coordinate (B3) at (1.67,0);
					\coordinate (C3) at (0,1.1);
					\coordinate (D3) at (0,-1.1);
					
					\filldraw[color=red, fill=red, opacity=0.3] (A) circle (1);
					
					\filldraw[color=red, fill=red, opacity=0.3] (B) circle (1);
					
					\filldraw[color=red, fill=red, opacity=0.3] (C) circle (1);
					
					\filldraw[color=red, fill=red, opacity=0.3] (D) circle (1);
					
					\filldraw[color=blue, fill=blue, opacity=0.3] (A2) circle (1);
					
					\filldraw[color=blue, fill=blue, opacity=0.3] (B2) circle (1);
					
					\filldraw[color=blue, fill=blue, opacity=0.3] (C2) circle (1);
					
					\filldraw[color=blue, fill=blue, opacity=0.3] (D2) circle (1);
					
					\filldraw[color=green, fill=green, opacity=0.3] (A3) circle (1);
					
					\filldraw[color=green, fill=green, opacity=0.3] (B3) circle (1);
					
					\filldraw[color=green, fill=green, opacity=0.3] (C3) circle (1);
					
					\filldraw[color=green, fill=green, opacity=0.3] (D3) circle (1);
					
					\node (xe) at (0,-4) {(e)};
					
				\end{tikzpicture}
				\label{fig:nonHelly13}
			\end{subfigure}
			~
			\begin{subfigure}[b]{0.31\linewidth}
				\centering
				\begin{tikzpicture}[scale=0.4]
					
					\draw[white] (2,-4) rectangle (9,-6);
					
					\centerarc[black,thick,dashed](6,1)(0:360:2);
					
					\centerarc[green,thick](6,1)(20:160:2.2);
					\centerarc[green,thick](6,1)(140:220:2.4);
					\centerarc[green,thick](6,1)(200:340:2.2);
					\centerarc[green,thick](6,1)(320:400:2.4);
					
					\centerarc[blue,thick](6,1)(45:135:2.4);
					\centerarc[blue,thick](6,1)(125:235:2.6);
					\centerarc[blue,thick](6,1)(225:315:2.4);
					\centerarc[blue,thick](6,1)(305:415:2.6);
					
					\centerarc[red,thick](6,1)(60:120:2.6);
					\centerarc[red,thick](6,1)(110:250:2.8);
					\centerarc[red,thick](6,1)(240:300:2.6);
					\centerarc[red,thick](6,1)(290:430:2.8);
					
					
					\node (xf) at (6,-4) {(f)};
					
				\end{tikzpicture}
				\label{fig:rectsCARC}
			\end{subfigure}
			~
			\begin{subfigure}[b]{0.31\linewidth}
				\centering
				\begin{tikzpicture}[scale=0.4]
					
					\draw[white] (2,-4) rectangle (9,-6);
					
					\centerarc[black,thick,dashed](6,1)(0:360:2);
					
					\draw[thick,dashed] (8,1) -- (10,1);
					
					\draw[thick,dashed] (4,1) -- (2,1);
					
					\draw[thick,dashed] (6,3) -- (6,5);
					
					\draw[thick,dashed] (6,-1) -- (6,-3);
					
					\centerarc[yellow,thick](6,1)(0:360:1.8);
					\centerarc[yellow,thick,dotted](6,1)(0:360:1.6);
					\centerarc[yellow,thick](6,1)(0:360:1.4);
					
					\centerarc[green,thick](6,1)(20:160:2.2);
					\centerarc[green,thick](6,1)(140:220:2.4);
					\centerarc[green,thick](6,1)(200:340:2.2);
					\centerarc[green,thick](6,1)(320:400:2.4);
					
					\centerarc[blue,thick](6,1)(45:135:2.4);
					\centerarc[blue,thick](6,1)(125:235:2.6);
					\centerarc[blue,thick](6,1)(225:315:2.4);
					\centerarc[blue,thick](6,1)(305:415:2.6);
					
					\centerarc[red,thick](6,1)(60:120:2.6);
					\centerarc[red,thick](6,1)(110:250:2.8);
					\centerarc[red,thick](6,1)(240:300:2.6);
					\centerarc[red,thick](6,1)(290:430:2.8);
					
					\node (xg) at (6,-4) {(g)};
					
				\end{tikzpicture}
				\label{fig:rectsCARCextra}
			\end{subfigure}
			
			\caption{The cases investigated in Claim~\ref{clm:onxx}, Claim~\ref{clm:onyy} and Claim~\ref{clm:C4together} to prove Theorem~\ref{theo:linearlyMany}.}
			\label{fig:nonHellyExx}
		\end{figure}
		
		\begin{clm}\label{clm:onxx}
			There is a linear number of non-Helly cliques which can occur by two disks from $A,B,C$ and another disk centered between the centers of $A$ and $C$.  
		\end{clm}
		
		\begin{proof}
			Observe that any such non-Helly clique contains either $A$ and $B$, or $C$ and $B$ by Lemma~\ref{lem:helCl}. Since non-Helly cliques are due to the cliques of size three by Lemma~\ref{lem:ONLYnonHelly}, we obtain at most linearly many non-Helly cliques on the order of $G$. Figure~\ref{fig:nonHellyExx} (a) and (b) show where we do and do not obtain non-Helly cliques other than the red disks forming a non-Helly clique.\hfill  $\triangleleft$
		\end{proof}
		
		\begin{clm}\label{clm:onyy}
			There is a linear number of non-Helly cliques which can occur by two disks from $A,B,C$ and another disk centered between the point $(0,0)$ and the center of $B$. 
		\end{clm}
		
		\begin{proof}
			Observe that any such non-Helly clique contains either $A$, $C$ and the new disk, or $A$, $B$ and the new disk (analogously $C$, $B$ and the new disk) by Lemma~\ref{lem:helCl}. However, $A$, $B$ and the new disk can not form a non-Helly clique since the new disk contains the intersection of $A$ and $B$. Thus, we consider the possible non-Helly cliques which may be formed by $A$, $C$ and the new disk. Since non-Helly cliques are due to the cliques of size three by Lemma~\ref{lem:ONLYnonHelly}, we obtain at most linearly many non-Helly cliques on the order of $G$. Figure~\ref{fig:nonHellyExx} (c) and (d) show where we do and do not obtain any non-Helly cliques except the red disks forming a non-Helly clique.\hfill $\triangleleft$
		\end{proof}

		\begin{clm}\label{clm:C4together}
			Nested induced $C_4$s result in $\mathcal{O}(n^3)$ non-Helly cliques. Moreover, nested induced $C_4$s without induced $W_4$s (thus, without induced diamonds) contain no non-Helly cliques.
		\end{clm}
		
		\begin{proof}
			
			Considering all pairs of disks in nested induced $C_4$s, the former argument follows from Claim~\ref{clm:onxx} and Claim~\ref{clm:onyy}. For the latter, since considered nested induced $C_4$s contain no induced $W_4$ or diamond, they do not contain a $C_3$ on disks from distinct sides from $\{\XX^+,\YY^+,\XX^-,\YY-\}$, and therefore, they do not contain a non-Helly clique.
			\hfill $\triangleleft$
		\end{proof}
		
		By Claim~\ref{clm:onxx}, Claim~\ref{clm:onyy}, Claim~\ref{clm:C4together}, there can be $\mathcal{O}(n^3)$ non-Helly cliques on the size of the input graph. Since there is a linear number of Helly maximal cliques~\cite{zemanWG,recogChordaLinear}, an $\APUD(1,1)$ has $\mathcal{O}(n^3)$ maximal cliques.\qed
	\end{proof}
	
	Figure~\ref{fig:nonHellyExx} shows the cases investigated in Theorem~\ref{theo:linearlyMany}. In each (a), (b), (c), (d) and (e), we have an $\APUD(1,1)$ where the red disks labelled $A,B,C$ form a non-Helly clique. 	
	In (a), $A,B$ and any blue disk, and $A,C$ and any blue disk form a Helly clique while $B,C$ and any blue disk forms a non-Helly clique. 	
	In (b), the yellow disk and any two disks from $A,B,C$ form a Helly clique. 	
	In (c), $A,C$ and any blue disk form a non-Helly clique.
	In (d), $A,C$ and any disk placed between and including each yellow disk form a Helly clique. In (e), three induced $C_4$s on four red, four blue and four green disks form only Helly cliques, and in (f), the corresponding circular-arc graph without a non-Helly clique is shown. In (g), nested induced $C_4$s with yellow arcs without a non-Helly clique are illustrated.
	
	\begin{lem}\label{lem:connectedInduced}
		Let $G$ be a connected $\APUD(1,1)$ which has at least two disjoint 
		induced $C_4$s $\mathcal{L}_1$ and $\mathcal{L}_2$. Then, each disk of 
		$\mathcal{L}_1$ is adjacent to at least one disk of $\mathcal{L}_2$. 
		Moreover, if there is exactly one such adjacency for each disk, then it 
		is between the disks belonging to the same set from 
		$\{\XX^+,\YY^+,\XX^-,\YY^-\}$.
	\end{lem}
	
	\begin{proof}
		It follows from the connectedness of $G$, 
		Lemma~\ref{lem:4cycle} and Lemma~\ref{lem:someMandatoryEdges}.\qed
	\end{proof}

	Lemma~\ref{lem:connectedInduced} also applies to any two disjoint $C_4$s, each inducing a $C_4$, diamond or $K_4$ if each of those cycles has disks belonging to distinct sets from $\{\XX^+,\YY^+,\XX^-,\YY^-\}$. Recall that $W_4$ denotes the wheel graph on five vertices. We get the following.

	\begin{lem}\label{lem:largerC4smallerK4}
		Let $G$ be an $\APUD(1,1)$ containing an induced $C_4$ $\mathcal{L}$ and a $K_4$ $\mathcal{S}$ disjoint from $\mathcal{L}$, and $\Sigma(G)$ be an $\APUD(1,1)$ embedding of $G$ where the disks in $\mathcal{S}$ belong to at least three distinct sets from $\{\XX^+,\YY^+,\XX^-,\YY^-\}$. Then, the following hold:
		\begin{itemize}
			\item $\mathcal{L} \cup \mathcal{S}$ accepts a partitioning into an induced $C_4$ $\mathcal{L}'$ and a $K_4$ $\mathcal{S}'$ in $\Sigma(G)$ such that at most one disk $U' \in \mathcal{L}'$ is closer to the point $(0,0)$ than a disk $U \in \mathcal{S}'$ belonging to the same set from $\{\XX^+,\YY^+,\XX^-,\YY^-\}$ that $U'$ belongs.
			\item $\mathcal{L} \cup \mathcal{S}$ contains an induced $W_4$ having a disk of $\mathcal{S}$ as its universal disk.
			\item $\mathcal{L} \cup \mathcal{S}$ contains an induced diamond $\mathcal{D}$ on disks belonging to distinct sets from $\{\XX^+,\YY^+,\XX^-,\YY^-\}$ with its unique chord incident to a disk in $\mathcal{S} \cap \mathcal{D}$.	
		\end{itemize} 
	\end{lem}
	\begin{proof}
		Let $\mathcal{S}=\{A,B,C,D\}$ and $\mathcal{L}=\{A',B',C',D'\}$. $A' \in \XX^+$, $B' \in \YY^+$, $C' \in \XX^-$, and $D'  \in \YY^-$ by Lemma~\ref{lem:4cycle}. Let the center of $U \in \mathcal{S} \cup \mathcal{L}$ be $(x_u,0)$,  $(0,y_u)$, $(-x_u,0)$, or $(0,-y_u)$ with respect to the axis it has its center on. Lemma~\ref{lem:connectedInduced} is clearly applicable here, thus, any pair of disks in $\mathcal{S} \cup \mathcal{L}$ intersect if they belong to the same set from $\{\XX^+,\YY^+,\XX^-,\YY^-\}$. 
		
		\begin{enumerate}[(i)]
			\item If the disks in $\mathcal{S}$ belong to the distinct sets from $\{\XX^+,\YY^+,\XX^-,\YY^-\}$, let $A \in \XX^+$, $B \in \YY^+$, $C \in \XX^-$ and $D \in \YY^-$ hold. Then, the following hold:
			\begin{align}
				0 \leq {x_a}+{x_c}, {y_b}+{y_d} &\leq 2\\
				0 \leq {x_a}^2+{y_b}^2, {x_a}^2+{y_d}^2, {y_b}^2+{x_c}^2, {x_c}^2+{y_d}^2 &\leq 4\\
				0 \leq {x_{a'}}^2+{y_{b'}}^2, {y_{b'}}^2+{x_{c'}}^2, {x_{c'}}^2+{y_{d'}}^2, {y_{d'}}^2+{x_{a'}}^2 &\leq 4\\
				{x_{a'}}+{x_{c'}}, {y_{b'}}+{y_{d'}} &> 2\\
				\therefore 0 < x_a+x_c \leq 2 < x_{a'}+x_{c'}  \text{ and } 0 < y_b+y_d \leq 2 &< y_{b'}+y_{d'}
			\end{align}
			
			Then, $x_{a'} > x_a$ or $x_{c'} > x_c$ holds true. Similarly, $y_{b'} > y_b$ or $y_{d'} > y_d$ holds true. Assume that $x_{a'} > x_a$ and $y_{b'} > y_b$ hold up to symmetry. 
			
			\begin{enumerate}[a)]
				\item If $x_{c'} < x_c$ and $y_{d'} < y_d$ hold \remove{true}, then both $\{A,B,C',D'\}$ and $\{A,B,C,D'\}$ form a $K_4$.
				Furthermore, if ${x_{a'}}^2 + {y_d}^2 > 4$, then $x_{a'} > x_c$ since ${x_c}^2 + {y_d}^2 \leq 4$. However, now ${x_{c}}^2 + {y_{b'}}^2 \leq 4$ since ${x_{a'}}^2 + {y_{b'}}^2 \leq 4$. Then, we know that at least one of ${x_{a'}}^2 + {y_d}^2 \leq 4$ and ${y_{b'}}^2 + {x_c}^2 \leq 4$ holds.
				
				\begin{itemize}
					\item If ${x_{a'}}^2 + {y_d}^2 \leq 4$ and ${y_{b'}}^2 + {x_c}^2 \leq 4$ hold, then  $\mathcal{L}'=\{C,D,A',B'\}$ and $\mathcal{S}'=\{A,B,C',D'\}$ is the prescribed partitioning of $\mathcal{L} \cup \mathcal{S}$. Moreover, $\mathcal{L}' \cup A$ forms an induced $W_4$ with $A$ as its universal disk.
					
					\item Otherwise, ${x_{a'}}^2 + {y_d}^2 \leq 4$ and ${y_{b'}}^2 + {x_c}^2 > 4$ hold (up to symmetry), and $\mathcal{L}'=\{A',B',C',D\}$ and  $\mathcal{S}'=\{A,B,C,D'\}$ the prescribed partitioning of $\mathcal{L} \cup \mathcal{S}$ where only $x_{c'} < x_c$. Moreover, $\mathcal{L}' \cup A$ forms an induced $W_4$ with $A$ as its universal disk.
				\end{itemize}
				
				\item Else if $x_{c'} < x_c$ and $y_{d'} > y_d$ hold \remove{true}, $\mathcal{L}'=\mathcal{L}$ and $\mathcal{S}'=\mathcal{S}$ is the prescribed partitioning of $\mathcal{L} \cup \mathcal{S}$ where only $x_{c'} < x_c$. Moreover, $\mathcal{L}' \cup A$ forms an induced $W_4$ with $A$ as its universal disk.
				
				\item Otherwise, $x_{c'} > x_c$ and $y_{d'} > y_d$ hold true, and $\mathcal{L}'=\mathcal{L}$ and $\mathcal{S}'=\mathcal{S}$ is the prescribed partitioning of $\mathcal{L} \cup \mathcal{S}$. Moreover, if $A'$ intersects $C$, then $\mathcal{L}' \cup C$ forms an induced $W_4$ with $C$ as its universal disk, and otherwise, $\{A,C,A',B',D'\}$ forms an induced $W_4$ with $A$ as its universal disk.
			\end{enumerate}
			
			\item Otherwise, the disks in $\mathcal{S}$ belong to three distinct sets from $\{\XX^+,\YY^+,\XX^-,\YY^-\}$, and let $A,C \in \XX^+$, $B \in \YY^+$ and $D \in \YY^-$, and $x_a < x_c$ hold (up to symmetry). Then, in addition to (2), (3) and (4), the following hold:
			\begin{align*}
				0 \leq {x_c}-{x_a}, {y_b}+{y_d} &\leq 2\\
				\therefore 0 < x_c-x_a \leq 2 < x_{a'}+x_{c'}  \text{ and } 0 < y_b+y_d \leq 2 &< y_{b'}+y_{d'}
			\end{align*}
			
			Also, at least one of $y_{b'} > y_b$ and $y_{d'} > y_d$ holds since otherwise, $y_{b'} + y_{d'} < y_{b} + y_{d} \leq 2$ which contradicts that $\mathcal{L}$ is an induced $C_4$.
			
			\begin{enumerate}[a)]
				\item If $x_{a'} < x_a < x_c$ holds, and also $y_{b'} > y_b$ and $y_{d'} > y_d$ hold, then $\mathcal{L}'=\mathcal{L}$ and $\mathcal{S}'=\mathcal{S}$ is the prescribed partitioning of $\mathcal{L} \cup \mathcal{S}$ where only $x_{a'} < x_a < x_c$. Moreover, if $D'$ intersects $B$, then $\mathcal{L}' \cup B$ forms an induced $W_4$ with $B$ as its universal disk, and otherwise, $\{B,D,A',C',D'\}$ forms an induced $W_4$ with $D$ as its universal disk.
				Otherwise, $y_{b'} < y_b$ and $y_{d'} > y_d$ hold (up to symmetry), and we consider the disks $B$ and $C'$.
				\begin{itemize}
					\item  If ${y_b}^2 + {x_{c'}}^2 > 4$, then $x_{c'} > x_c$ since ${y_b}^2 + {x_c}^2 \leq 4$. However, now ${x_c}^2 + {y_{d'}}^2 \leq 4$ since ${x_{c'}}^2 + {y_{d'}}^2 \leq 4$. Furthermore, ${x_c}+{x_{c'}} > 2$ since ${x_{a'}}+{x_{c'}} > 2$ and $x_{a'} < x_c$, and ${y_{b'}}^2 + {x_c}^2 \leq 4$ since ${y_{b}}^2 + {x_c}^2 \leq 4$ and $y_{b'} < y_b$. Then, $\mathcal{L}'=\{C,B',C',D'\}$ and $\mathcal{S}'=\{A,B,D,A'\}$ is the prescribed partitioning of $\mathcal{L} \cup \mathcal{S}$ where only $y_{b'} < y_b$. Moreover, $\mathcal{L}' \cup D$ forms an induced $W_4$ with $D$ as its universal disk.
					
					\item  Otherwise, ${y_b}^2 + {x_{c'}}^2 \leq 4$. If ${x_c}^2 + {y_{d'}}^2 \leq 4$, then $\mathcal{L}'=\{B,C,C',D'\}$ and $\mathcal{S}'=\{A,D,A',B'\}$ is the prescribed partitioning of $\mathcal{L} \cup \mathcal{S}$. Moreover, $\mathcal{L}' \cup D$ forms an induced $W_4$ with $D$ as its universal disk. Otherwise, ${x_c}^2 + {y_{d'}}^2 > 4$, and $x_c > x_{c'}$ since ${x_{c'}}^2+ {y_{d'}}^2 \leq 4$. However, now ${y_{d'}} > y_b$ since ${y_b}^2 + {x_c}^2 \leq 4$ but ${x_c}^2 + {y_{d'}}^2 > 4$. Then, $\mathcal{L}'=\{B,A',C',D'\}$ and $\mathcal{S}'=\{A,C,D,B'\}$ is the prescribed partitioning of $\mathcal{L} \cup \mathcal{S}$ where only $x_{a'} < x_a < x_c$. Moreover, $\mathcal{L}' \cup D$ forms an induced $W_4$ with $D$ as its universal disk.
				\end{itemize}
				
				\item Else if $x_a < x_{a'} < x_c$ holds, and also $y_{b'} > y_b$ and $y_{d'} > y_d$ hold, then $\mathcal{L}'=\mathcal{L}$ and $\mathcal{S}'=\mathcal{S}$ is the prescribed partitioning of $\mathcal{L} \cup \mathcal{S}$ where only $x_{a'} < x_c$. Moreover, if $D'$ intersects $B$, then $\mathcal{L}' \cup C$ forms an induced $W_4$ with $B$ as its universal disk, and otherwise, $\{B,D,A',C',D'\}$ forms an induced $W_4$ with $D$ as its universal disk. Otherwise, $y_{b'} < y_b$ and $y_{d'} > y_d$ hold (up to symmetry), and we consider the disks $C$ and $D'$. 
				\begin{itemize}
					\item If ${x_c}^2 + {y_{d'}}^2 > 4$, then $x_c > x_{c'}$ since ${x_{c'}}^2+{y_{d'}}^2 \leq 4$. However, now ${y_b}^2+{x_{c'}}^2 \leq 4$ since ${y_b}^2+{x_{c}}^2 \leq 4$. Furthermore, $y_b+y_{d'}>y_{b'}+y_{d'}>2$ since $y_b>y_{b'}$, and ${x_{a'}}^2+{y_b}^2 \leq 4$ since $x_{a'}<x_c$ and ${y_b}^2+{x_c}^2 \leq 4$. Then, $\mathcal{L}'=\{B,A',C'D'\}$ and $\mathcal{S}'=\{A,C,D,B'\}$ is the prescribed partitioning of $\mathcal{L} \cup \mathcal{S}$ where only $x_{a'} < x_c$. Moreover, $\mathcal{L}' \cup D$ forms an induced $W_4$ with $D$ as its universal disk.
					
					\item Otherwise, ${x_c}^2 + {y_{d'}}^2 \leq 4$. Furthermore, $x_c+ x_{c'} > 2$ since $x_{a'}+ x_{c'} > 2$, and ${y_{b'}}^2+{x_c}^2 \leq 4$ since ${y_{b}}^2+{x_c}^2 \leq 4$ and $y_{b'} < y_b$. Then, $\mathcal{L}'=\{C,B',C',D'\}$ and $\mathcal{S}'=\{A,B,D,A'\}$ is the prescribed partitioning of $\mathcal{L} \cup \mathcal{S}$ where only $x_{b'} < x_b$. Moreover, $\mathcal{L}' \cup D$ forms an induced $W_4$ with $D$ as its universal disk.
				\end{itemize}
				
				\item Otherwise, $x_a < x_c < x_{a'}$ holds, and also $y_{b'} > y_b$ and $y_{d'} > y_d$ hold, then $\mathcal{L}'=\mathcal{L}$ and $\mathcal{S}'=\mathcal{S}$ is the prescribed partitioning of $\mathcal{L} \cup \mathcal{S}$. Moreover, if $B'$ intersects $D$, then $\mathcal{L}' \cup D$ forms an induced $W_4$ with $D$ as its universal disk, and otherwise, $\{B,D,A',B',C'\}$ forms an induced $W_4$ with $B$ as its universal disk. Otherwise, $y_{b'} < y_b$ and $y_{d'} > y_d$ hold (up to symmetry), and $\mathcal{L}'=\mathcal{L}$ and $\mathcal{S}'=\mathcal{S}$ is the prescribed partitioning of $\mathcal{L} \cup \mathcal{S}$ where only $y_{b'} < y_b$. Moreover, $\mathcal{L}' \cup D$ forms an induced $W_4$ with $D$ as its universal disk.	
			\end{enumerate}
		\end{enumerate}
		
		Finally, since a diamond is an induced subgraph of a $W_4$, the universal vertex of an induced $W_4$ is adjacent to all other vertices in a $W_4$, an induced $W_4$ without its universal vertex forms an induced $C_4$, all universal disks mentioned above are in $\mathcal{S}$, and by Lemma~\ref{lem:4cycle}, the last claim holds.\qed
	\end{proof}
	
	\begin{figure} [h!]
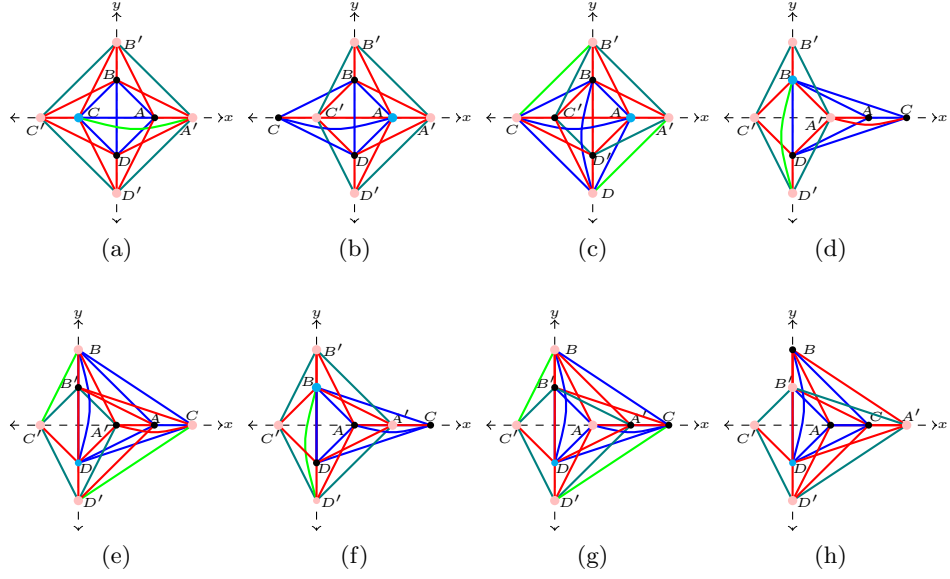

		\centering
		
		\begin{subfigure}[t]{0.23\linewidth}
			\centering

			\label{fig:ii.c}
		\end{subfigure}

		\caption{The possible cases in the setting of Lemma~\ref{lem:largerC4smallerK4}.}
		\label{fig:lem25}
	\end{figure}

	Figure~\ref{fig:lem25} shows the possible cases investigated in the proof of Lemma~\ref{lem:largerC4smallerK4} where the nodes correspond to disk centers, the blue edges exist in the induced $K_4$, the teal edges exist in the induced $C_4$, the orange edges correspond to edges appearing both in the induced $K_4$ and in the induced $C_4$, the red edges must exist due to triangle inequalities by the teal edges, the green edges illustrate the additional edges due to the prescribed subcases, and the pink and cyan nodes correspond to disks of induced $W_4$ with cyan as the universal disk.

	\remove{before starting general section, say that there is at least one induced C4 since if chordal use algo fro previous section and C4 is around (0,0)}

	\begin{clm}\label{clm:NONDISJOINTinducedK4diamond}
		Lemma~\ref{lem:largerC4smallerK4} also holds when the sets $\mathcal{L}$ and $\mathcal{S}$ have at most two disks in common. Moreover, they can not have more than two disks in common.
	\end{clm}
	
		\begin{proof}
		First of all, if $\mathcal{L}$ and $\mathcal{S}$ have three disks in common, then either $\mathcal{L}$ is not an induced $C_4$ since it has a chord, or $\mathcal{S}$ is not a $K_4$ since it has a missing chord. Thus, they have at most two disks in common. The partitioning described in Lemma~\ref{lem:largerC4smallerK4} clearly exists since the newly considered disks will be at the same distance to the point $(0,0)$. 
		
		Assume that they have two disks in common. If $A'=A$ and $B'=B$ (up to symmetry to the cases when $A'=A$ and $D'=D$, $B'=B$ and $C'=C$, $C'$ = $C$ and $D'=D$), then $x_c < x_{c'}$ since otherwise, $A$ and $C$ do not intersect, which means that $\{C,A',B',C',D'\}$ forms an induced $W_4$ with $C$ as its universal disk. Otherwise, $A'=A$ and $C'=C$ (up to symmetry to the case when $B'=B$ and $D'=D$) both of which contradict that $\mathcal{L}$ is an induced $C_4$ and $\mathcal{S}$ is a $K_4$.  
		
		Assume that they have one disk in common, say $A'=A$ up to symmetry. Then, $x_c < x_{c'}$ since otherwise, $A$ and $C$ do not intersect, which means that $\{C,A',B',C',D'\}$ forms an induced $W_4$ with $C$ as its universal disk. 
		
		The mentioned induced diamond exists by the proof of Lemma~\ref{lem:largerC4smallerK4}. \qed
	\end{proof}

	\begin{lem}\label{lem:inducedCyclesHellyCARC}
		Let $G$ be a connected $\APUD(1,1)$ and $\mathcal{L}^*$ be the set of 
		all vertices appearing in induced $C_4$s of $G$ identified in 
		polynomial time by Remark~\ref{rem:allInducedinPoly}. If 
		$\mathcal{L}^*$ contains no induced diamond or an induced $W_4$, then 
		$\mathcal{L}^*$ forms a Helly circular-arc graph, and the addition of any 
		universal vertex to $\mathcal{L}^*$ is also a Helly circular-arc graph.
	\end{lem}
	
	\begin{proof}
		Since $\mathcal{L}^*$ contains no induced $W_4$, it contains no $K_4$ 
		prescribed as in Lemma~\ref{lem:largerC4smallerK4}. It is known that a graph 
		is a Helly circular arc graph if its maximal cliques can assigned a 
		cyclic order such that the maximal cliques containing each vertex 
		appears consecutively in this cyclic order \cite{HellyCarcLinerRecog}. 
		First of all, by Claim~\ref{clm:C4together}, the cliques of 
		$\mathcal{L}^*$ satisfy the Helly property. Therefore, we show here 
		that there exists a cyclic ordering on $\mathcal{L}^*$ resulting in a 
		circular-arc graph representation. 
		
		For each induced $C_4$ on the disks $A \in \XX^+$,$B \in 
		\YY^+$, $C \in \XX^-$ and $D \in \YY^-$, we 
		already have a fixed cyclic ordering. Thus, if there is exactly one 
		induced $C_4$ in $\mathcal{L}^*$, the lemma holds. Similarly, if there 
		exists no disjoint induced $C_4$s in $\mathcal{L}^*$, we again have a 
		cyclic ordering on such induced cycles since the fixed ordering of one 
		of them fixes the ordering of the other non-disjoint induced $C_4$s. 
		Otherwise,
		for each pair of disjoint induced $C_4$s $\mathcal{L}_1$ and 
		$\mathcal{L}_2$, each disk of $\mathcal{L}_1$ is adjacent to at least 
		one disk of $\mathcal{L}_2$ and if there is exactly one such adjacency, 
		then it is between the disks placed on the same side from 
		$\{\XX^+,\YY^+,\XX^-,\YY^-\}$ by 
		Lemma~\ref{lem:connectedInduced}. Since $G$ is an $\APUD(1,1)$ and has 
		an $\APUD(1,1)$ representation, all disks contained in distinct induced 
		$C_4$s and centered on the same side from 
		$\{\XX^+,\YY^+,\XX^-,\YY^-\}$ mutually 
		intersect at a common point, i.e. each of $\mathcal{L}^* \cap 
		\XX^+$, $\mathcal{L}^* \cap \YY^+$, $\mathcal{L}^* \cap 
		\XX^-$ and $\mathcal{L}^* \cap \YY^-$ forms a clique. 
		To obtain a circular-arc representation of $\mathcal{L}^*$, we first 
		place those four cliques on the rightmost, topmost, leftmost and 
		bottommost points of the circle regarding the centers of the disks 
		belonging to $\XX^+$, $\YY^+$, $\XX^-$ and 
		$\YY^-$, respectively.

		If $\mathcal{L}^* \cap (\XX^+ \cup \YY^+)$ forms a 
		clique, then place this clique between the topmost and the rightmost 
		points of the circle, and we are done. Otherwise, there exists at least 
		one clique in $\mathcal{L}^*$ containing disks from both 
		$\XX^+$ and $\YY^+$ since $\mathcal{L}^*$ contains at 
		least one induced $C_4$. We first prove that such cliques of 
		$\mathcal{L}^*$ can be linearly ordered between the rightmost and the 
		topmost points of the circle. Let $\{A_1, \dots, A_i\} = \mathcal{L}^* 
		\cap \XX^+$ and $\{B_1, \dots, B_j\} = \mathcal{L}^* \cap 
		\YY^+$ be the disks contained in $\mathcal{L}^*$, belonging to 
		$\XX^+$ and $\YY^+$, respectively. Since $\mathcal{L}^* 
		\cap (\XX^+ \cup \YY^+)$ does not form a clique, 
		$\{A_1, \dots, A_i, B_1, \dots, B_j\}$ does not form a clique. Identify 
		the disk $A_k$ with maximum $k \leq i$ intersecting some disk from 
		$B_1, \dots, B_j$. Let $B_l$ be such a disk with maximum $l \leq j$. 
		Then, by Lemma~\ref{lem:someMandatoryEdges} item (i) and (ii), $\{ A_1, 
		\dots, A_k, B_1, \dots, B_l \}$ forms a clique. Place this clique 
		between the topmost and the rightmost points of the circle next to the 
		clique $\mathcal{L}^* \cap \XX^+$ by prolonging the arcs 
		corresponding to $\{A_1, \dots, A_k\}$ placed on the rightmost point of 
		the circle through the topmost point, and  $\{B_1, \dots, B_l\}$ placed 
		on the topmost point of the circle through the rightmost point of the 
		circle. Then, identify the disk $A_{k'}$ with maximum $k'< k$ 
		intersecting some disk from $B_1, \dots, B_j$. Let $B_{l'}$ be such a 
		disk with maximum $l' > l$. Again by Lemma~\ref{lem:someMandatoryEdges} 
		item (i) and (ii), $A_k$ intersects $A_{k'}$ and $A_{k'}$ intersects 
		$B_l$. Then, $\{ A_1, \dots, A_{k'}, B_1, \dots, B_l, B_{l'} \}$ forms 
		a clique. Place this clique between the clique $\{ A_1, \dots, A_k, 
		B_1, \dots, B_l \}$ and the topmost point of the circle where the 
		clique $\mathcal{L}^* \cap \YY^+$ is placed by prolonging the 
		arcs corresponding to $\{A_1, \dots, A_{k'}\}$ placed on $\{ A_1, 
		\dots, A_k, B_1, \dots, B_l \}$ through the topmost point, and  $\{B_1, 
		\dots, B_{l'}\}$ placed on the topmost point of the circle through the 
		rightmost point of the circle. This way, we order all maximal cliques 
		appearing in $\mathcal{L}^* \cap (\XX^+ \cup \YY^+)$, 
		and analogously all maximal cliques appearing in $\mathcal{L}^* \cap 
		(\XX^+ \cup \YY^-)$, $\mathcal{L}^* \cap (\XX^- 
		\cup \YY^+)$ and $\mathcal{L}^* \cap (\XX^- \cup 
		\YY^-)$.
		
		By Lemma~\ref{lem:largerC4smallerK4}, if $G$ contains a $K_4$ on four disks 
		belonging to distinct sides from 
		$\{\XX^+,\YY^+,\XX^-,\YY^-\}$, then $G$ 
		also contains an induced $W_4$ (thus, an induced $C_4$). However, since $\mathcal{L}^*$ contains 
		no induced $W_4$, it contains no such $K_4$. Thus, there is no clique 
		containing disks from four sides of 
		$\{\XX^+,\YY^+,\XX^-,\YY^-\}$. Also, 
		$\mathcal{L}^* \cap (\XX^+ \cup 
		\YY^- \cup \XX^-)$ contains no clique having at least 
		one disk from each of $\XX^+$, $\YY^-$ and 
		$\XX^-$ (up to symmetry) since if some $A_i$ intersects some 
		$C_k$ such that $x_{a_i} \geq x_{c_k}$, then any induced $C_4$ of 
		$\mathcal{L}^*$ containing $A_i$ forms an induced $W_4$ with $C_k$ as 
		its universal vertex since $C_k$ intersects all $B_j$ and $D_l$ 
		intersected by $A_i$ by Lemma~\ref{lem:someMandatoryEdges} 
		item (iii). Finally, the addition of a universal vertex to 
		$\mathcal{L}^*$ is also a Helly circular-arc graph since it can be 
		represented with an arc covering the whole circle. Thus, the lemma 
		holds.\qed
	\end{proof}

	Since the class of Helly circular-arc graphs are hereditary~\cite{HellyCarcLinerRecog}, the deletion 
	of any vertex from $\mathcal{L}^*$ results in a Helly circular-arc graph if 
	$\mathcal{L}^*$ is a Helly circular-arc graph. An \emph{$H$-graph} is the intersection graph on some fixed graph $H$ where each vertex is represented with a connected subgraph of $H$. We also obtain the following.
	
	\begin{lem}\label{lem:inducedCyclesH-GRAPH}
		Let $G$ be a connected $\APUD(1,1)$ and $\mathcal{L}^*$ be the set of 
		all vertices appearing in induced $C_4$s of $G$ identified in 
		polynomial time by Remark~\ref{rem:allInducedinPoly}. Then, 
		$\mathcal{L}^*$ forms an $H$-graph where $H$ consists of one cycle and 
		$4$ rays attached to it. Moreover, in an $H$-representation, the 
		induced subgraph of $G$ on the vertices placed on the cycle forms a 
		Helly circular-arc graph \add{if $\mathcal{L}^*$ contains no induced $W_4$}.
	\end{lem}
	
	\begin{proof}
		If $\mathcal{L}^*$ contains no induced diamond nor  induced $W_4$, it 
		follows from Lemma~\ref{lem:inducedCyclesHellyCARC}. Thus, we assume 
		that $\mathcal{L}^*$ contains an induced diamond or an induced $W_4$. 
		Let the rays of $H$ be labeled as $X^+_H,Y^+_H,X^-_H,Y^-_H$ respecting 
		the axes. By Lemma~\ref{lem:middleClique}, if there is a set 
		$\mathcal{F}$ of disks that are centered in $\Gamma_{L}$ for an induced 
		$C_4$ $L$, $\mathcal{F}$ forms a clique that satisfies the Helly 
		property. Since the cliques formed by disks placed on $\XX$ and 
		$\YY$ can be placed around a circle by 
		Lemma~\ref{lem:inducedCyclesHellyCARC}, here we consider a disk $A$ 
		with its center on $\XX^+$ intersecting a disk $C$ with its 
		center on $\XX^-$ (and, a disk $B$ with its center on 
		$\YY^+$ intersecting a disk $D$ with its center on 
		$\YY^-$ analogously follows). Then, $A$ and $C$, and 
		potentially some other disks, intersect mutually since all their 
		centers are inside $\Gamma_{L}$ for an induced $C_4$ $L$. Then, all 
		these disks can cover the whole circle except some disks placed on 
		$\YY^+$ and $\YY^-$ which do not intersect them. 
		Consider the disk $B$ with its center on $\YY^+$ both $A$ and 
		$C$ intersect which has the greatest $y$-coordinate. Now, $A$ and $C$ 
		intersect all other disks $B'$ with their centers on $\YY^+$ 
		having a smaller $y$-coordinates. Then, all disks $B''$ with their 
		centers on $\YY^+$ having a greater $y$-coordinate than $B$ can 
		be placed on the ray $Y^+_H$ when $B$, together with all disks $A'$ 
		with their centers on $\XX^+$ intersecting $B$, thus having a 
		smaller $x$-coordinate than $A$, and all disks $C'$ with their centers 
		on $\XX^-$ intersecting $B$, thus having a greater 
		$x$-coordinate than $C$ are prolonged true $Y^+_H$. This holds true 
		also considering the disk $D$ with its center on $\YY^-$ both 
		$A$ and $C$ intersect which has the greatest $y$-coordinate. Thus, the 
		lemma holds.\qed
	\end{proof}

	
	\begin{corl}\label{corl:W4UNI}
		Let $\{A,B,C,D,U\}$ be an induced $W_4$ in an $\APUD(1,1)$ such that $U$ is the universal disk and $\{A,B,C,D\}$ is the induced $C_4$. Let $(x_a,0)$, $(0,y_b)$, $(-x_c,0)$, $(0,-y_d)$, $(x_u,0)$ denote the centers of $\{A,B,C,D,U\}$, respectively where $x_a,y_b,x_c,y_d,x_u \in \mathbb{R}^+$. Then, $0 < x_a-x_u \leq 2$ and $0 < x_c+x_u \leq 2$, thus $2 < x_a+x_c \leq 4$.
	\end{corl}
	
	\begin{clm}\label{clm:inducedW4sinCenter}
		If $G$ is an $\APUD(1,1)$ containing two disjoint induced $W_4$s, then the universal disks of those $W_4$s are adjacent to each other if they belong to the same set from $\{\XX^+,\YY^+,\XX^-,\YY^-\}$. 
	\end{clm}
	\begin{proof}
		Let $\mathcal{L}_1$ and $\mathcal{L}_2$ denote the induced $C_4$s appearing in two disjoint $W_4$s. By Lemma~\ref{lem:4cycle}, the centers of the four disks in $\mathcal{L}_1$ (and also $\mathcal{L}_2$) are on $\XX^+$, $\YY^+$, $\XX^-$ and $\YY^-$, respectively. By Lemma~\ref{lem:connectedInduced},
		each disk of $\mathcal{L}_1$ is adjacent to at least one disk of $\mathcal{L}_2$, and such adjacencies are at least between the disks with their centers on the same side from $\{\XX^+,\YY^+,\XX^-,\YY^-\}$. Let those disks be $A_1,B_1,C_1,D_1 \in \mathcal{L}_1$ and $A_2,B_2,C_2,D_2 \in \mathcal{L}_2$ with their centers on $\{\XX^+,\YY^+,\XX^-,\YY^-\}$, respectively, such that $V_1$ is adjacent to $V_2$ for $V \in \{A,B,C,D\}$, and let $U_1$ and $U_2$ be those universal disks of such $W_4$s. Assume that the center of $U_1$ is on $\XX^+$, i.e. it is $(x_{u_1},0)$ for $x_{u_1} \in \mathbb{R}^+$. Let the center of $U_2$ be also on $\XX^+$. If $x_{u_2} < x_{u_1}$, then the claim holds since $U_1$ and $U_2$ must intersect each other so that $U_1$ can intersect $C_1$ by Corollary~\ref{corl:W4UNI}.
		Otherwise, $x_{u_2} > x_{u_1}$ and the claim holds since $U_1$ and $U_2$ must intersect each other so that $U_2$ can intersect $C_2$ by Corollary~\ref{corl:W4UNI}.\qed
	\end{proof}
	
	\begin{clm}\label{clm:inducedW4sinCenterClique}
		If $G$ is an $\APUD(1,1)$, all universal disks forming an induced $W_4$ with the same induced $C_4$ form a clique.
	\end{clm}
	
	\begin{proof}
		Let $\mathcal{L}=\{A,B,C,D\}$ be an induced $C_4$ such that the addition of any disk $U \in \mathcal{F} \subseteq V(G) \setminus \mathcal{L}$ results in an induced $W_4$. Note that $U$ must be the universal disk of the formed $W_4$. Let the centers of $\{A,B,C,D,U\}$ be as in Corollary~\ref{corl:W4UNI}. Assume another disk $V \neq U \in \mathcal{F}$. Let $x_v$ and $y_v$ be two positive real numbers. If the center of $V$ is $(x_v,0)$, then the claim follows trivially by the similar arguments used to prove Claim~\ref{clm:inducedW4sinCenter}. Else if the center of $V$ is $(-x_v,0)$, then the claim holds since $U$ cannot intersect $C$ without intersecting $V$ as both of them are centered between $A$ and $C$ by Corollary~\ref{corl:W4UNI}. Else, the center of $V$ is $(0,y_v)$ (and the case $(0,-y_v)$ is analogous). Now, by Corollary~\ref{corl:W4UNI}, $0 < y_b-y_v \leq 2$ and $0 < y_d+y_v \leq 2$, and $U$ cannot intersect $B$ without intersection $V$, thus the claim holds.\qed
	\end{proof}


	\begin{clm}\label{clm:inducedC4diamond}
		If $G$ is an $\APUD(1,1)$ which contains at least one induced $C_4$ $\mathcal{L}$, and an induced diamond $\mathcal{S}$, disjoint from $\mathcal{L}$, formed by disks belonging to at least three distinct sets from $\{\XX^+,\YY^+,\XX^-,\YY^-\}$, then $G$ contains an induced $W_4$ formed by five of those disks.
	\end{clm}

	\begin{proof}
		Let $\mathcal{L}=\{A',B',C',D'\}$ and $\mathcal{S}=\{A,B,C,D\}$ s.t. $A \neq A' \in \XX^+$, $B \neq B' \in \YY^+$, $C \neq C' \in \XX^-$, and $D \neq D' \in \YY^-$, $U$ and $U'$ intersect for each $U \in \{A,B,C,D\}$ clearly by Lemma~\ref{lem:connectedInduced}, and $B$ and $D$ do not intersect, i.e. $v_bv_d$ is the missing edge of $S$. Then, at least one of $x_a < x_{a'}$ or $x_c < x_{c'}$ holds since otherwise, $A'$ and $C'$ are also adjacent. Assume that $x_a < x_{a'}$. Now, at least one of the following is an induced $W_4$ having $A$ or $C$ as its universal disk listed first in the corresponding induced $W_4$s.
		
		\begin{itemize}
			\item \textbf{$\boldsymbol{y_b <y_{b'}}$, $\boldsymbol{x_c < x_{c'}}$, $\boldsymbol{y_d < y_{d'}}$:} If $x_{a'}+x_c >2$, $\{A,B,C,D,A'\}$, and otherwise, $\{C,A',B',C',D'\}$.
			
			\item \textbf{$\boldsymbol{y_b <y_{b'}}$, $\boldsymbol{x_c > x_{c'}}$, $\boldsymbol{y_d < y_{d'}}$:} $\{A,A',B',C',D'\}$.
			
			\item \textbf{$\boldsymbol{y_b >y_{b'}}$, $\boldsymbol{x_c < x_{c'}}$, $\boldsymbol{y_d < y_{d'}}$:} If $x_{a'}+x_c >2$, $\{A,C,A',B',D'\}$, and otherwise, $\{C,A',B',C',D'\}$. 
			
			\item \textbf{$\boldsymbol{y_b <y_{b'}}$, $\boldsymbol{x_c < x_{c'}}$, $\boldsymbol{y_d > y_{d'}}$:} Analogous to the previous, up to symmetry.
			
			\item \textbf{$\boldsymbol{y_b >y_{b'}}$, $\boldsymbol{x_c < x_{c'}}$, $\boldsymbol{y_d > y_{d'}}$:} If $x_{a'}+x_c >2$, $\{A,C,A',B',D'\}$, and otherwise, $\{C,A',B',C',D'\}$.
			
			\item \textbf{$\boldsymbol{y_b >y_{b'}}$, $\boldsymbol{x_c > x_{c'}}$, $\boldsymbol{y_d < y_{d'}}$:} $\{A,A',B',C',D'\}$.
			
			\item \textbf{$\boldsymbol{y_b <y_{b'}}$, $\boldsymbol{x_c > x_{c'}}$, $\boldsymbol{y_d > y_{d'}}$:} Analogous to the previous, up to symmetry. 
			
			\item \textbf{$\boldsymbol{y_b >y_{b'}}$, $\boldsymbol{x_c > x_{c'}}$, $\boldsymbol{y_d > y_{d'}}$:} $\{A,A',B',C',D'\}$.

		\end{itemize}
		
		Thus, the claim holds. Note also that, if the universal vertex of such a $W_4$ is removed, then the mentioned diamond disappears.\qed
	\end{proof}

	\begin{figure} [h!]
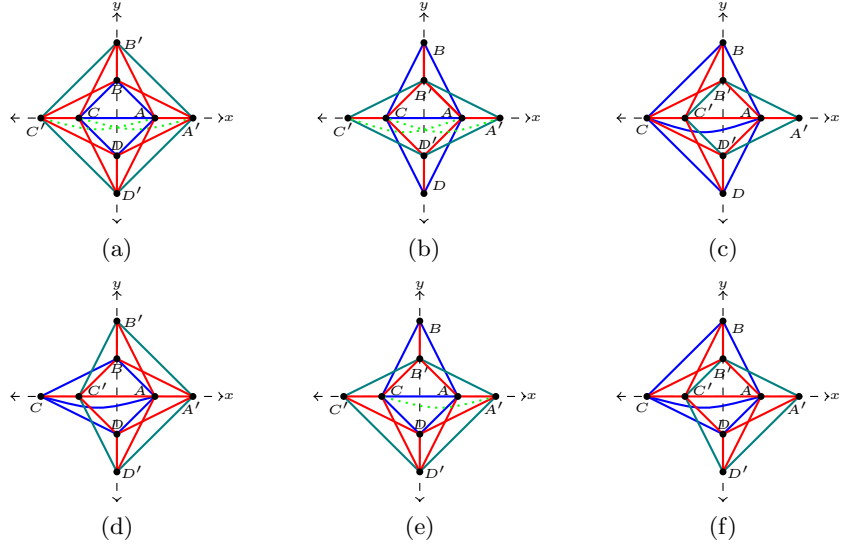

		\captionsetup[subfigure]{position=b}	
		\centering
		
		\begin{subfigure}[t]{0.30\linewidth}
			\centering

			\label{fig:case6C4diamondW4}
		\end{subfigure}
		
		\caption{The possible cases in the setting of Claim~\ref{clm:inducedC4diamond}. }
		\label{fig:chordalS4x}
	\end{figure}

	Figure~\ref{fig:chordalS4x} shows the possible cases investigated in the proof of Claim~\ref{clm:inducedC4diamond} where the black nodes correspond to disk centers, the blue edges exist in the induced diamond, the teal edges exist in the induced $C_4$,  the red edges must exist due to triangle inequalities by the teal edges, and the dotted green edges illustrate the possible additional edges.

	\begin{clm}\label{clm:NONDISJOINTinducedC4diamond}
		Claim~\ref{clm:inducedC4diamond} also holds when the sets $\mathcal{L}$ and $\mathcal{S}$ have at most three disks in common. Moreover, they can not have more than three disks in common.
	\end{clm}
	
	\begin{proof}
		First of all, if $\mathcal{L}$ and $\mathcal{S}$ have four disks in common, then $\mathcal{L}=\mathcal{S}$ which can not be both an induced $C_4$ and $K_4$. Thus, they have at most three disks in common. Since $v_bv_d$ is the missing edge of $\mathcal{S}$, they can not have both $A$ and $C$ in common since otherwise $v_av_c$ is a chord in $\mathcal{L}$ which contradicts that $\mathcal{L}$ is an induced $C_4$. Moreover, if $A'=A$, then $C' \neq C$ and $x_c < x_{c'}$ since otherwise, $A$ and $C$ do not intersect, which means that $\{C,A',B',C',D'\}$ forms an induced $W_4$ with $C$ as its universal vertex. Thus, we only consider the following cases.
		
		Assume that they have three disks in common.  Then, up to symmetry, $A'=A$, $B'=B$, $D'=D$, and $\{C,A',B',C',D'\}$ is the induced $W_4$ mentioned above. 
		
		Assume that they have two disks in common. Then, up to symmetry, $A'=A$ and $B'=B$, or $B'=B$ and $D'=D$. If $A'=A$ and $B'=B$,  $\{C,A',B',C',D'\}$ is the induced $W_4$ mentioned above. Otherwise, $B'=B$ and $D'=D$, and at least one of $x_a < x_{a'}$ or $x_c < x_{c'}$, say $x_a < x_{a'}$, holds since $A$ and $C$ intersect, thus $\{A,A',B'C',D'\}$ forms an induced $W_4$ with $A$ as its universal vertex.
		
		Assume that they have one disk in common. Then, up to symmetry, $A'=A$ or $B'=B$. If $A'=A$,  $\{C,A',B',C',D'\}$ is the induced $W_4$ mentioned above. Otherwise, $B'=B$ and at least one of $x_a < x_{a'}$ or $x_c < x_{c'}$, say $x_a < x_{a'}$, holds since $A$ and $C$ intersect, thus $\{A,A',B'C',D'\}$ forms an induced $W_4$ with $A$ as its universal vertex.\qed 
	\end{proof}

	\begin{corl}\label{corl:inducedW4C4}
		By Claim~\ref{clm:inducedC4diamond}, if $G$ is an $\APUD(1,1)$ which contains at least three disjoint induced $C_4$s, at least one induced diamond in those $C_4$s forms an induced $W_4$ with the addition of a disk from those $C_4$s.
	\end{corl}

	\begin{lem}\label{lem:NOW4NOdiamond}
		If $G$ is a non-chordal $\APUD(1,1)$ which contains no induced $W_4$, then $G$ contains no induced diamond $\{A,B,C,D\}$ where $A \in \mathcal{X}^+, B \in \mathcal{Y}^+, C \in \mathcal{X}^-, D \in \mathcal{Y}^-$.
	\end{lem}
	\begin{proof}
		By Claim~\ref{clm:inducedC4diamond} and Claim~\ref{clm:NONDISJOINTinducedC4diamond}, if $G$ contains an induced $C_4$ which is on four disks belonging to distinct sets from $\{\XX^+,\YY^+,\XX^-,\YY^-\}$ by Lemma~\ref{lem:4cycle} and an induced diamond on disks belonging to at least three distinct sets from $\{\XX^+,\YY^+,\XX^-,\YY^-\}$, it contains an induced $W_4$ which is on disks belonging to distinct sets from $\{\XX^+,\YY^+,\XX^-,\YY^-\}$ since an induced $W_4$ contains an induced $C_4$.
		Since $G$ is not chordal, it contains an induced $C_4$. Then, $G$ does not contain an induced diamond $\{A,B,C,D\}$ where $A \in \mathcal{X}^+, B \in \mathcal{Y}^+, C \in \mathcal{X}^-, D \in \mathcal{Y}^-$ since otherwise, $G$ must contain an induced $W_4$.\qed
	\end{proof}

	On the other hand, if a graph does not contain an induced diamond, it does not contain an induced $W_4$ since a diamond is an induced subgraph of a $W_4$. Thus, we get the following.
	
	\begin{thm}\label{thm:nonCHOR-W4-iff-diamond}
		A non-chordal graph $G \in \APUD(1,1)$ contains an induced $W_4$ if and only if it contains an induced diamond on disks belonging to distinct sets from $\{\XX^+,\YY^+,\XX^-,\YY^-\}$.
	\end{thm}
	
	\begin{lem}\label{thm:nonCHOR-noW4/diamond-noK4}
		If $G$ is a non-chordal $\APUD(1,1)$ which contains no induced $W_4$ or an induced diamond, then it contains no $K_4$ on disks belonging to at least three distinct sets from $\{\XX^+,\YY^+,\XX^-,\YY^-\}$.
	\end{lem}
	
		\begin{proof}
		By Theorem~\ref{thm:nonCHOR-W4-iff-diamond}, $G$ contains an induced $W_4$ if and only if it contains an induced diamond on disks belonging to at least three distinct sets from $\{\XX^+,\YY^+,\XX^-,\YY^-\}$. Thus, we prove our lemma only considering the fact that $G$ contains no induced $W_4$. By Claim~\ref{lem:largerC4smallerK4} and Claim~\ref{clm:NONDISJOINTinducedK4diamond}, if $G$ contains an induced $C_4$ which is on four disks belonging to distinct sets from $\{\XX^+,\YY^+,\XX^-,\YY^-\}$ by Lemma~\ref{lem:4cycle} and a $K_4$ on disks belonging to at least three distinct sets from $\{\XX^+,\YY^+,\XX^-,\YY^-\}$, it contains an induced $W_4$ which is on disks belonging to distinct sets from $\{\XX^+,\YY^+,\XX^-,\YY^-\}$ since an induced $W_4$ contains an induced $C_4$.
		Since $G$ is not chordal, it contains an induced $C_4$. Then, $G$ does not contain a $K_4$ on disks belonging to at least three distinct sets from $\{\XX^+,\YY^+,\XX^-,\YY^-\}$.\qed
	\end{proof}

	\subsection*{The recognition algorithm for $\boldsymbol{\APUD(1,1)}$}

	Given a connected graph $G$, we decide whether $G$ is an $\APUD(1,1)$ as follows: 
	
	\begin{enumerate}
		\item If $G$ is chordal, then use Algorithm~\ref{alg:chordalapud} and return its result. Here, we also emphasize that an $\APUD(1,1)$ which can be realized on $x$-axis and only one side of the $y$-axis, say $y+$, is already a chordal graph by Lemma~\ref{lem:4cycle}.

		\item Let $\mathcal{L}$ denote the set of all vertices appearing in induced \replace{$4$-cycles}{$C_4$s} of the input graph $G$, identified in polynomial time by Remark~\ref{rem:allInducedinPoly}.

		\item If $G - \mathcal{L}$ is not an $S_4$-graph, return that $G$ is not an $\APUD(1,1)$ by Lemma~\ref{lem:Sd}.

		\item If $G[\mathcal{L}]$ contains no induced $W_4$, and it is not a Helly circular-arc graph, return that $G \not\in \APUD(1,1)$ by Lemma~\ref{lem:inducedCyclesHellyCARC} and Lemma~\ref{lem:NOW4NOdiamond}.

		\item \replace{Identify all induced $W_4$s of $G$ each having all its $C_4$ vertices from $\mathcal{L}$. Remove all their universal vertices from $\mathcal{L}$ obtaining $\mathcal{L}'$. If $G[\mathcal{L}']$ is not a Helly circular-arc graph, return that $G \not\in \APUD(1,1)$ by Lemma~\ref{lem:inducedCyclesH-GRAPH}.}{Identify all induced $W_4$s of $G$ each having all its $C_4$ vertices from $\mathcal{L}$. Let $\mathcal{U}$ be the set of all their universal vertices and  $\mathcal{L}'= \mathcal{L} \setminus \mathcal{U}$. If $G[\mathcal{L}']$ is not a Helly circular-arc graph, return that $G \not\in \APUD(1,1)$ by Lemma~\ref{lem:inducedCyclesH-GRAPH}.}

		\item \replace{Let $\XX$ denote the set of connected components of $G - \mathcal{L}'$. If $\vert \XX \vert > 5$, return that $G$ is not an $\APUD(1,1)$. (Lemma~\ref{lem:middleClique}, Claim~\ref{clm:inducedW4sinCenterClique}, Claim~\ref{clm:inducedK4diamond}, Claim~\ref{clm:inducedC4diamond}).}{Let $\Delta$ denote the set of connected components of $G - \mathcal{L}$. By the characterization item \textbf{A2} in Corollary~\ref{corl:combine}, if $G \in \APUD(1,1)$, it must hold that $\Delta = (\XX^+ \cup \YY^+ \cup \XX^- \cup \YY^-) \setminus \mathcal{L}$ and each component in $\Delta$ must be a unit interval graph. Therefore, if $\vert \Delta \vert > 4$ or a component in $\Delta$ is not a unit interval graph, return that $G \not\in \APUD(1,1)$.}
		
		\add{\item By Theorem~\ref{theo:linearlyMany}, if $G \in \APUD(1,1)$, it must have $\mathcal{O}(n^3)$ maximal cliques. Using the algorithm of \cite{DBLP:conf/swat/MakinoU04}, start listing the maximal cliques of $G$, and if the algorithm returns an $n^4$th maximal clique, terminate the algorithm, and return that $G \not\in \APUD(1,1)$.} 
			
		\add{\item Let $\mathcal{Q}$ denote the maximal cliques of $G[\mathcal{L}']$. By Lemma~\ref{lem:inducedCyclesHellyCARC}, $G[\mathcal{L}']$ is a Helly circular-arc graph, thus $\mathcal{Q}$ can be computed in polynomial time \cite{DBLP:conf/swat/MakinoU04}. }
		
		\add{\item By Lemma~\ref{lem:connectedInduced}, two disks appearing in induced $C_4$s intersect if they belong to the same set from $\{\XX^+, \YY^+, \XX^-, \YY^-\}$. Therefore, $\mathcal{Q}$ must contain four maximal cliques $C_1,C_2,C_3,C_4$ such that $\bigcup_{i=1}^{4} C_i = \mathcal{L}'$. If such maximal cliques do not exist, return that $G \not\in \APUD(1,1)$.}

		\add{\item Let $\mathcal{R}$ denote the maximal cliques of $G[\mathcal{U}]$. Step 7 witnesses that $G$, therefore its induced subgraph $G[\mathcal{U}]$, has at most $\mathcal{O}(n^3)$ maximal cliques. Thus, $\mathcal{R}$ can be computed in polynomial time \cite{DBLP:conf/swat/MakinoU04}. }
		
		\add{\item By Claim~\ref{clm:inducedW4sinCenter}, universal disks of induced $W_4$s belonging to the same set from $\{\XX^+, \YY^+, \XX^-, \YY^-\}$ form a clique. By Claim~\ref{clm:inducedW4sinCenterClique}, all universal disks forming induced $W_4$s with the same induced $C_4$ form a clique. By Lemma~\ref{lem:atLeastThree}, the non-zero coordinate of each disk appearing in an induced $C_4$ of an $\APUD(1,1)$ is at most $2$, and by Lemma~\ref{lem:largerC4smallerK4}, this also holds for universal disks. Therefore, $\mathcal{R}$ must contain four maximal cliques $S_1, S_2, S_3, S_4$ such that $\bigcup_{i=1}^{4} S_i = \mathcal{R}$ and $G[C_i \cup S_i]$ is a clique for $i \in \{1,2,3,4\}$. If such maximal cliques do not exist, return that $G \not\in \APUD(1,1)$.} 
		 
		\add{\item Since there is a polynomial number of maximal cliques in both $\mathcal{L}'$ and $\mathcal{R}$, there is a polynomial number of ordered 4-tuples of maximal cliques $(C_1,C_2,C_3,C_4)$ and $(S_1, S_2, S_3, S_4)$. }
		
		\add{\item Let $\Delta_1,\Delta_2,\Delta_3,\Delta_4$ be the (possible empty) connected components in $\Delta$. There is a constant number of ordered 4-tuples $(\Delta_1,\Delta_2,\Delta_3,\Delta_4)$ among those sets. }
		
		\add{\item Looping over all three ordered 4-tuples $(C_1,C_2,C_3,C_4)$, $(S_1, S_2, S_3, S_4)$ and $(\Delta_1,\Delta_2,\Delta_3,\Delta_4)$, we determine if $G \in \APUD(1,1)$ as follows:}
		\begin{enumerate}
			\item If, for all $i \neq j \in \{1,2,3,4\}$, $G[\Delta_i \cup C_i \cup S_i \cup S_j \cup C_j \cup \Delta_j] \in \UIG$ with its maximal cliques appearing in that order, return that $G \in \APUD(1,1)$.
		\end{enumerate}
	
		\item Return that $G \not\in \APUD(1,1)$.
	\end{enumerate}

	\section{Conclusions and future work} \label{sec:conc}
		In this paper, we studied the base case of $\APUD(k,m)$ recognition 
		which is an NP-hard problem.
		By the properties of Helly cliques and unit interval graphs, we showed that given a simple graph $G$, we can tell in polynomial time whether $G \in \APUD(1,1)$, i.e., $G$ has an embedding $\Sigma(G)$ as disks onto two perpendicular lines. Note that our algorithm does not output an embedding despite recognizing an $\APUD(1,1)$ in polynomial time since it is unknown whether the center of every disk of an $\APUD(1,1)$ can have
		coordinates that are represented using polynomially many decimals. Therefore, we give the following:
		
		\begin{open}
			Given a graph $G \in \APUD(1,1)$, can we find an $\APUD(1,1)$ embedding of $G$ in polynomial time?
		\end{open}
		
		In \cite{CSRapud}, it was shown that $\APUD(k,0)$ recognition is NP-hard when $k \geq 3$. Therefore, we would like to consider $\APUD(2,0)$ recognition as future work. A graph $G \in \APUD(2,0)$ can be embedded on two horizontal lines. Let those horizontal lines be $y=i$ and $y=j$. Observe that if $\vert i-j \vert > 2$, then we deal with a disconnected unit interval graph which can be recognized in linear time~\cite{DBLP:journals/ipl/Keil85}. Therefore, we consider $\APUD(2,0)_{d \leq 2}$ recognition problem which asks whether a graph $G$ is an $\APUD(2,0)$ where the distance between two horizontal lines witnessing that $G \in \APUD(2,0)$ is $d \leq 2$. We give the following which may be crucial to recognize an $\APUD(2,0)_{d \leq 2}$.
		
		\begin{clm}\label{clm:onlyOneAdjacent}
			Let $G$ be an $\APUD(2,0)_{d=2}$. Then, in any $\APUD(2,0)_{d=2}$ embedding of $G$, a disk $A$ with its center on $y=i$ for $i \in \{1,3\}$ intersecting a disk $B$ with its center on $y=j$ for $j \neq i$ intersects no other disk $C$ on $y=j$, and the centers of $A$ and $B$ have the same $x$-coordinate.
		\end{clm}
		\begin{proof}
			Since two unit disks intersect if and only if the distance between their centers is at most $2$-units, the centers of $A$ and $B$ must have the same $x$-coordinate since the distance between their $y$-coordinates is exactly $2$-units. \qed
		\end{proof}
	
		Claim~\ref{clm:onlyOneAdjacent} directly gives the following.
	
		\begin{corl}\label{corl:induced}
			Let $G$ be an $\APUD(2,0)_{d=2}$. Then, any induced cycle $L$ of length at least $4$ of $G$ contains at least two disks with their centers on $y=1$ and at least two disks with their centers on $y=3$ in every $\APUD(2,0)_{d=2}$ embedding of $G$.
		\end{corl}
		
		By definition, induced cycles of length at least $4$ are chordless which means that they do not contain induced cycles other than themselves. By Claim~\ref{clm:onlyOneAdjacent}, in an $\APUD(2,0)_{d=2}$, each pair of disks one with its center on $y=1$ and the other with its center on $y=3$ having the same $x$-coordinate $x_i$ can belong to at most two induced cycles, i.e. one induced cycle on disks with $x$-coordinates less than $x_i$, and another induced cycle on disks with $x$-coordinates more than $x_i$. Thus, by Corollary~\ref{corl:induced}, there is a linear number of induced cycles of length at least $4$ on the number of disks. Therefore, we give the following.
		
		\begin{conj}
			Given a graph $G$, one can decide whether $G \in \APUD(2,0)_2$ in polynomial time. 
		\end{conj}
	
		We conclude with the following problem which we would like to consider in the future.
	
		\begin{open}
			Given a graph $G$, can we decide whether $G \in \APUD(2,0)_{<2}$ in polynomial time?
		\end{open}

	
	\small
	\bibliographystyle{abbrv}
	\bibliography{bibliography}

\end{document}